\documentclass[10pt]{article}
\usepackage{amsfonts,amsmath,latexsym}
\usepackage[T1]{fontenc}
\usepackage[dvips]{graphicx}
\usepackage{epsfig}
\usepackage{enumerate}
\usepackage{epsf}
\usepackage{amsthm}
\usepackage{color}
\usepackage{amssymb}
\usepackage{fancyhdr} 
\newtheorem{thm}{Theorem}[section]
\newtheorem{propo}[thm]{Proposition}
\newtheorem{lemme}[thm]{Lemma}
\newtheorem{corro}[thm]{Corollary}
\newtheorem{defi}[thm]{Definition}
\newtheorem{remarque}[thm]{Remark}
\newtheorem{example}[thm]{Example}
\newtheorem{Hyp}{Assumption}

\newcommand{\argmin}{\mathop{\mathrm{arg\,min}}}

\def\R{\mathbb R}

\def\C{\mathbb C}
\def\E{\mathbb E}

\def\shd{{\cal D}}
\def\shf{{\cal F}}

\def\shl{{\cal L}}

\def\halb{{\frac{1}{2}}}

\author{{\sc Stéphane GOUTTE $^*$}\footnote{Universit\'e Paris 13,
    Mathématiques LAGA, Institut Galilée, 99 Av. J.B. Clément 93430
    Villetaneuse. E-mail:{\tt
      goutte@math.univ-paris13.fr}}\thanks{Luiss Guido Carli - Libera
    Università Internazionale degli Studi Sociali Guido Carli di Roma}
  {\sc,}\ {\sc Nadia OUDJANE $^*$}\thanks{EDF R\&D 
E-mail:{\tt  
nadia.oudjane@edf.fr}}
\ {\sc and}\ {\sc Francesco RUSSO $^*$}\thanks{INRIA Rocquencourt 
and Cermics Ecole des Ponts, Projet MATHFI. Domaine de Voluceau, BP 105
F-78153 Le Chesnay Cedex (France) .  E-mail:{\tt  russo@math.univ-paris13.fr}}
}

\date{December 2nd 2009}
\title{\bf Variance Optimal Hedging for continuous time processes 
with independent increments and applications}

\newcommand{\MBFigure}[6]{
$\left. \right.$ \\
\refstepcounter{figure}
\addcontentsline{lof}{figure}{\numberline{\thefigure}{\ignorespaces #5}}
\begin{center}
\begin{minipage}{#1cm}
\centerline{\includegraphics[width=#2cm,angle=#3]{#4}}
\begin{center}
\upshape{F\textsc{ig} \normal
\end{center}
size{\thefigure}. $-$} #5
\end{center}
\label{#6}
\end{minipage}
\end{center}
$\left. \right.$ \\}

\begin{document}
\maketitle

\begin{abstract} For a large class of vanilla contingent claims, 
we establish an explicit F\"ollmer-Schweizer 
decomposition when the underlying is 
a process with independent increments (PII)  and an exponential of a PII process.
This allows to provide an efficient algorithm for solving the
mean variance hedging problem.
Applications to models derived from the electricity market are performed.
\end{abstract}


\bigskip\noindent  {\bf 2000  AMS-classification}: 60G51, 60H05, 60J25, 60J75

\medskip
 \bigskip\noindent   {\bf JEL-classification}: C02, G11, G12, G13



\newpage

\section{Introduction}

\setcounter{equation}{0}

There are basically two main approaches to define the
 \textit{mark to market} of a contingent claim: one relying on the \textit{no-arbitrage assumption} and the other related to a \textit{hedging portfolio}, those two approaches converging in the specific case of complete markets. A simple introduction to the different hedging and pricing models in incomplete markets can be found in chapter~10 of~\cite{livreTankovCont}.\\
%
%
\textit{The fundamental theorem of Asset Pricing}~\cite{delbaen94} implies that a pricing rule without arbitrage that moreover satisfies some usual conditions (linearity non anticipativity \ldots) can always be written as an expectation under a martingale measure. In general, the resulting price is not linked with a hedging strategy except in some specific cases such as complete markets. 
More precisely, it is proved~\cite{delbaen94} that the market
 completeness is equivalent to uniqueness of the equivalent martingale
 measure. Hence, when the market is not complete, there exist several equivalent martingale measures (possibly an infinity) and one has to specify a criterion to select one specific pricing measure: 
to recover some given option prices (by calibration)~\cite{goll-ruschendorf}; 
to simplify calculus and obtain a simple process under the pricing measure;
to maintain the structure of the real world dynamics; 
to minimize a \textit{distance} to the objective probability (entropy~\cite{fritelli}) \ldots 
In this framework, the difficulty is to understand in a practical way the impact of the choice of the martingale measure on the resulting prices.\\
If the resulting price is in general not connected to a hedging
strategy, yet it is possible to consider the hedging question in a
second step, optimizing the hedging strategy for the given price. In this framework, one approach consists in deriving the hedging
strategy minimizing the \textit{global quadratic hedging error} under
the pricing measure where the martingale property of the underlying highly simplifies calculations. This approach, is developed in~\cite{livreTankovCont}, in the case of exponential-Lévy models: the optimal quadratic hedge is then expressed as a solution of an integro-differential equation involving the Lévy measure. Unfortunately, minimizing the quadratic hedging error under the pricing measure has  no clear interpretation since the resulting hedging strategy can lead to huge quadratic error under the objective measure.\\ 
%
Alternatively, one can define option prices  as a by product of the hedging strategy. In the case of complete markets, any option can be replicated perfectly by a self-financed hedging portfolio continuously rebalanced, then the option \textit{hedging value} can be defined as the cost of the hedging strategy. 
When the market is not complete,  it is not possible, in general, to
hedge perfectly an option. One has to specify a risk criteria,  and
 consider the hedging strategy that minimizes the distance (in terms of the given criteria) between the pay-off of the option and the terminal value of the hedging portfolio. Then, the price of the option is related to the cost of this imperfect hedging strategy to which is added in practice another prime related to the residual risk induced by incompleteness. \\
Several criteria can be adopted. 
The aim of super-hedging is to hedge all cases. This approach yields in general prices that are too  expensive to be realistic~\cite{elkaroui-quenez}. 
Quantile hedging modifies this approach allowing for a limited probability of loss~\cite{fl99}. 
Indifference Utility pricing introduced by~\cite{hodges-neuberger} defines the price of an option to sell (resp. to buy) as the minimum initial value
 s.t. the hedging portfolio with the option sold (resp. bought) is equivalent
 (in term of utility) to the initial portfolio. 
Quadratic hedging is developed in~\cite{S94},~\cite{S95bis}.  The quadratic distance
 between the hedging portfolio and the pay-off is 
minimized. Then, contrarily to the case of utility maximization, losses and gains are treated in a symmetric manner, which yields a \textit{fair price} for both the buyer and the seller of the option. \\
%
In this paper, we follow this last approach and our developments can be used in both the \textit{no-arbitrage value} and the \textit{hedging value}  framework:     
either to derive the hedging strategy minimizing the \textit{global quadratic hedging error} under the objective measure, for a given pricing rule; 
or to derive both the price and the hedging strategy minimizing the \textit{global quadratic hedging error} under the objective measure.\\ 
We spend now some words related to the global quadratic hedging approach
 which is also 
called {\it mean-variance hedging} or {\it global risk minimization}.
Given a square integrable r.v. $H$, we say that the pair $(V_0, \varphi)$
is optimal if $(c, v) = (V_0, \varphi)$ minimizes the functional
$\E \left (H-c- \int_0^T v dS\right)^2$. The price $V_0$ represents the price of
the contingent claim $H$ and $\varphi$ is the optimal strategy. \\
Technically speaking, the global risk minimization problem,  
is based on the so-called {\it F\"ollmer-Schweizer} decomposition 
(or FS decomposition) of a square integrable random variable
 (representing the contingent claim)
with respect to an $(\shf_t)$-semimartingale $S = M + A$  
modeling the asset price. $M$ is an $(\shf_t)$-local martingale
and $A$ is a bounded variation process with $A_0 = 0$.
 Mathematically, the FS decomposition, constitutes the generalisation of the
 martingale representation theorem
(Kunita-Watanabe representation) 
when $S$ is a Brownian motion or a martingale. Given square integrable 
random variable $H$, the problem consists in expressing 
$H$ as $H_0 + \int_0^T \xi dS + L_T$ where $\xi$ is predictable and
$L_T$ is the terminal value of an orthogonal martingale $L$ to $M$,
i.e. the martingale part of $S$.
The seminal paper is \cite{FS91} where the problem is treated 
in the case that $S$ is continuous.
In the general case $S$ is said to have the {\bf structure condition}
(SC) condition if there is a predictable process $\alpha$ such that 
$A_t = \int_0^t \alpha_s d\langle M\rangle_s$ and 
$\int_0^T \alpha_s^2 d\langle M\rangle_s < \infty $ a.s.
In the sequel most of contributions were produced
in the multidimensional case. Here for simplicity we will
formulate all the results in the one-dimensional case. \\
An interesting connection with the theory of backward stochastic differential
equations (BSDEs)  in the sense of \cite{pardouxpeng}, was proposed in 
\cite{S94}.  \cite{pardouxpeng} considered BSDEs driven by Brownian motion;
in \cite{S94} the Brownian motion is in fact replaced by $M$.
The first author who considered
a BSDE driven by a martingale was \cite{buckdahn}.
Suppose that $V_t = \int_0^t \alpha_s d \langle M\rangle_s$. The BSDE problem consists
in finding a triple $(V,\xi, L)$ where
$$ V_t = H - \int_t^T \xi_s dM_s - \int_t^T \xi_s \alpha_s d \langle M\rangle_s
- (L_T - L_t), $$
and $L$ is an $(\shf_t)$-local martingale
orthogonal to $M$. 
\\
In fact, this decomposition provides  the solution to the so called
local risk minimization problem, see \cite{FS91}. 
In this case, $V_t$ represents the {\it price} of the contingent claim
at time $t$ and the price $V_0$ constitutes in fact
 the expectation under
the so called {\it variance optimal measure} (VOM), 
as it will be explained at Section \ref{sec:linkVOM},
with references to
 \cite{S01},
\cite{Arai052} and 
\cite{Arai05}. \\
In the framework of FS decomposition, a
 process which plays a significant role is the so-called
{\bf mean variance tradeoff} (MVT) process $K$. This notion is 
inspired by the theory in
discrete time started by \cite{Schal94}; in the continuous time case
$K$ is defined as $K_t = \int_0^t \alpha^2_s d\langle M\rangle_s, \ t \in [0,T]$.
\cite{S94} shows the existence  of the mean-variance hedging problem 
if the MVT process is deterministic. In fact, a slight more general 
condition was the (ESC) condition and the EMVT process but we will
not discuss here further  details.
We remark that in the  continuous case, treated by \cite{FS91},
 no need of any condition on $K$ is required.
When the MVT process is deterministic, \cite{S94} is able to
solve the global quadratic variation problem and provides an efficient 
relation, see Theorem \ref{mainthm}
 with the FS decomposition. He also shows that, for the obtention 
of the mentioned relation, previous 
condition is not far from being optimal.   
The next important step was done in 
\cite{MS95} where under the only condition that $K$ is uniformly bounded,
the FS decomposition of any square integrable random variable 
admits existence and uniqueness and the global minimization
problem admits a solution. \\
More recently has appeared an incredible amount of papers in
the framework of global (resp. local) risk minimization, so that 
it is impossible 
to list all of them and it is beyond our scope. 
Two significant papers containing a good list of references
are \cite{S01}, \cite{nunno} and \cite{CK08}.
The present paper puts emphasis on processes with independent
increments (PII) and exponential of those processes.
It provides explicit FS decompositions when the process $S$
is of that type when the contingent claims are 
provided by some Fourier transform (resp.
 Laplace-Fourier transform) of a finite  measure.
  Some results of \cite{Ka06} concerning exponential 
of L\'evy processes  are generalized trying to investigate some 
peculiar properties behind and to consider the case of  PII with possibly non stationary increments. The motivation came from 
hedging problems in the electricity market. Because of non-storability of electricity, the hedging instrument 
is in that case, 
a forward contract with value $S^{0}_t = e^{-r(T_d-t)}(F_t^{T_d} - F_0^{T_d})$ where $F_t^{T_d}$ 
is the forward price given at time $t\leq T_d$ for delivery of 1MWh at time $T_d$. Hence, the dynamic of the underlying $S^0$ is directly related to the dynamic of forward prices.  Now, forward prices show a volatility term structure that requires the use of models with non stationary increments and motivates the generalization of the pricing and hedging approach developed in~\cite{Ka06} for Lévy processes to the case of PII with possibly non stationary increments.  
 \\
The paper is organized as follows. After this introduction and 
some generalities
about semimartingales, we introduce the notion of FS decomposition and 
describe local and global risk minimization. 
 Then, we examine at Chapter
3 (resp. 4) the explicit FS decomposition for  PII processes
(resp. exponential of PII processes). Chapter 5 is devoted to the
solution to the global minimization problem and Chapter 6 to
the case of a model intervening in the electricity market. Chapter 7 
is devoted to simulations. 
This paper will be followed by a companion paper, i. e. \cite{GOR}
which concentrates on the discrete time case leaving more
 space to numerical implementations.




\section{Generalities on semimartingales
 and Föllmer-Schweizer decomposition}
\label{Generalities}

\setcounter{equation}{0}

In the whole paper, $T>0$, will be a fixed terminal time and we will 
denote by $(\Omega,\mathcal{F},(\mathcal{F}_t)_{t \in [0,T]},P)$ a
 filtered probability space,
fulfilling the usual conditions.

\subsection{Generating functions}
Let $X=(X_t)_{t \in [0,T]}$ be a real valued stochastic process.
\begin{defi}
\label{defPhi}
The {\bf characteristic function} of (the law of) $X_t$ 
is the continuous mapping 
$$
\varphi_{X_t}:\mathbb{R} \rightarrow \mathbb{C} \quad \textrm{with}\quad \varphi_{X_t}(u)=\mathbb{E}[e^{iu X_t}]\ .
$$
In the sequel, when there will be no ambiguity on the underlying process $X$, we will use the shortened notation $\varphi_t$ for $\varphi_{X_t}$. 
\end{defi}
\begin{defi}
\label{defKappa}
The {\bf cumulant generating function} of (the law of) $X_t$ is the
mapping $z \mapsto {\rm Log} (\mathbb{E}[e^{zX_t}])$ 
where ${\rm Log}(w) = \log(\vert w \vert ) + i {\rm Arg (w)}$
where ${\rm Arg (w)}$ is the Argument of $w$, chosen in $]-\pi,
\pi]$; ${\rm Log}$ is the principal value logarithm.
In particular we have
$$
\kappa_{X_t}:D \rightarrow \mathbb{C} \quad \textrm{with} \quad e^{\kappa_{X_t}(z)}=\mathbb{E}[e^{zX_t}]\ ,
$$
where $D:=\{z \in \mathbb{C}\ \vert\ \mathbb{E}[e^{Re(z)X_t}]<\infty, \ 
 \forall t \in [0,T]\}$.\\

In the sequel, when there will be no ambiguity on the underlying process $X$, we will use the shortened notation $\kappa_t$ for $\kappa_{X_t}$. 
\end{defi}
We observe that $D$ includes the imaginary axis.
\begin{remarque}\label{remarkR2}
\begin{enumerate}
\item For all $z \in D$, $\kappa_t(\bar{z})=\overline{\kappa_t(z)}\ ,$ where   $\bar z$ denotes the conjugate complex of $z\in \mathbb{C}$.
Indeed, for any $z\in D$, 
$$
\exp(\kappa_t(\bar{z}))=\mathbb{E}[\exp(\bar{z}X_t)]=\mathbb{E}[\overline{\exp(zX_t)}]=\overline{\mathbb{E}[\exp(zX_t)]}=\overline{\exp(\kappa_t(z))}=\exp(\overline{\kappa_t(z)})\ .
$$
\item For all $z \in D\cap\mathbb{R}\ ,\ \kappa_t(z) \in \mathbb{R}\ .$
\end{enumerate}
\end{remarque}

\subsection{Semimartingales}
An $(\mathcal{F}_t)$-semimartingale $X=(X_t)_{t \in [0,T]}$ is a process of the form $X=M+A$, where $M$ is an $(\mathcal{F}_t)$-local martingale and $A$ is a bounded variation adapted process vanishing at zero. $||A||_T$ will denote the total variation of $A$ on $[0,T]$. Given two  $(\mathcal{F}_t)$-
local martingales $M$ and $N$, $\left\langle M,N\right\rangle$
 will denote the angle bracket of $M$ and $N$, i.e. the 
unique bounded variation predictable process vanishing at zero such that $MN-\left\langle M,N\right\rangle$ is an $(\mathcal{F}_t)$-local martingale. If $X$ and $Y$ are $(\mathcal{F}_t)$-semimartingales, $\left[X,Y\right]$ denotes the square bracket of $X$ and $Y$, i.e. the quadratic covariation of $X$ and $Y$.  
In the sequel, if there is no confusion about the underlying
 filtration $(\shf_t)$, we will simply speak about
semimartingales, local martingales, martingales.
All the local martingales admit a càdlàg version. By default, 
when we speak about local martingales we always refer to their
càdlàg version.
\\
More details about previous notions are given in chapter I.1. of~\cite{JS03}.
\begin{remarque}\label{remarque17}
\begin{enumerate}
\item All along this paper we will consider $\mathbb{C}$-valued
  martingales (resp. local martingales, semimartingales). Given two
  $\mathbb{C}$-valued local martingales $M^1,M^2$ then
  $\overline{M^1},\overline{M^2}$ are still local martingales. Moreover
  $\langle \overline{M^1},\overline{M^2}\rangle=\overline{\langle
    M^1,M^2 \rangle}\ .$
\item If $M$ is a $\C$-valued martingale then $\langle M, \overline{M} \rangle$
is a real valued increasing process.
\end{enumerate}
\end{remarque}
\begin{thm}
\label{th:XVT}
$(X_t)_{t \in [0,T]}$ is a real semimartingale iff the characteristic function, $t \mapsto \varphi_t (u)$, has bounded variation over all finite intervals,
for all $u \in \R$.
\end{thm}

\begin{remarque}\label{remarque14}
According to Theorem I.4.18 of~\cite{JS03}, any local martingale $M$ admits a unique (up to indistinguishability) decomposition, 
$$
M=M_0+M^c+M^d\ ,
$$
where $M_0^c=M_0^d=0$, $M^c$ is a continuous local martingale and
 $M^d$ is a purely discontinuous local martingale 
in the sense that $\langle N, M^d \rangle = 0$ for 
  all continuous local martingales $N$. 
$M^c$ is called the \textbf{continuous part} of $M$ and $M^d$ the 
\textbf{purely discontinuous part}.
\end{remarque}

\begin{defi}
An $(\shf_t)$-\textbf{special semimartingale} is an $(\shf_t)$-semimartingale 
X with the decomposition $X = M + A$, where $M$ is a local martingale and $A$ is a bounded variation predictable process starting at zero. 
\end{defi}

\begin{remarque}
The decomposition of a special semimartingale of the form $X = M + A$ is 
unique, see~\cite{JS03} definition 4.22. 
\end{remarque}
For any special semimartingale X we define
$$
|| X||^2_{\delta^2}=\mathbb{E}\left[[M,M]_T\right]+\mathbb{E}\left(||A||_T^2\right)\ .
$$
The set $\delta^2$ is the set of $(\mathcal{F}_t)$-special semimartingale $X$ for which $|| X||^2_{\delta^2}$ is finite.

A \textbf{truncation function} defined on $\mathbb{R}$ is a bounded function $h: \mathbb{R} \rightarrow \mathbb{R}$ with compact support such that $h(x)=x$ in a neighbourhood of $0$.  
 
An important notion, in the theory of semimartingales, is the notion of characteristics, defined in definition II.2.6 of~\cite{JS03}. Let $X=M+A$ be a real-valued semimartingale. A \textbf{characteristic} is a triplet, $(b,c,\nu)$, depending on a fixed truncation function, where 

\begin{enumerate}
\item $b$ is a predictable process with bounded variation;
\item $c=\left\langle M^c,M^c\right\rangle$, $M^c$ being the continuous part of $M$ according to Remark \ref{remarque14};
\item $\nu$ is a predictable random measure on $\mathbb{R}^+ \times \mathbb{R}$, namely the compensator of the random measure $\mu^X$ associated to the jumps of X.
\end{enumerate}


\subsection{Föllmer-Schweizer Structure Condition}
\label{sec:FMStruc}

Let $X=(X_t)_{t \in [0,T]}$ be a real-valued
 special semimartingale with canonical decomposition, 
$$
X= M+A\ .
$$
For the clarity of the reader, we formulate in dimension one,
 the concepts appearing in the literature, see e.g. \cite{S94}
 in the multidimensional case.
\begin{defi}\label{defiL2ML2A}
For a given local martingale $M$, the space $L^2(M)$ consists of all predictable $\mathbb{R}$-valued processes $v=(v_t)_{t \in [0,T]}$ such that 
$$
\mathbb{E}\left[\int_0^T|v_s|^2d\left\langle M\right\rangle_s\right]<\infty\ .
$$
For a given predictable bounded variation process $A$, the space $L^2(A)$ consists of all predictable $\mathbb{R}$-valued processes $v=(v_t)_{t \in [0,T]}$ such that 
$$
\mathbb{E}\left[(\int_0^T|v_s|d||A||_s)^2\right]<\infty\ .
$$
Finally, we set 
$$
\Theta:=L^2(M)\cap L^2(A)\ .
$$
\end{defi}
For any $v \in \Theta$, the stochastic integral process
$$
G_t(v):=\int_0^tv_sdX_s,\quad \textrm{for all}\ t \in [0,T]\ ,
$$
is therefore well-defined and is a semimartingale in $\delta^2$ with canonical decomposition
$$
G_t(v)=\int_0^t v_s dM_s+\int_0^t v_sdA_s\ ,\quad \textrm{for all}\ t \in [0,T]\ .
$$
We can view this stochastic integral process as the gain process
associated 
with strategy $v$ on the underlying process $X$.

\begin{defi}
The \textbf{minimization problem} we aim to study is the following.\\
Given $H \in \mathcal{L}^2$, an admissible strategy pair $(V_0,\varphi)$ will be called \textbf{optimal} if $(c,v)=(V_0,\varphi)$ minimizes the expected squared hedging error
\begin{equation}\mathbb{E}[(H-c-G_T(v))^2]\ ,\label{problem2}\end{equation}
over all admisible strategy pairs $(c,v) \in \mathbb{R} \times \Theta$.
 $V_0$ will represent the {\bf price} of the contingent claim $H$ at time zero.
\end{defi}




\begin{defi}\label{defSC}
Let $X=(X_t)_{t \in [0,T]}$ be a real-valued special semimartingale. 
$X$ is said to satisfy the
 \textbf{structure condition (SC)} if there is a predictable
 $\mathbb{R}$-valued process $\alpha=(\alpha_t)_{t \in [0,T]}$ such that 
the following properties are verified.
\begin{enumerate}
	\item  $A_t=\int_0^t \alpha_s d\left\langle M\right\rangle_s\
 ,\quad \textrm{for all}\ t \in [0,T],$ so that
  $dA\ll d\left\langle M\right\rangle $. 
\item  ${\displaystyle \int_0^T \alpha^2_s d\left\langle 
M\right\rangle_s<\infty\ ,\quad P-}$a.s.
\end{enumerate}

\end{defi}

\begin{defi}
From now on, we will denote by $K=(K_t)_{t \in [0,T]}$ the càdlàg process 
$$K_t=\int_0^t \alpha^2_s d\left\langle M\right\rangle_s\ ,\quad \textrm{for all}\ t\in [0,T]\ .$$
This process will be called the \textbf{mean-variance tradeoff} (MVT)  process.
\end{defi}

\begin{remarque}
In~\cite{S94}, the process $(K_t)_{t \in [0,T]}$ is denoted by $(\widehat{K}_t)_{t \in [0,T]}$.
\end{remarque}

We provide here a technical proposition which allows to make
 the class $\Theta$ of integration of $X$ explicit. 
\begin{propo}\label{lemmaTheta}
If $X$ satisfies (SC) such that $\mathbb{E}[K_T]<\infty$, then $\Theta=L^2(M)$.
\end{propo}
\begin{proof}
Assume that $\mathbb{E}[K_T]<\infty$, we will
 prove that $L^2(M)\subseteq L^2(A)$.
 Let us consider a process $v \in L^2(M)$, then 
\begin{eqnarray*}
\mathbb{E} \left (\int_0^T|v_s|d||A||_s \right)^2&=&\mathbb{E}
\left (\int_0^T|v_s||\alpha_s|d\left\langle M\right\rangle_s\right )^2\ 
\leq \left[\mathbb{E}\left(\int_0^T|v_s|^2d\left\langle
 M\right\rangle_s\right)\mathbb{E}
\left (\int_0^T|\alpha|^2d\left\langle M\right\rangle_s \right)
\right]^\frac{1}{2} \ ,\\
&=&\left (\mathbb{E}\left(
\int_0^T|v_s|^2d\left\langle M\right\rangle_s \right) \mathbb{E}(K_T) \right)
^\frac{1}{2}
 <\infty \ .
\end{eqnarray*}
\end{proof}
Schweizer  in~\cite{S94} also introduced the extended structure condition (ESC)
 on $X$ and he provided the Föllmer-Schweizer decomposition 
in this more extended framework. 
We recall that notion (in dimension 1). 
Given a real càdlàg stochastic process $X$, the quantity 
$\Delta X_t$ will represent the jump $X_t - X_{t-}$. 
\begin{defi} \label{ESC}
Let $X=(X_t)_{t \in [0,T]}$ be a real-valued special semimartingale. 
$X$ is said to satisfy the \textbf{extended structure condition (ESC)}
 if there is a predictable
 $\mathbb{R}$-valued process $\alpha=(\alpha_t)_{t \in [0,T]}$ with 
the following properties.  
\begin{enumerate}
\item  $A_t=\int_0^t \alpha_s d\left\langle M\right\rangle_s\
 ,\quad \textrm{for all}\ t \in [0,T]\ ,$ so that
  $dA\ll \left\langle dM\right\rangle $. 
\item The  quantity  
$$\int_0^T
\frac{\alpha_s^2}{1 + \alpha_s^2 \Delta \langle M \rangle_s}
d\left\langle M\right\rangle_s   $$
is finite.
\end{enumerate}
If condition (ESC) is fulfilled,  
then the process
$$ \widetilde{K}_t:=\int_0^t
\frac{\alpha_s^2}{1 + \alpha_s^2 \Delta \langle M \rangle_s}\ , \quad \textrm{for all}\ t\in[0,T]\ ,
$$
is well-defined.
It is called \textbf{extended mean-variance tradeoff (EMVT)}  process.
\end{defi}




\begin{remarque}\label{remarqueKKtilde}
\begin{enumerate}
\item (SC) implies (ESC).
\item If $\left\langle M \right \rangle$ is continuous then
(ESC) and (SC) are equivalent and $K = \tilde K$.
\item $\tilde{K}_t={\displaystyle \int_0^t\frac{|\alpha_s|^2}
{1+\Delta K_s}d\left\langle M\right\rangle_s=\int_0^t\frac{1}
{1+\Delta K_s}dK_s}\ ,$ for all $t\in [0,T]\ .$
\item $K_t={\displaystyle \int_0^t \frac{1}{1-\Delta \tilde{K}_s}d\tilde{K}_s}\ ,$ for all $t\in [0,T]\ .$
\item If $K$ is deterministic then $\tilde{K}$ is deterministic.
\end{enumerate}
\end{remarque}

The structure condition (SC) appears quite naturally
in applications to financial mathematics. In fact, it is mildly related
to the no arbitrage condition. In fact (SC) is a natural extension of 
the existence of an equivalent martingale measure from the case where
 $X$ is continuous. Next proposition will show that every adapted
 continuous process X admitting an equivalent martingale measure
satisfies (SC).
In our applications (ESC) will be equivalent to (SC) since in Section 
\ref{SSCPII} and Section \ref{OSSDAC}, $\langle M \rangle$ will always be continuous.


\begin{propo}\label{NAandSC}
Let $X$ be a $(P,\mathcal{F}_t)$ continuous semimartingale. Suppose the existence of a locally equivalent probability $Q\,\sim\, P$ under which $X$ is an $(Q,\mathcal{F}_t)$-local martingale, then (SC) is verified.
\end{propo}
\begin{proof}
Let $(D_t)_{t \in [0,T]}$ be the strictly positive continuous $Q$-local
martingale such that $dP=D_TdQ$. By Theorem~VIII.1.7  of~\cite{RY},
$M=X-\langle X,L \rangle$ is a continuous $P$-local martingale,
 where $L$ is the continuous $Q$-local martingale associated to the
 density process i.e. 
$$
D_t=\exp\{L_t-\frac{1}{2}\langle L \rangle_t\}\ ,\quad 
\textrm{for all}\ t\in [0,T]\ .
$$
According to Lemma~IV.4.2 in~\cite{RY}, there is a progressively
 measurable process $R$ such that for all $t\in [0,T]$, 
$$
L_t=\int_0^tR_sdX_s+O_t \quad\textrm{and}\quad \int_0^T R^2_sd\left
 \langle X \right \rangle_s<\infty \ , \quad Q-\textrm{a.s.} \ , 
$$
where $O$ is a $Q$-local martingale such that $\langle X,O\rangle=0$. Hence,
$$
\langle X,L\rangle_t = \int_0^t R_s d\langle X \rangle_s  \quad\textrm{and}\quad 
X_t=M_t+\int_0^t R_s d[X]_s\ ,\quad \textrm{for all}\ t\in [0,T]. 
$$
We end the proof by setting
${\displaystyle 
\alpha_t=\frac{d\langle X,L \rangle_t}{d\langle X \rangle_t}=R_t\ .
}$
\end{proof}

\subsection{Föllmer-Schweizer Decomposition 
and variance optimal hedging}
\label{sec:FMDecomp}

Throughout this section, as in Section~\ref{sec:FMStruc}, $X$ is supposed to be an $(\mathcal{F}_t)$-special semimartingale fulfilling the (SC) condition.

We recall here the definition stated in Chapter~IV.3 p.~179 of~\cite{Pr92}. 
\begin{defi}
Two $(\mathcal{F}_t)$-martingales $M,N$ are said to be \textbf{strongly 
orthogonal} if $MN$ is a uniformly integrable martingale.
\end{defi}
\begin{remarque}\label{R126}
If $M,N$ are strongly orthogonal,  then they are (weakly) orthogonal in the sence that $\mathbb{E} [M_TN_T]=0\ .$
\end{remarque}
\begin{lemme}
\label{lem:ortho}
Let $M,N$ be two square integrable martingales.
Then $M$ and $N$ are strongly orthogonal if and only if
 $\left\langle M,N\right\rangle=0$.
\end{lemme}
\begin{proof}
Let $\mathcal{S}(M)$ be the stable subspace generated by M. $\mathcal{S}(M)$ includes the space of martingales of the form
$$
M^f_t:=\int_0^tf(s)dM_s\ ,\quad \textrm{for all}\ t\in [0,T]\ ,
$$
where $f\in L^2(dM)$ is deterministic. According to Lemma~IV.3.2 of~\cite{Pr92}, it is enough to show that, for any $f\in L^2(dM)$, $g\in L^2(dN)$, $M^f$ and $N^g$ are weakly orthogonal in the sense that $\mathbb{E}[M^f_TN^g_T]=0$. This is clear since previous expectation equals
\begin{eqnarray*}\mathbb{E}[\left\langle M^f,N^g\right\rangle_T]=
\mathbb{E}\left (\int_0^Tfgd\left\langle M,N\right\rangle
\right)=0\,\end{eqnarray*}
if $\left\langle M,N\right\rangle=0$.
This shows the converse implication.\\
The direct implication follows from the fact that $MN$ is
a martingale, the definition of the angle bracket and
uniqueness of special semimartingale decomposition.
\end{proof}

\begin{defi}\label{DefFSDecomp}
We say that a random variable $H \in \mathcal{L}^2(\Omega,\mathcal{F},P)$ admits a \textbf{Föllmer-Schweizer (FS) decomposition}, if 
it  can be written as
\begin{eqnarray}H=H_0+\int_0^T\xi_s^HdX_s+L_T^H\ ,
 \quad P-a.s.\label{FSdecompo}\ ,\end{eqnarray}
where $H_0\in \mathbb{R}$ is a constant, $\xi^H \in \Theta$ and $L^H=(L^H_t)_{t \in [0,T]}$ is a square integrable martingale,  
with $\mathbb{E}[L_0^H]=0$ and  strongly orthogonal to $M$.
\end{defi}

We formulate for this section one basic assumption.
\begin{Hyp}\label{Hypo1}
We assume that X satisfies (SC) and that the MVT process $K$
 is uniformly bounded in $t$ and $\omega$.
\end{Hyp}
The first result below gives the existence and the uniqueness of the
Föllmer-Schweizer decomposition for a random variable $H \in
\mathcal{L}^2(\Omega,\mathcal{F},P)$. The second affirms that
subspaces
 $G_T(\Theta)$ and $\{\mathcal{L}^2(\mathcal{F}_0)+G_T(\Theta)\}$ are closed subspaces of $\mathcal{L}^2$ . The last one provides existence and uniqueness of the solution of the minimization problem~(\ref{problem2}). We recall Theorem 3.4 of~\cite{MS95}.
\begin{thm}\label{ThmExistenceFS}
Under Assumption~\ref{Hypo1}, every random variable $H \in \mathcal{L}^2(\Omega,\mathcal{F},\mathcal{P})$ admits a FS decomposition. Moreover, this decomposition is unique in the following sense:\\
If
\begin{eqnarray*}H=H_0+\int_0^T\xi_s^HdX_s+L_T^H=H^{'}_0+\int_0^T\xi_s^{'H}dX_s+L_T^{'H},\end{eqnarray*}
where $(H_0,\xi^H,L^H)$ and $(H^{'}_0,\xi^{'H},L^{'H})$ satisfy the conditions of the FS decomposition, then
$$
\left\{\begin{array}{lll}
H_0&=&H_0^{'}\ ,\quad P-a.s.\ ,\\
\xi^H&=&\xi^{'H}\quad \textrm{in}\   L^2(M) \ , \\
L_T^H&=&L_T^{'H},\quad P-a.s.\ .
\end{array}
\right .
$$ 
\end{thm}
We recall Theorem 4.1 of~\cite{MS95}. 
\begin{thm}\label{ThmGferme}
Under Assumption~\ref{Hypo1}, the subspaces $G_T(\Theta)$ and $\{\mathcal{L}^2(\mathcal{F}_0)+G_T(\Theta)\}$ are closed subspaces of $\mathcal{L}^2$.
\end{thm}
So we can project any random variable $H \in \mathcal{L}^2$ on $G_T(\Theta)$. By Theorem~\ref{ThmExistenceFS}, we have the uniqueness of the solution of the minimization problem~(\ref{problem2}). This is given by Theorem~4.6 of~\cite{MS95}, which is stated below.

\begin{thm}\label{thm128}
We suppose Assumption~\ref{Hypo1}.
\begin{enumerate}
\item For every $H \in \mathcal{L}^2(\Omega,\mathcal{F}, P)$ and every $c \in \mathcal{L}^2(\mathcal{F}_0)$, there exists a unique strategy $\varphi^{(c)}\in \Theta$ such that
\begin{eqnarray}\mathbb{E}[(H-c-G_T(\varphi^{(c)}))^2]=\min_{v\in \Theta}\mathbb{E}[(H-c-G_T(v))^2]\label{problem1}\ .\end{eqnarray}
\item For every $H \in \mathcal{L}^2(\Omega,\mathcal{F},\mathcal{P})$
  there exists a unique $(c^{(H)},\varphi^{(H)}) \in \mathcal{L}^2(\mathcal{F}_0) \times \Theta$ such that
\begin{eqnarray*}\mathbb{E}[(H-c^{(H)}-G_T(\varphi^{(H)}))^2]=\min_{(c,v)\in \mathcal{L}^2(\mathcal{F}_0) \times \Theta}\mathbb{E}[(H-c-G_T(v))^2]\ .\end{eqnarray*}
\end{enumerate}
\end{thm}
Next theorem gives the explicit form of the optimal strategy
$\varphi^{(c)}$, 
which is valid even in the case where X satisfies the extended
 structure condition 
(ESC). For the purpose of the present work, this will not be useful,
see considerations following Remark \ref{remarqueKKtilde} 2.

%
From Föllmer-Schweizer decomposition follows the solution to the global minimization problem~(\ref{problem2}).
\begin{thm}\label{ThmSolutionPb1}
Suppose that X satisifies (SC) and that the MVT process $K$ of X is
deterministic. If $H\in \mathcal{L}^2$ admits a FS decomposition of
type~(\ref{FSdecompo}), then the minimization problem~(\ref{problem1})
has a solution 
$\varphi^{(c)} \in \Theta$ for any $c\in \mathbb{R}$, such that
\begin{eqnarray}
\varphi^{(c)}_t=\xi^H_t+\frac{\alpha_t}{1+\Delta K_t}(H_{t-}-c-G_{t-}
(\varphi^{(c)}))\ ,\quad\textrm{for all}\ t \in [0,T]
\label{xisolution}
\end{eqnarray}
where the process $(H_t)_{t \in [0,T]}$ is defined by
\begin{eqnarray}H_t:=H_0+\int_0^t\xi_s^HdX_s+L_t^H\ ,\label{Vsolution}\end{eqnarray}
and the process $\alpha$ is the process appearing in Definition~\ref{defSC} of~(SC).
\end{thm}
\begin{proof}
Theorem~3 of~\cite{S94} states the result under the (ESC) condition.
 We recall that (SC) implies (ESC), see Remark 
\ref{remarqueKKtilde} and
 the result follows.
\end{proof}
To obtain the solution to the minimization problem~(\ref{problem2}), we use Corollary~10 of~\cite{S94} that we recall.
\begin{corro}\label{ThmSolutionPb2}
Under the assumption of Theorem~\ref{ThmSolutionPb1}, the solution of the minimization problem~(\ref{problem2}) is given by the pair $(H_0,\varphi^{(H_0)})\ .$
\end{corro}
In the sequel of this paper we will only refer to the structure condition (SC) and to the MVT process $K$. The definition below can be found in
section II.8 p.~85 of~\cite{Pr92}.
%

\begin{defi}\label{expdolean} 
The {\bf Doléans-Dade exponential}
 of a semimartingale $X$ is defined to be the unique 
càdlàg adapted solution $Y$ to the stochastic differential equation,  
$$
dY_t=Y_{t-}dX_t\ ,\quad \textrm{for all}\ t\in[0,T]\quad \textrm{with}\ 
 Y_0=1\ .
$$  
This process is denoted by $\mathcal{E}(X)$.
\end{defi}
This solution is a semimartingale given by 
$$
\mathcal{E}(X)_t=\exp\left(X_t-X_0-[X]_t/2\right)\prod_{s\leq t}(1+\Delta X_s)\exp(-\Delta X_s)\ .
$$ 

Theorem below is stated in~\cite{S94}.
\begin{thm}\label{corro9}
Under the assumptions of Theorem~\ref{ThmSolutionPb1}, 
for any $c\in \mathbb{R}$, we have
\begin{eqnarray}\min_{v\in \Theta}\mathbb{E}[(H-c-G_T(v))^2]=\mathcal{E}(-\tilde{K}_T)\left((H_0-c)^2+\mathbb{E}[(L_0^H)^2]+\int_0^T\frac{1}{\mathcal{E}(-\tilde{K}_s)}d\left(\mathbb{E}[\left\langle L^H\right\rangle_s]\right)\right)\ .\label{corro9Sc1}\end{eqnarray}
\end{thm}
\begin{proof}
See the proof of Corollary~9 of~\cite{S94} with Remark~\ref{remarqueKKtilde}.
\end{proof}

\begin{corro}\label{corro132Bis}
If 
$ \left \langle M, M \right \rangle $
 is continuous
\begin{eqnarray}
\min_{v\in \Theta}\mathbb{E}[(H-c-G_T(v))^2]&=&\exp(-K_T)\left((H_0-c)^2+\mathbb{E}[(L_0^H)^2]\right)\nonumber\\
&&+\mathbb{E}\left[\int_0^T\exp\{-(K_T-K_s)\}d\left\langle L^H\right\rangle_s\right]\ .\label{corro9Sc2}\end{eqnarray}
\end{corro}
\begin{proof}
 Remark \ref{remarqueKKtilde}  implies
that $K = \tilde K$. Since   $K$ is continuous and with bounded
variation, its Doléans-Dade exponential coincides with the
classical exponential. The result follows from Theorem \ref{corro9}.

\end{proof}

In the sequel, we will find an explicit expression of the
 FS decomposition for a large class of square integrable random variables, when the underlying process is a process with independent increments, or is an exponential of process with independent increments. For
 this, the first step will consist in verifying (SC) and the boundedness condition on  the MVT process, see Assumption \ref{Hypo1}.

\subsection{Link with the equivalent signed martingale measure}
\label{sec:linkVOM}
\subsubsection{The Variance optimal martingale (VOM) measure}

\begin{defi}
\begin{enumerate}
\item A signed measure, $Q$, on $(\Omega, \mathcal{F}_T)$, is called a
{\bf signed $\Theta$-martingale measure}, if 
\begin{enumerate}
\item $Q(\Omega)=1\ ;$
\item $Q \ll  P$ with ${\displaystyle \frac{dQ}{dP} \in \mathcal{L}^2(P)}\ ;$
\item ${\displaystyle \mathbb{E}[\frac{dQ}{dP}G_T(v)]=0}$ for all $v \in \Theta$.
\end{enumerate}
We denote by $\mathbb{P}_s(\Theta)$, the set of all such signed 
$\Theta$-martingale measures. Moreover, we define
 $$\mathbb{P}_e(\Theta):=\{Q\in\mathbb{P}_s(\Theta) \ |\  Q \sim P\quad \textrm{and Q is a probability measure}\}\ ,$$
and introduce the closed convex set, 
\begin{eqnarray*}\mathcal{D}_d:=\{D \in \mathcal{L}^2(P)\ |\ 
 D=\frac{dQ}{dP}& for \ some &Q \in \mathbb{P}_s(\Theta)\}\ .\end{eqnarray*}
\item A signed martingale measure $\widetilde{P}\in
  \mathbb{P}_s(\Theta)$ is called variance-optimal martingale (VOM)
  measure  if  $\widetilde{D}=argmin_{D \in
    \mathcal{D}_d}Var[D^2]=argmin_{D \in \mathcal{D}_d}
  \left(\mathbb{E}[D^2]-1
\right)$, where $\widetilde{D}={\displaystyle \frac{d\widetilde{P}}{dP}}\ .$
\end{enumerate}
\end{defi}
The space $G_T(\Theta):=\{G_T(v)\ |\ v \in \Theta\}$ is a linear subspace of $\mathcal{L}^2(P)$. Then, we denote by $G_T(\Theta)^\perp$
its orthogonal complement, that is, 
$$
G_T(\Theta)^\perp:=\{D \in \mathcal{L}^2(P)\ |\ \mathbb{E}[DG_T(v)]=0\quad\textrm{for any} \ v\in \Theta\}\ .
$$
Furthermore, $G_T(\Theta)^{\perp\perp}$ denotes the orthogonal complement of $G_T(\Theta)^\perp$, which is the $\shl^2(P)$-closure of 
$G_T(\Theta)$.\\
A simple example when $\mathbb{P}_e(\Theta)$ is non empty is given by the following proposition, that anticipates some material treated in the next
section.

\begin{propo}\label{propoRusso2}
Let $X$ be a process with independent increments such that
\begin{itemize}
\item  $X_t$ has the same law as $-X_t$, for any $t \in [0,T]$;
\item $\frac{1}{2}$ belongs to the domain $D$ of the 
 cumulative generating function $(t, z) \mapsto \kappa_t(z)$.
\end{itemize}
 Then, there is a probability $Q \sim P$ 
such that $S_t=\exp(X_t)$ is a martingale.
\end{propo}
\begin{proof}
For all $t\in[0,T]$, we set 
$D_t=\exp\{-\frac{X_t}{2}-\kappa_t(-\frac{1}{2})\}$.
 Notice that $D$ is a martingale so that the measure $Q$
on $(\Omega, \shf_T)$ defined by  $dQ=D_TdP$ is 
an (equivalent) probability to $P$. 
On the other hand, the symmetry of the law of $X_t$ 
implies for all $t\in [0,T]$, 
\begin{eqnarray*}S_tD_t=\exp\{\frac{X_t}{2}-
\kappa_t(-\frac{1}{2})\}=\exp\{\frac{X_t}{2}-\kappa_t(\frac{1}{2})\}\ 
.\end{eqnarray*}
So $SD$ is also a martingale. According to~\cite{JS03}, 
chapter III, Proposition 3.8 a), $S$ is a $Q$-martingale and so
 $S$ is a $Q$-martingale.
\end{proof}
\begin{example}
Let $Y$ be a process with independent increments. We consider two copies $Y^1$ of $Y$ and $Y^2$ of $-Y$. We set $X=Y^1+Y^2$. Then X has the same law of $-X$.
\end{example}
For simplicity, we suppose from now that Assumption~\ref{Hypo1} is verified, even if one could consider a more general framework, see
\cite{Arai052} Therorem 1.28. This ensures that the linear space 
$G_T(\Theta)$ is closed in $\shl^2(\Omega)$, 
therefore $G_T(\Theta)=\overline{G_T(\Theta)}=G_T(\Theta)^{\perp\perp}$.
 Moreover, Proposition~\ref{lemmaTheta} ensures that
 $\Theta=L^2(M)$.
%
%
We recall an almost known fact cited in~\cite{Arai052}. 
For completeness, we give a proof.
\begin{propo}\label{propo461}
  $\mathbb{P}_s(\Theta) \neq \emptyset\ $ is equivalent to
$1\notin G_T(\Theta)\ .$
\end{propo}
\begin{proof}
Let us prove the two implications. 
\begin{itemize} 
\item Let $Q\in \mathbb{P}_s(\Theta)$. 
If $1\in G_T(\Theta)$, 
then $Q(\Omega)=\mathbb{E}^Q(1)=0$ which leads to a contradiction
since $Q$ is a probability.
Hence $1\notin G_T(\Theta)$.
\item Suppose that $1\notin G_T(\Theta)$. 
 We denote by $f$ the orthogonal projection of 1 on $G_T(\Theta)$. Since $\mathbb{E}[f(1-f)]=0$, then $\mathbb{E}[1-f]=\mathbb{E}[(1-f)^2]$. Recall that $1\neq f\in G_T(\Theta)$, hence we have $\mathbb{E}[f]\neq 1$. Therefore, we can define the signed measure $\widetilde{P}$ by setting
\begin{eqnarray}\widetilde{P}(A)=\int_A\widetilde{D}dP\ ,\quad \textrm{with}\quad \widetilde{D}=\frac{1-f}{1-\mathbb{E}[f]}\ .\label{F1}\end{eqnarray}
%
We check now that $\widetilde{P}\in \mathbb{P}_s(\Theta)$. 
\begin{itemize}
\item Trivially $\widetilde{P}(\Omega)=\mathbb{E}(\widetilde{D})=1\ ;$
\item $\widetilde{P} \ll  P\ ,$ by construction.
\item Let $v\in \Theta$,
${\displaystyle \mathbb{E}[\widetilde{D}G_T(v)]=\frac{1}{1-\mathbb{E}[f]}\left(\mathbb{E}[(1-f)G_T(v)]\right)=0}\ ,$ 
since $1-f\in G_T(\Theta)^{\perp}$. 
\end{itemize}
Hence, $\widetilde{P}\in \mathbb{P}_s(\Theta)$ which concludes the proof of the Proposition.
\end{itemize}
\end{proof}
\begin{remarque} \label{R461}  
If 1 is orthogonal to $G_T(\Theta)$, then $f = 0$
and $P\in\mathbb{P}_s(\Theta)$ so $\mathbb{P}_s(\Theta) \neq \emptyset$. 
\end{remarque}
In fact, $\widetilde{P}$ constructed in the proof of Proposition~\ref{propo461} coincides with the VOM measure.
\begin{propo}
\label{propo:D}
Let $\widetilde{P}$ be the signed measure defined in~(\ref{F1}). Then, 
$$
\widetilde{D}=\argmin_{D \in \mathcal{D}_d}\mathbb{E}[D^2]=\argmin_{D
  \in  \mathcal{D}_d}Var[D]\ .
$$
\end{propo}
\begin{proof}
Let $D \in \mathcal{D}_d$ and Q such that $dQ=DdP$. We have to 
show that $\mathbb{E}[D^2]\geq \mathbb{E}[\widetilde{D}^2]$. We write
\begin{eqnarray*}
\mathbb{E}[D^2]
=\mathbb{E}[(D-\widetilde{D})^2]+\mathbb{E}[\widetilde{D}^2]+\frac{2}{1-\mathbb{E}[f]}\mathbb{E}[(D-\widetilde{D})(1-f)]\ .\end{eqnarray*}
Moreover, since $f\in G_T(\Theta)$ yields
\begin{eqnarray*}\mathbb{E}[(D-\widetilde{D})(1-f)]&=&\mathbb{E}[D]-\mathbb{E}[\widetilde{D}]-\mathbb{E}[Df]+\mathbb{E}[\widetilde{D}f]\ ,\\
&=&Q(\Omega)-\widetilde{Q}(\Omega)\ .\\
&=&0\ .\end{eqnarray*}
\end{proof}
\begin{remarque}
\begin{enumerate}
\item Arai~\cite{Arai05} 
 gives sufficient conditions under which the VOM measure 
is a probability, see Theorem 3.4 in~\cite{Arai05}.
\item Taking in account Proposition~\ref{propo461}, 
the property $1\notin G_T(\Theta)$
   may be viewed as non-arbitrage condition. In fact, in
   \cite{delbaen94}, the existence of a martingale measure
which is a probability is equivalent to a {\it no free lunch condition}.
\end{enumerate}
\end{remarque}
Next proposition can be easily deduced for a more general formulation, see~\cite{S01}.
\begin{propo}\label{P2} 
We assume Assumption \ref{Hypo1}. Let $H\in \shl^2(\Omega)$ and consider the solution $(c^{H},\varphi^{H})$ of the minimization problem~(\ref{problem2}). Then, the price $c^{H}$ equals the expectation under the VOM measure $\widetilde{P}$ of H.
\end{propo}
\begin{proof}
We have 
\begin{eqnarray*}H=c^{H}+G_T(\varphi^{H})+R\ ,\end{eqnarray*}
where $R$ is orthogonal to $G_T(\Theta)$ and 
$\mathbb{E}[R]=0$. Since $\widetilde{P}\in \mathbb{P}_s(\Theta)$,
 taking the expectation with respect to $\widetilde{P}$, 
denoted by $\widetilde{\mathbb{E}}$ we obtain
\begin{eqnarray*}\widetilde{\mathbb{E}}[H]=c^{H}+\widetilde{\mathbb{E}}[R]\ .\end{eqnarray*}
From the proof of Proposition~\ref{propo461}, we have
\begin{eqnarray*}\widetilde{\mathbb{E}}[R]=\frac{\mathbb{E}[(1-f)R]}{1-\mathbb{E}[f]}=\frac{1}{1-\mathbb{E}[f]}\left(\mathbb{E}[R]-\mathbb{E}[fR]\right)\ .
\end{eqnarray*}
Since $f \in G_T(\Theta)$ and $R$ is orthogonal to $G_T(\Theta)$, we get $\widetilde{\mathbb{E}}[R]=0\ .$
\end{proof}

\section{Processes with independent increments (PII)}

\setcounter{equation}{0}

This section deals with the case of Processes with Independent
Increments. The preliminary part recalls some useful properties of such
processes. Then, we obtain a sufficient condition  on the characteristic
function for the existence of the FS decomposition. Moreover, an 
explicit FS decomposition is derived. 

Beyond its own theoretical interest, this work is motivated by its possible application to hedging and pricing energy derivatives and specifically electricity derivatives. Indeed, one  way of modeling electricity forward prices is to use arithmetic models such as the Bachelier model which was developed for standard financial assets. 
The reason for using arithmetic models, is that the usual hedging intrument available on electricity markets are swap contracts which give a fixed price for the delivery of electricity over a contracted time period. Hence, electricty swaps can be viewed as a strip of forwards for each hour of the delivery period. In this framework, arithmetic models have the significant advantage to yield closed pricing formula for swaps which is not the case of geometric models. 


However, in whole generality, an arithmetic model allows  negative
prices which could be underisable. Nevertheless, in the electricity
market,  negative prices may occur because it can be more expensive for a producer to switch off some generators than to pay someone to consume the resulting excess of production. Still, in~\cite{Benth-Kallsen}, is introduced a class of arithmetic models where the positivity of spot prices is ensured, using a specific choice of increasing Lévy process. The parameters estimation of this kind of model is studied in~\cite{Meyer-Tankov}.

\subsection{Preliminaries}
\begin{defi}\label{defPAI}
$X=(X_t)_{t \in [0,T]}$ is a (real) 
\textbf{process with independent increments (PII)} iff
\begin{enumerate}
\item $X$ is adapted to the filtration $(\mathcal{F}_t)_{t \in [0,T]}$ and has
 càdlàg paths. 

\item $X_0=0$.
\item $X_t-X_s$ is independent of $\mathcal{F}_s$ for $0 \leq s < t \leq T$.\\\\
Moreover we will also suppose
\item $X$ is continuous in probability, i.e. 
 $X$ has no fixed time of discontinuties.
\end{enumerate}
\end{defi}

We recall Theorem II.4.15 of~\cite{JS03}. 
\begin{thm}\label{ThmCharaDeter}
Let $(X_t)_{t \in [0,T]}$ be a real-valued special semimartingale, with $X_0=0$. Then, $X$ is a process with independent increments,
 iff there is a version $(b,c,\nu)$ of its characteristics that is deterministic. 
\end{thm}
\begin{remarque} \label{RemCharaDeter}
In particular, $\nu$ is a (deterministic non-negative) measure on the Borel
$\sigma$-field of $[0,T] \times \mathbb{R}$. 
\end{remarque}
From now on, given two reals $a,b$, we  denote by $a\vee b$
 (resp. $a \wedge b$)  
the maximum (resp. minimum) between $a$ and $b$.

\begin{propo}\label{propoJSA}
Suppose $X$ is a semimartingale with independent increments  with
characteristics $(b,c,\nu)$, then there exists an increasing
function
 $t\mapsto a_t$ such that
\begin{eqnarray}
db_t\ll da_t\ ,\quad
dc_t\ll da_t \quad \textrm{and}\quad 
\nu(dt,dx)=\widetilde{F}_t(dx)da_t\ ,
\label{170}\end{eqnarray}
where $\widetilde{F}_t(dx)$ is a non-negative kernel from $\big([0,T],
\mathcal{B}([0,T])\big )$ into $(\mathbb{R} ,\mathcal{B} )$
verifying
\begin{equation} \label{E1000}
\int_\R  (\vert x \vert^2 \wedge 1) \tilde F_t(dx) \le 1\,, \quad \forall t
\in [0,T].
\end{equation}
and
\begin{eqnarray}a_t=||b||_t+c_t+\int_{\mathbb{R} } (|x|^2\wedge 1)\nu([0,t],dx)\ ,\label{171}\end{eqnarray}
\end{propo}
\begin{proof}
The existence of $(a_t)$ as a process fulfilling $(\ref{171})$ 
and $\tilde F$ fulfilling \eqref{E1000} 
is provided by the statement and the proof of  Proposition II.
2.9 of~\cite{JS03}. $(\ref{171})$ and Theorem \ref{ThmCharaDeter} guarantee that $(a_t)$ is deterministic.
\end{proof}
\begin{remarque}\label{remarque16}
In particular, $(b_t)$, $(c_t)$ and 
$t\mapsto \int_{[0,t] \times B} (\vert x \vert^2 \wedge 1) \nu(ds,dx)  $ 
has bounded variation for any $B \in \mathcal{B} $.
\end{remarque}

The proposition below provides the so called
{\bf Lévy-Khinchine Decomposition}.
\begin{propo}\label{LKdeco}
Assume that $(X_t)_{t \in [0,T]}$ is a process with independent increments.
Then
\begin{eqnarray}\varphi_t (u)= e^{\Psi_t(u)}\ ,
\quad \textrm{for all} \ u \in \R
\ ,  \end{eqnarray}
 $\Psi_t$, is given by the 
Lévy-Khinchine decomposition of the process $X$, 
\begin{eqnarray}
\Psi_t(u)=iu b_t-\frac{u^2}{2}c_t+\int_{\mathbb{R}} 
(e^{iu x}-1-iu h(x))F_t(dx)\ ,\quad \textrm{for all}\ u\in\mathbb R 
\label{LKFormule}\ ,
\end{eqnarray}
where $B\mapsto F_t(B)$ is the positive measure $\nu([0,t]\times B)$ 
which integrates $1 \wedge |x|^2$ for any $t \in [0,T]$.
\end{propo}
We introduce here a simplifying hypothesis for this section.
\begin{Hyp} \label{HPAIND} For any $t > 0$, $X_t$ is never deterministic.
\end{Hyp}
\begin{remarque}\label{rem:moment} We suppose Assumption \ref{HPAIND}.
\begin{enumerate}
\item Up to a $2 \pi i$ addition of $\kappa_t(e)$, we can write
$\Psi_t(u) = \kappa_t(iu), \ \forall u \in \R$.
From now on we will always make  use of this modification. 
\item $\varphi_t(u) $ is never a negative number. Otherwise,
there would be $u \in \R^\ast, t > 0$ such that
$E(cos(uX_t)) = -1$. Since $cos(uX_t) +1 \ge 0$ a.s.
then $cos(uX_t) = -1$ a.s. and this is not possible since $X_t$
is non-deterministic.  
\item Previous point implies that all the differentiability properties 
of $u \mapsto \varphi_t(u)$ are equivalent to those of 
 $u \mapsto \Psi_t(u)$.
\item If $\mathbb{E}[|X_t|^2]<\infty$, then for all $u \in \mathbb{R}$, $\Psi^{'}_t(u)$ and $\Psi^{''}_t(u)$ exist.
\end{enumerate}
\end{remarque}
We come back to the cumulant generating function $\kappa$ and its domain $D$.
\begin{remarque} \label{CumGenPAI}
In the case where the underlying process is a PII, then 
\begin{eqnarray*}
 D := \{z \in \mathbb{C}\ \vert \ \mathbb{E}[e^{Re(z)X_t}]<\infty, \  \forall t \in
 [0,T]\} =\{z \in \mathbb{C}\ \vert \ \mathbb{E}[e^{Re(z)X_T}]<\infty\}\ .
\end{eqnarray*}
In fact, for given $t \in [0,T], \gamma \in \mathbb{R} $ we have
$$   \mathbb{E}(e^{\gamma X_T}) = \mathbb{E}(e^{\gamma X_t}) 
\mathbb{E}(e^{\gamma (X_T - X_t)})  < \infty.$$
Since each factor is positive,  and if the left-hand side is finite,
then $\mathbb{E}(e^{\gamma X_t})$ is also 
 finite.
 \end{remarque}

We need now a result which extends the L\'evy-Khinchine decomposition
to the cumulant generating function.
Similarly to Theorem~25.17 of~\cite{Sa99} we have.
\begin{propo}\label{LKdecoKappa}
Let $D_0=\left\{c \in \mathbb{R}\ |\ \int_{[0,T]\times \left\{|x|>1\right\}}e^{cx}\nu(dt,dx)<\infty \right\}$. Then, 
\begin{enumerate}
\item $D_0$ is convex and contains the origin.
\item $D_0= D\cap\mathbb{R}$.
\item If $z\in \mathbb{C}$ such that $Re(z)\in D_0$, i.e. $z\in D$, then
\begin{eqnarray} \label{E3.5bis}
\kappa_t(z)=z b_t+\frac{z^2}{2}c_t+\int_{[0,t]
\times\mathbb{R}} (e^{zx}-1-z h(x))\nu(ds,dx)\ .\end{eqnarray}
\end{enumerate}
\end{propo}

\begin{proof}
\begin{enumerate}
\item is a consequence of Hölder inequality similarly as i) in Theorem 25.17 of~\cite{Sa99}\ .
\item The characteristic function of the law of $X_t$ is given by~(\ref{LKFormule}). According to Theorem~II.8.1~(iii) of Sato~\cite{Sa99}, there is an
infinitely divisible distribution with characteristics 
$(b_t,c_t, F_t(dx))$, fulfilling 
$F_t(\{0\}) = 0$ and $\int (1 \wedge x^2) F_t(dx)< \infty$ and $c_t \ge
0$.
By uniqueness of the characteristic function, that law is precisely the
law of $X_t$. By Corollary~II.11.6, in~\cite{Sa99}, there is a L\'evy process
$(L^t_s, 0 \le s \le 1)$ such that $L^t_1$ and $X_t$ are identically 
distributed. We define
$$
C_0^t = \{ c\in \mathbb{R} \ \vert \ 
\int_{\{\vert x \vert > 1\}} e^{cx} F_t(dx) < \infty \} \quad\textrm{and}\quad
C^t = \{z \in \mathbb C \ \vert\  \mathbb{E}\left[\exp (Re (z L_1^t)\right ] < \infty \}  \ .
$$
Remark \ref{CumGenPAI} says that $C^T = D$, moreover clearly
$C_0^T = D_0$.
Theorem V.25.17 of~\cite{Sa99} implies $D_0 = D \cap \mathbb{R}$,
i.e. point 2. is established.
\item Let $t \in [0,T]$ be fixed; let $w \in D$.
We apply point (iii) of Theorem V.25.17 of~\cite{Sa99} 
to the L\'evy process $L^t$.
\end{enumerate}
\end{proof}

\begin{propo}\label{corroR1}
Let $X$ be a semimartingale with independent increments.
For all $z \in D$, 
 $t\mapsto \kappa_t(z)$  has bounded variation and
\begin{eqnarray}\kappa_{dt}(z)\ll da_t\ .\end{eqnarray}
\end{propo}

\begin{proof}

Using (\ref{E3.5bis}), it remains to prove that
\begin{eqnarray*}t\mapsto \int_{[0,T]\times \R}(e^{zx}-1-zh(x))\nu(ds,dx)\end{eqnarray*}
is absolutely continuous with respect to $(da_t)$. 
We can conclude
 \begin{eqnarray*}\kappa_t(z) = 
  \int_0^t \frac{db_s}{da_s} da_s +\frac{z^2}{2}
 \int_0^t \frac{dc_s}{da_s} da_s  \nonumber +
 \int_0^tda_s\int_{\mathbb{R}}\left(e^{zx}-1-zh(x)\right)\widetilde{F}_s(dx)\ ,
\end{eqnarray*}
if we show that
\begin{eqnarray}\int_0^T da_s\int_{\R}|e^{zx}-1-zh(x)|\widetilde{F}_s(dx)<\infty\ . \label{A131}\end{eqnarray}
Without restriction of generality we can suppose 
$h(x)=x1_{|x|\leq 1}$.
(\ref{A131}) can be bounded by the sum $I_1+I_2+I_3$ where
%
$$
I_1=\int_0^Tda_s\int_{|x|>1}|e^{zx}|\widetilde{F}_s(dx)\ ,\quad 
I_2=\int_0^Tda_s\int_{|x|>1}\widetilde{F}_s(dx)\ ,\quad\textrm{and}\quad 
I_3=\int_0^Tda_s\int_{|x|\leq1}|e^{zx}-1-zx|\widetilde{F}_s(dx)\ .
$$
Using Proposition \ref{propoJSA}, we have
\begin{eqnarray*}
I_1&=&\int_0^Tda_s\int_{|x|>1}|e^{zx}|\widetilde{F}_s(dx)\,\\
&=&\int_0^Tda_s\int_{|x|>1}|e^{Re(z)x}|\widetilde{F}_s(dx)\,\\
&=&\int_{[0,T]\times |x|>1}|e^{Re(z)x}|\nu(ds,dx);
\end{eqnarray*}
this quantity is finite because $Re(z)\in D_0$ taking into account
 Proposition \ref{LKdecoKappa}. Concerning $I_2$ we have
\begin{eqnarray*}
I_2&=&\int_0^Tda_s\int_{|x|>1}\widetilde{F}_s(dx)\,\\
&=&\int_0^Tda_s\int_{|x|>1}(1\wedge |x^2|)\widetilde{F}_s(dx)\,\\
&\leq&a_T
\end{eqnarray*}
because of \eqref{E1000}. As far as $I_3$ is concerned, we have
\begin{eqnarray*}
I_3&\leq&e^{Re(z)}\frac{z^2}{2}\int_{[0,T]\times |x|\leq 1}da_s(x^2\wedge 1)\widetilde{F}_s(dx)\,\\
&=&e^{Re(z)}\frac{z^2}{2}a_T
\end{eqnarray*}
again because of \eqref{E1000}. This concludes the proof the Proposition.

\end{proof}

The converse of the first part of previous corollary also holds.
For this purpose we formulate first a simple remark.

\begin{remarque}\label{R3}
For every $z \in D$, $\left(\exp(zX_t-\kappa_t(z))\right)$ is a martingale.
In fact, for  all $0\leq s\leq t\leq T$, we have
\begin{eqnarray}\mathbb{E}[\exp(z(X_t-X_s))]=\exp(\kappa_t(z)-\kappa_s(z))\ .
\label{F}
\end{eqnarray}
\end{remarque}

\begin{propo} Let $X$ be a PII. 
Let $z \in D \cap \mathbb{R^\star}$.
$(X_t)_{t \in [0,T]}$ is a semimartingale iff 
$t \mapsto \kappa_t(z)$ has bounded variation.
\end{propo}
\begin{proof} 
It remains to prove the converse implication.\\
If $t \mapsto  \kappa_t(z)$ has  bounded variation then 
$ t \mapsto e^{\kappa_t(z)})$
 has the same property. Remark \ref{R3} says that
 $e^{z X_t} = M_t e^{\kappa_t (z)}$ where $(M_t)$ is a martingale.
 Finally, $(e^{z X_t})$ is a semimartingale  
and taking the logarithm 
 $(zX_t)$  has the same property.
\end{proof}

\begin{remarque}\label{R10}
Let $z \in D$. If $(X_t)$ is a semimartingale with independent
increments then $(e^{zX_t})$
 is necessarily a special semimartingale
since it is the product of a martingale and a
 bounded variation continuous deterministic function, by use
of integration by parts.
\end{remarque}

\begin{lemme}\label{P1}
Suppose that $(X_t)$ is a semimartingale 
with independent increments. 
Then for every $z \in Int(D)$, $t\mapsto \kappa_t(z)$ is continuous.
\end{lemme}
\begin{remarque}\label{R11}
 The conclusion remains true for any process which is continuous
in probability, whenever $t \mapsto \kappa_t(z)$ is (locally) bounded.
\end{remarque}

\begin{proof}
[\textbf{Proof of Lemma \ref{P1}}]
Since $z \in Int(D)$, there is $\gamma > 1$ such that $\gamma z \in D$; so 
\begin{eqnarray*}\mathbb{E}[\exp(z\gamma X_t)]=\exp(\kappa_t(\gamma z))\leq \exp(\sup_{t\leq T}(\kappa_t(\gamma z))) \ ,\end{eqnarray*}
because $t\mapsto \kappa_t(\gamma z)$ is bounded, being of bounded variation. This implies that $(\exp(zX_t))_{t \in [0,T]}$ is uniformly integrable. Since $(X_t)$ is continuous in probability, then $(\exp(zX_t))$ is continuous in $\mathcal{L}^1$. The result easily follows.
\end{proof}

\begin{propo}\label{proposition38}
The function $(t,z)\mapsto \kappa_t(z)$ is continuous. In
particular, $(t,z) \mapsto \kappa_t(z)$, $t \in [0,T]$, 
$z$ belonging to a compact real subset, is bounded.
\end{propo}
\begin{proof}
\begin{itemize}
\item Proposition \ref{LKdecoKappa} implies that 
$z \mapsto \kappa_t(z)$ is continuous uniformly with respect to $t \in [0,T]$.
\item By Lemma
 \ref{P1}, for  $z \in {\rm Int D}$, $t\mapsto \kappa_t(z)$ is continuous.
\item To conclude it is enough to show that  $t\mapsto \kappa_t(z)$ is continuous
for every $z \in D$.
Since $\bar D = \overline {\rm Int D}$, there is a sequence $(z_n)$ in the interior of $D$ 
converging to $z$. Since a uniform limit of continuous functions on $[0,T]$ 
converges to a continuous function, the result follows.
\end{itemize}
\end{proof}


\subsection{Structure condition for PII (which are semimartingales)} \label{SSCPII}


Let $X=(X_t)_{t \in [0,T]}$ be a real-valued semimartingale with
independent  increments and $X_0 = 0$.
 We assume that $\mathbb{E}[|X_t|^2]<\infty$.
We denote by  $\varphi_t(u)= \mathbb{E}[\exp(iu X_t)]$ the 
characteristic function of $X_t$ and by 
$u \mapsto \Psi_t(u)$ its  log-characteristic function 
introduced in Proposition \ref{LKdeco}. We recall 
that  $\varphi_t(u)=\exp(\Psi_t(u))$. \\
 $X$ has the property of independent increments; therefore
\begin{eqnarray}\exp(iu X_t)/\mathbb{E}[\exp(iu X_t)]=\exp(i u X_t)/\exp(\Psi_t(u)) \ ,\end{eqnarray}
is a martingale.


\begin{remarque}\label{rem:moment:b}
Notice that the two first order moments of $X$ are related to
 the log-characterisctic function of $X$, as follows 
\begin{eqnarray}
\mathbb{E}[X_t]&=&-i\Psi_t^{'}(0)\ , \quad
\mathbb{E}[X_t-X_s]=-i(\Psi_t^{'}(0)-\Psi_s^{'}(0)),\label{EspXts} \\
Var(X_t)&=&-\Psi^{''}_t(0)\ , \quad
Var(X_t-X_s)=-[\Psi^{''}_t(0)-\Psi^{''}_s(0)]  \ .\label{VarXts}
\end{eqnarray}
\end{remarque}

%

\begin{propo}\label{prop:SCPII}
Let $X=(X_t)_{t \in [0,T]}$ be a real-valued semimartingale
 with independent increments.
\begin{enumerate}
\item $X$ is a special semimartingale with decomposition
$X = M + A$ with the following properties:
\begin{eqnarray} \left\langle M\right\rangle_t=-\Psi^{''}_t(0)\
 \quad \textrm{and}\quad A_t=-i\Psi^{'}_t(0)\ .\label{MAPII}\end{eqnarray}
In particular $t \mapsto  -\Psi^{''}_t(0)$ is increasing and therefore
of bounded variation.
\item $X$ satisfies condition~(SC) of Definition~\ref{defSC} if and only if 
\begin{equation}
\label{SCPII}
 \Psi^{'}_t(0) \ll\Psi^{''}_t(0) \quad \textrm{and}\quad
{\displaystyle \int_0^T\left|\frac{d_t \Psi^{'}_s}{d_t \Psi^{''}_s}(0)\right|^2|d\Psi_s^{''}(0)|<\infty}\ .
\end{equation}
In that case
\begin{equation}
A_t=\int_0^t \alpha_sd\left\langle M\right\rangle_s \quad  \textrm{with}\quad \alpha_t=i\frac{d_t \Psi^{'}_t(0)}{d_t \Psi^{''}_t(0)}\quad \textrm{for all}\ t\in [0,T]\ .\label{AMPII}
\end{equation}
\item Under condition~(\ref{SCPII}), FS decomposition exists (and it is unique)
 for every square integrable random variable. 
\end{enumerate}
\end{propo}
In the sequel, we will provide an explicit decomposition for a class of contingent claims, under condition~(\ref{SCPII}). 

\begin{proof}
\begin{enumerate}
\item
Let us first determine $A$ and $M$ in terms of the log-characteristic function of $X$. Using~(\ref{EspXts}) of Remark~\ref{rem:moment:b}, we get
\begin{eqnarray*}
\mathbb{E}[X_t|\mathcal{F}_s] & = & \mathbb{E}[X_t-X_s+X_s\ |\
 \mathcal{F}_s]\ ,
  =  \mathbb{E}[X_t-X_s]+X_s\ ,\\
 & = & -i\Psi^{'}_t(0)+i\Psi^{'}_s(0)+X_s\ ,\quad \textrm{then}\ ,\\
 \mathbb{E}[X_t+i\Psi^{'}_t(0)|\mathcal{F}_s] &=&X_s +i\Psi^{'}_s(0)\ .
\end{eqnarray*}
 Hence, $(X_t+i\Psi^{'}_t(0))$ is a martingale and the canonical decomposition of $X$ follows 
\begin{eqnarray*}
X_t=\underbrace{X_t+i\Psi^{'}_t(0)}_{M_t}\underbrace{-i\Psi^{'}_t(0)}_{A_t} \ ,
\end{eqnarray*}
where $M$ is a local martingale and $A$ is a locally bounded variation process thanks to 
the  semimartingale property of $X$.  
%
Let us now determine  $\langle M\rangle $, in terms of the log-characteristic function of $X$. 
\begin{center}$ \begin{array}{ccl}
M_t^2&=&[X_t+i\Psi^{'}_t(0)]^2\ ,\\\\
 \mathbb{E}[M^2_t|\mathcal{F}_s] & = & \mathbb{E}[(X_t+i\Psi^{'}_t(0))^2|\mathcal{F}_s]\ ,\\\\
   & = & \mathbb{E}[(X_s+i\Psi^{'}_s(0)+X_t-X_s+i\Psi^{'}_t(0)-i\Psi^{'}_s(0))^2|\mathcal{F}_s]\ ,\\\\
   & = & \mathbb{E}[(M_s+X_t-X_s+i(\Psi^{'}_t(0)-\Psi^{'}_s(0)))^2|\mathcal{F}_s]\ ,\\\\
\end{array}$\end{center}
Using~(\ref{EspXts}) and~(\ref{VarXts}) of Remark~\ref{rem:moment:b}, yields 
\begin{eqnarray*}
M_t^2 & = & \mathbb{E}[(M_s-\mathbb{E}[X_t-X_s]+X_t-X_s)^2]\ ,\\
 & = & M_s^2+Var(X_t-X_s)
= M_s^2-\Psi^{''}_t(0)+\Psi^{''}_s(0) \ .
\end{eqnarray*}
Hence, $(M^2_t+\Psi^{''}_t(0))$ is a $(\mathcal{F}_t)$-martingale,
and point 1. is established.
$$
A_t=\int_0^t \alpha_sd\left\langle M\right\rangle_s \quad 
 \textrm{with}\quad \alpha_t=i\frac{d_t \Psi^{'}_t(0)}{d_t
   \Psi^{''}_t(0)}
\quad \textrm{for all}\ t\in [0,T]\ .
$$
\item is a consequence of point 1. and of Definition~\ref{defSC}.
\item follows from  Theorem \ref{ThmExistenceFS}. In fact 
$K_t =  
- \displaystyle \int_0^T \left(\frac{d_t \Psi^{'}_s}{d_t \Psi^{''}_s}(0)
\right)^2d\Psi_s^{''}(0)$
is deterministic and so Assumption \ref{Hypo1} is fulfilled.
\end{enumerate}

\end{proof}

\subsection{Examples}
\subsubsection{A continuous process example}
Let $\psi:[0,T]\rightarrow     \mathbb{R}$ be a continuous strictly increasing
function, $\gamma:[0,T]\rightarrow \mathbb{R}$ be a bounded variation
function such that $d\gamma \ll  d\psi$. We set
$X_t=W_{\psi(t)}+\gamma(t)$, 
where $W$ is the standard Brownian motion on $\mathbb{R}$. Clearly, $X_t=M_t+\gamma(t)$, where $M_t=W_{\psi(t)}$, defines a continuous martingale, such that $\left\langle M\right\rangle_t=\left[M\right]_t=\psi(t)$. Since 
$X_t\sim \mathcal{N}(\gamma(t),\psi(t))$ for all $u \in \mathbb{R}$ and $t \in [0,T]$, we have
\begin{eqnarray*}
\Psi_t(u)=i\gamma(t)u-\frac{u^2\psi(t)}{2}\ ,
\end{eqnarray*}
which yields 
\begin{eqnarray*}
\Psi^{'}_t(0)=i\gamma(t) \quad \textrm{and}\quad 
\Psi^{''}_t(0)=-\psi(t)\ ,
\end{eqnarray*}
Therefore, if ${\displaystyle \frac{d\gamma}{d\psi} \in \mathcal{L}^2(d\psi)}$, then $X$ satisfies condition~(SC) of Definition~\ref{defSC} with 
\begin{eqnarray*}A_t=\int_0^t\alpha_sd\left\langle M\right\rangle_s \quad \textrm{and}\quad \alpha_t=\left . \frac{d\gamma}{d\psi}\right\vert_t
\quad \textrm{for all}\ t\in [0,T]\ .
\end{eqnarray*}

\subsubsection{Processes with independent and stationary increments
  (Lévy processes)  }
\label{sec:example:Levy}

\begin{defi}\label{defLEVY}
$X=(X_t)_{t \in [0,T]}$ is called {\bf Lévy process} 
or process with stationary and independent
increments if the following properties hold.
\begin{enumerate}
\item $X$ is adapted to the filtration $(\mathcal{F}_t)_{t \in [0,T]}$
  and has càdlàg trajectories.
\item $X_0=0$.
\item The distribution of $X_t-X_s$ depends only on $t-s$ for $0 \leq s \leq t \leq T$.
\item $X_t-X_s$ is independent of $\mathcal{F}_s$ for $0 \leq s \leq t \leq T$.
\item $X$ is continuous in probability.
\end{enumerate}
\end{defi}
For details on Lévy processes, we refer the reader to \cite{Pr92},
\cite{Sa99}
 and \cite{JS03}.\\
Let $X=(X_t)_{t \in [0,T]}$ be a real-valued Lévy process, with $X_0=0$. 
 We assume that $\mathbb{E}[|X_t|^2]<\infty$ and we do not consider the trivial case where $L_1$ is deterministic. 
\begin{remarque}\label{rem:moment:levy}
\begin{enumerate}
\item Since $X=(X_t)_{t \in [0,T]}$ is a Lévy process then $\Psi_t(u)=t\Psi_1(u)$. In the sequel, we will use the shortened notation  $\Psi:=\Psi_1$. 
\item $\Psi$ is a function of class $C^2$ and $\Psi^{''}(0)=Var(X_1)$ which
is strictly positive if $X$ has no stationary increments.

\end{enumerate}
\end{remarque}

\subsection{Cumulative and characteristic functionals in some particular cases}
\label{LevyPC}

We recall some cumulant and log-caracteristic functions 
of some typical Lévy processes.
\begin{remarque}\label{remark47Bis}
\begin{enumerate}
\item \underline{Poisson Case:}
If $X$ is a Poisson process with intensity $\lambda$,
 we have that $\kappa^\Lambda(z)=\lambda(e^z-1)$.
 Moreover, in this case the set $D=\mathbb{C}$. \\
Concerning the log-characteristic function we have
$$
\Psi(u)=\lambda(e^{iu}-1)\ ,\quad 
\Psi^{'}(0) =  i\lambda \quad \textrm{and}\quad 
\Psi^{''}(0) =  -\lambda, u \in \R.
$$
\item \underline{NIG Case:}
This process was introduced by Barndorff-Nielsen in~\cite{bnh}.
Then $X$ is a Lévy process with 
$X_1 \sim NIG(\alpha,\beta,\delta,\mu)$, with $\alpha>|\beta|>0$, $\delta>0$
 and $\mu \in \mathbb{R}$.
 We have $\kappa^\Lambda(z)=\mu z + \delta (\gamma_0-\gamma_z)$ and
 $\gamma_z=\sqrt{\alpha^2-(\beta+z)^2}$, 
$D = ]-\alpha-\beta,\alpha-\beta[+i\mathbb{R}\,.$\\
Therefore
$$
\Psi(u)=\mu iu + \delta (\gamma_0-\gamma_{iu})\ ,
\quad\textrm{where}\quad \gamma_{iu}=\sqrt{\alpha^2-(\beta+iu)^2}\ .
$$
By derivation, one gets 
$$
\Psi^{'}(0) =  i\mu + \delta \frac{i\beta}{\gamma_0} \quad \textrm{and}\quad 
\Psi^{''}(0) =  -\delta(\frac{1}{\gamma_0}+ \frac{\beta^2}{\gamma^3_0} ),
$$
Which yields
${\displaystyle \alpha=i\frac{\Psi^{'}(0)}{\Psi^{''}(0)}=
\frac{\gamma^2_0(\gamma_0\mu+\delta\beta)}{\delta(\gamma^2_0+\beta)}\ .}
$
\item \underline{Variance Gamma case:} Let $\alpha, \beta > 0, \delta \neq 0$.
If $X$ is a Variance Gamma process with
 $X_1 \sim VG(\alpha,\beta,\delta,\mu)$ with
 $\kappa^\Lambda(z)=\mu z + \delta Log\left(\frac{\alpha}
{\alpha-\beta z - \frac{z^2}{2}}\right)\,,$
where $Log$ is again the principal value complex logarithm
defined in Section 2.
%
The expression of $\kappa^\Lambda(z)$ can be found in
\cite{Ka06, madan} or also
~\cite{livreTankovCont},  table IV.4.5 in the particular case  $\mu = 0$.
 In particular an easy calculation shows that we need $z \in \mathbb{C}$ 
such that $Re(z) \in ]-\beta-\sqrt{\beta^2+2\alpha},-\beta+
\sqrt{\beta^2+2\alpha}[$ so that $\kappa^\Lambda(z)$ is well defined
so that
$$ D =  ]-\beta-\sqrt{\beta^2+2\alpha},-\beta+
\sqrt{\beta^2+2\alpha}[ + i \R.$$
Finally we obtain 
$$ \Psi(u) = \mu iu + \delta Log\left(\frac{\alpha}{\alpha-\beta iu +
 \frac{u^2}{2}}\right). $$
After derivation it follows
$$
\Psi^{'}(0) =  i(\mu - \delta \beta), \quad
\Psi^{''}(0) =  \frac{\delta}{\alpha} (\alpha^2 - \beta^2) .
$$
\end{enumerate}

\end{remarque}

%
%

\subsection{Structure condition in the L\'evy case}

By application of Proposition~\ref{prop:SCPII} and 
Remark~\ref{rem:moment:levy}, we get the following result. 
\begin{corro}\label{cor:SCLevy}
Let $X=M+A$ be the canonical decomposition of $X$, then for all $t\in[0,T]$, 
\begin{eqnarray}\left\langle M\right\rangle_t=-t\Psi^{''}(0) \quad \textrm{and}\quad A_t=-it\Psi^{'}(0)\ .\label{MALevy}\end{eqnarray}
Moreover $X$ satisfies condition~(SC) of Definition~\ref{defSC} with 
\begin{equation}
A_t=\int_0^t \alpha d\left\langle M\right\rangle_s \quad  \textrm{with}\quad \alpha=i \frac{\Psi^{'}(0)}{ \Psi^{''}(0)}\quad \textrm{for all}\ t\in [0,T]\ .\label{AMPIIS}
\end{equation}
Hence, FS decomposition exists for every square integrable random variable. 
\end{corro}
\begin{remarque} \label{LevyPCSC}
We have the following in previous three examples of subsubsection
\ref{LevyPC}
\begin{enumerate}
\item{\bf Poisson case:}
$\alpha = 1$.
\item{\bf NIG process:}
${\displaystyle \alpha =
\frac{\gamma^2_0(\gamma_0\mu+\delta\beta)}{\delta(\gamma^2_0+\beta)}\ .}$
\item{\bf VG process:}
$\alpha = {\displaystyle \frac{\mu - \delta \beta}{\alpha^2 - \beta^2} \frac{\alpha}{\delta}}\ .$
\end{enumerate}

\end{remarque}

\subsubsection{Wiener integrals of Lévy processes}
We take $X_t=\int_0^t\gamma_sd\Lambda_s$, where $\Lambda$ is a square integrable
Lévy process as in Section~\ref{sec:example:Levy}.  Then,
$\int_0^T\gamma_s d\Lambda_s$ is well-defined for at least $\gamma \in
\shl^\infty([0,T])$. It is then possible to calculate the characteristic
function and the cumulative function of $\int_0^\cdot \gamma_sd\Lambda_s$.
Let $(t,z) \mapsto t \Psi_\Lambda (z),$ \        
 (resp. $(t,z) \mapsto t \kappa^\Lambda(z)$)
 denoting the log-characteristic function 
 (resp. the cumulant generating function) of $\Lambda$.
\begin{lemme}\label{lemme27}
Let $\gamma:[0,T] \rightarrow \R  $ 
be a Borel bounded function.
\begin{enumerate}
	\item The log-characteristic function of $X_t$
is such that for all $u\in\mathbb{R}$, 
$$
\Psi_{X_t}(u)=\int_0^t\Psi_\Lambda(u\gamma_s)ds\ ,\quad \textrm{where}\quad 
\mathbb{E}[\exp(iu X_t)]=\exp\big(\Psi_{X_t}(u)\big)\ ;
$$
	\item Let $D_\Lambda$ be the domain related to $\kappa^\Lambda$ 
in the sense of
Definition \ref{defKappa}.
The cumulant generating function  of $X_t$
 is such that for all $z\in\{z\,\vert\, 
Re z\gamma_t\in D_\Lambda\  \textrm{for all} 
\ t\in[0,T]\}$, 
$$
\kappa_{X_t}(z)=\int_0^t\kappa^\Lambda(z\gamma_s)ds.
$$
\end{enumerate}

%
\end{lemme}
\begin{proof}
We only prove 1. since 2. follows similarly. Suppose first $\gamma$ to
 be continuous, then $\int_0^T\gamma_sd\Lambda_s$ is the 
limit in probability of $\sum_{j=0}^{p-1}\gamma_{t_j}(\Lambda_{t_{j+1}}-
\Lambda_{t_j})$ where $0=t_0<t_1<...<t_p=T$ is a subdivision of $[0,T]$
 whose mesh converges to zero. Using the independence of the increments,
 we have 
\begin{eqnarray*}
\mathbb{E}\left[\exp\{i\sum_{j=0}^{p-1}\gamma_{t_j}
(\Lambda_{t_{j+1}}-\Lambda_{t_j})\}\right]
&=&\prod_{j=0}^{p-1}\mathbb{E}\left[\exp\{i\gamma_{t_j}
(\Lambda_{t_{j+1}}-\Lambda_{t_j})\}\right]\ ,\\ \\
&=&\prod_{j=0}^{p-1}\exp\{\Psi_\Lambda(\gamma_{t_j})(t_{j+1}-{t_j})\}\ ,\\ \\
&=&\exp\{\sum_{j=0}^{p-1}(t_{j+1}-{t_j})\Psi_\Lambda(\gamma_{t_j})\}\ .
\end{eqnarray*}
This converges to $\exp\left(\int_0^T\Psi_\Lambda(\gamma_s)ds\right)$, when the
mesh
 of the subdivision goes to zero.\\
Suppose now that $\gamma$ is only bounded and
 consider,  using convolution, a sequence $\gamma_n$
of continuous functions, such that $\gamma_n \rightarrow \gamma$ a.e. 
and $\sup_{t \in [0,T]}|\gamma_n(t)|\leq \sup_{t \in
  [0,T]}|\gamma(t)|$. 
We have proved that
\begin{eqnarray}
\mathbb{E}\left[\exp\left(i\int_0^T\gamma_n(s)d\Lambda_s\right)\right]
=\exp\left(\int_0^T\Psi_\Lambda(\gamma_n(s))ds\right)\label{F11}
\end{eqnarray}
Now, $\Psi_\Lambda$ is continuous therefore bounded, so Lebesgue dominated 
convergence and continuity of stochastic integral imply statement 1.
\end{proof}
\begin{remarque}\label{remarque27}
\begin{enumerate}
\item Previous proof, which is left to the reader, also applies for statement 2. This statement in a slight different form is proved in~\cite{BS03}
\item We prefer to formulate a direct proof. In particular statement 1. holds with the same proof even if
 $\Lambda$ has no moment condition and $\gamma$ is a continuous function with bounded variation. Stochastic integrals are then defined using integration by parts.
\end{enumerate}
\end{remarque}
We suppose now that $\Lambda$ is a Lévy process such that
$\Lambda_1$ is not deterministic. In particular $Var(\Lambda_1) \neq 0$
and so $\Psi_\Lambda'' \neq 0$. \\
  In this case
\begin{eqnarray*}
\Psi_t^{'}(u)=\int_0^t\Psi_\Lambda^{'}(u\gamma_s)\gamma_sds
\quad \textrm{and}\quad \Psi_t^{''}(u)=\int_0^t\Psi_\Lambda^{''}
(u\gamma_s)\gamma^2_sds\ .\end{eqnarray*}
So
\begin{eqnarray*}\Psi_t^{'}(0)=\Psi_\Lambda^{'}(0)
\int_0^t\gamma_sds \quad \textrm{and}\quad 
\Psi_t^{''}(0)=\Psi_\Lambda^{''}(0)\int_0^t\gamma^2_sds\ .\end{eqnarray*}
Condition (SC) is verified since $d\Psi_t^{'}(0)\ll d\Psi_t^{''}(0)$ with 
$$
{\displaystyle \alpha_t=i\frac{d\Psi_t^{'}(0)}{d\Psi_t^{''}(0)}=
\frac{\Psi_\Lambda^{'}(0)}{\Psi_\Lambda^{''}(0)}\,\frac{i}{\gamma_t}} 
1_{\{\gamma_t \neq 0 \}} 
\quad \textrm{and}\quad 
{\displaystyle \int_0^T \alpha_s^2 \,
\vert \Psi_s^{''}(0)\vert \gamma_s^2  ds
=T \frac{\vert \Psi_\Lambda^{'}(0)\vert ^2}{\vert 
\Psi_\Lambda^{''}(0)\vert } <\infty}\ .
$$

\subsection{Explicit Föllmer-Schweizer decomposition in the PII case}
\label{sec:FSPII}

\subsubsection{Preliminaries}\label{s61}
Let $X=(X_t)_{t \in [0,T]}$ be a semimartingale 
 with independent increments with log-characteristic function 
$(t,u)\mapsto \Psi_t(u)$. We assume that $(X_t)_{t \in [0,T]}$ is
 square integrable and satisfies Assumption \ref{HPAIND}.
\begin{remarque} \label{RQEFS}
\begin{enumerate}
\item $u\mapsto \Psi_t(u)$ is of class $C^2$, for any $t \in [0,T]$
because $X_t$ is square integrable. 
\item $t\mapsto \Psi^{''}_t(0)$ and $t\mapsto \Psi^{'}_t(0)$
have bounded variation 
because of Proposition \ref{prop:SCPII}. Therefore, they
are bounded.
\item $t\mapsto \Psi^{'}_t(u)$ is continuous for every $u \in \R$.
In fact, first $t \mapsto X_t$ is continuous in probability.
Since $M_t = X_t - \Psi_t'(0)  $ is a square integrable martingale and
 $t \mapsto \Psi^{'}_t(0)$  is bounded, then
the family $(E(X_t^2))$ is bounded and so $(X_t)$  is uniformly
integrable.
So $t \mapsto \varphi_t'(u)$ is continuous and the result follows
by Assumption \ref{HPAIND}  
\item $t\mapsto \Psi^{''}_t(0)$ is continuous.
In fact, again it is enough to prove $t \mapsto \varphi_t''(0)$
is continuous.
This follows if we prove that $(M_t)$ is continuous
in ${\cal L}^2$. This is true because $M$ is continuous in probability
and for any $N > 0$, $ t \in [0,T]$, Chebyshev implies that 
$$ P\{ \vert M^2_t \vert > N \} \le \frac{{\rm Var} (X_t)}{N} 
\le  \frac{{\rm Var} (X_T)}{N}, $$
and so the family $(M_t^2)$ is again uniformly integrable.
\end{enumerate}
\end{remarque}
 We suppose the following.
\begin{Hyp}\label{HypoFSPII:1}
\begin{enumerate}
\item  $t \mapsto \Psi^{'}_t(u)$ is absolutely continuous with respect to
 $d\Psi^{''}_t(0)$.
\item For every $u \in \R$, we suppose that the following quantity
\begin{equation} \label{KappaU2}
K(u) : = \int_0^T \left \vert\frac{d\Psi_t^{'}(u )}{d\Psi_t^{''}(0)}
\right \vert^2 d(- \Psi_s^{''}(0)) 
\end{equation}
is finite.
\end{enumerate}
\end{Hyp}
\begin{remarque} \label{RA2}
If $u=0$, the previous quantity~\eqref{KappaU2} is finite
because of the (SC) condition.
\end{remarque}
We consider a contingent claim which is given as a Fourier transform of $X_T$, 
\begin{eqnarray}H=f(X_T) \quad \textrm{with}\quad f(x)=\int_{\mathbb{R}}e^{iu x}\mu(du)\ ,\quad \textrm{for all}\ x\in\mathbb{R}\ ,\label{ES1}\end{eqnarray}
for some finite signed measure $\mu$. 
\begin{Hyp}\label{HypoFSPII:2}
$$ \int_\R K(u) d\vert \mu(u)\vert < \infty.$$
\end{Hyp}
\begin{remarque} \label{RA3}
We observe that the function
$$ (u,t) \mapsto \exp(\Psi_T(u) - \Psi_t(u))$$
is uniformly bounded because the characteristic function is bounded.
\end{remarque}


We will  first  evaluate an explicit Kunita-Watanabe decomposition 
 of $H$ w.r.t. the martingale part $M$ of $X$. Later, we will finally obain the decomposition with respect to $X$.

\subsubsection{Explicit elementary Kunita-Watanabe decomposition}

By Propostion~\ref{prop:SCPII}, $X$ admits the following semimartingale 
decomposition, $X_t=A_t+M_t$, where
\begin{eqnarray}A_t=-i\Psi^{'}_t(0)\ 
\quad\textrm{and}\quad \left\langle M\right\rangle_t=-\Psi^{''}_t(0)\label{97bis}\ .\end{eqnarray}
\begin{propo}\label{propo521}
Let $H=f(X_T)$ where $f$ is of the form~(\ref{ES1}).
 We suppose that the PII $X$ satisfies Assumptions~
\ref{HPAIND},
\ref{HypoFSPII:1}
 and~\ref{HypoFSPII:2}. Then, $H$ admits the decomposition
\begin{equation} \label{h}
 \left  \{ 
\begin{array}{ccc}  
V_t &=& V_0 + \int_0^t Z_s dM_s + O_t  \\
V_T &=& H \ ,
\end{array} \right.      
\end{equation}
with the following properties.
\begin{enumerate}
\item For all $t\in [0,T]$, 
\begin{eqnarray}Z_t=i\int_{\mathbb{R}}e^{iu X_{t-}}\,
\frac{d\left(\Psi^{'}_t(u )-\Psi^{'}_t(0)\right)}{d\Psi^{''}_t(0)}
\exp\left\{\Psi_T(u )-\Psi_t(u )\right\}\, \mu(du)\label{Z}\ ;\end{eqnarray}
\item $O$ is a square integrable $(\mathcal{F}_t)$-martingale such that $\left\langle O,M\right\rangle=0\ ;$ 
\item $H=V_T$ where $(V_t)_{t\in[0,T]}$ is the $(\mathcal{F}_t)$-martingale defined by
\begin{eqnarray}V_t=\mathbb{E}[H\vert \mathcal{F}_t]=\int_{\mathbb{R}}e^{iu X_t}\exp\left\{\Psi_T(u )-\Psi_t(u )\right\}\mu(du)\label{V}\ .\end{eqnarray}
\end{enumerate}
\end{propo}

\begin{remarque}
In particular, 
\begin{enumerate}
\item $V_0=\mathbb{E}[H]\,;$
\item $\mathbb{E}\left[\int_0^TZ_s^2d\left\langle M\right\rangle_s\right]<\infty\,.$
\end{enumerate}
\end{remarque}

\begin{proof}
We start with the case $\mu=\delta_u (dx)$ for some $u \in \mathbb{R}$ so that $f(x)=e^{iux}$. We consider the $(\mathcal{F}_t)$-martingale $V_t=\mathbb{E}[f(X_T)|\mathcal{F}_t]=\mathbb{E}[e^{ iuX_T}|\mathcal{F}_t]$.
\begin{enumerate}
\item Clearly $V_0=\mathbb{E}[e^{iuX_T}]\,.$
\item We calculate explicitely $V_t$, which gives 
\begin{eqnarray*}V_t &=& \mathbb{E}[e^{iu X_T}|\mathcal{F}_t]\ 
= e^{iu X_t}\mathbb{E}[e^{iu (X_T-X_t)}]\ 
= \exp(iu X_t-\Psi_t(u ))\exp(\Psi_T(u )) \\
&=&\widetilde{V_t}\exp(\Psi_T(u ))\ ,\end{eqnarray*}
where $\widetilde{V_t}=\exp(iu X_t-\Psi_t(u ))$ defines an $(\mathcal{F}_t)$-martingale.
\item We evaluate $\left\langle V,M\right\rangle$.

\begin{lemme}\label{lemme63}
$\left\langle V,M\right\rangle_t=-i\int_0^tV_s(\Psi^{'}_{ds}(u )-\Psi^{'}_{ds}(0))\ .$
\end{lemme}
\begin{proof}
We evaluate $\mathbb{E}[\widetilde{V_t}M_t|\mathcal{F}_s]$. Since $\widetilde{V}$ and $M$ are $(\mathcal{F}_t)$-martingales and using the property of independent increments we get
\begin{eqnarray*}
\mathbb{E}[\widetilde{V_t}M_t|\mathcal{F}_s]
&=&\mathbb{E}[\widetilde{V_t}M_s|\mathcal{F}_s]+\mathbb{E}[\widetilde{V_t}(M_t-M_s)|\mathcal{F}_s]\ ,\\\\
&=&M_s\widetilde{V_s}+\widetilde{V_s}\mathbb{E}[\exp\{iu (X_t-X_s)-(\Psi_t(u )-\Psi_s(u ))\}(M_t-M_s)]\ ,\\\\
&=&M_s\widetilde{V_s}+\widetilde{V_s}e^{-(\Psi_t(u )-\Psi_s(u ))}\mathbb{E}[e^{iu (X_t-X_s)}(M_t-M_s)]\ .
\end{eqnarray*}
Previous expectation gives
\begin{eqnarray*}
\mathbb{E}[e^{iu (X_t-X_s)}(M_t-M_s)]&=&\mathbb{E}[e^{iu (X_t-X_s)}(X_t-X_s)]+\mathbb{E}[e^{iu (X_t-X_s)}i(\Psi_t^{'}(0)-\Psi_s^{'}(0))]\ ,\\\\
&=&-i\frac{\partial}{\partial u }\mathbb{E}[e^{iu (X_t-X_s)}]+i(\Psi_t^{'}(0)-\Psi_s^{'}(0))\mathbb{E}[e^{iu (X_t-X_s)}]\ ,\\\\
&=&-ie^{\Psi_t(u )-\Psi_s(u )}(\Psi^{'}_t(u )-\Psi^{'}_s(u ))
+i(\Psi_t^{'}(0)-\Psi_s^{'}(0))e^{\Psi_t(u )-\Psi_s(u )}\ .\end{eqnarray*}
Consequently,
\begin{eqnarray*}
\mathbb{E}[\widetilde{V_t}M_t|\mathcal{F}_s]
&=&
M_s\widetilde{V_s}-i\widetilde{V_s}(\Psi^{'}_t(u )-\Psi^{'}_s(u ))+i\widetilde{V_s}(\Psi_t^{'}(0)-\Psi_s^{'}(0))\\ \\
&=&M_s\widetilde{V_s}-i\widetilde{V_s}\left(\Psi^{'}_t(u )-\Psi_t^{'}(0)-(\Psi^{'}_s(u )-\Psi_s^{'}(0))\right)\ .
\end{eqnarray*}
This implies that $\left( \widetilde{V_t}M_t+i\widetilde{V_t}(\Psi^{'}_t(u )-\Psi_t^{'}(0))\right )_t$ is an $(\mathcal{F}_t)$-martingale.
Then by integration by parts, 
$$
\widetilde{V_t}(\Psi^{'}_t(u )-\Psi_t^{'}(0))=\int_0^t\widetilde{V_s}(\Psi^{'}_{ds}(u )-\Psi_{ds}^{'}(0))+\int_0^t(\Psi^{'}_{s}(u )-\Psi_{s}^{'}(0))d\widetilde{V_s}\ .
$$
The second integral term of the right-hand side being a martingale,  it follows that
\begin{eqnarray*}
\left\langle \widetilde{V},M\right\rangle_t=-i \int_0^t\widetilde{V_s}(\Psi^{'}_{ds}(u )-\Psi_{ds}^{'}(0))  \ .
\end{eqnarray*}
and so
\begin{eqnarray}
\left\langle V,M\right\rangle_t=-i \int_0^tV_s(\Psi^{'}_{ds}(u )-\Psi_{ds}^{'}(0))\ .
\label{VM}\end{eqnarray}

\end{proof}

\item We continue the proof of the Proposition \ref{propo521}.
 For given $(Z_t)$ we have
\begin{eqnarray*} \left\langle \int_0^t ZdM,M\right\rangle_t=\int_0^t Z_{s-}d\left\langle M\right\rangle_s=-\int_0^tZ_s\Psi^{''}_{ds}(0)\ .\end{eqnarray*}
\item  We want to identify
\begin{eqnarray*}-\int_0^tZ_s\Psi^{''}_{ds}(0)=-i\int_0^tV_s(\Psi^{'}_{ds}(u )-\Psi_{ds}^{'}(0))\ .\end{eqnarray*}
This naturally leads to
\begin{eqnarray}Z_s=i\frac{d(\Psi^{'}_{s}(u )-\Psi_{s}^{'}(0))}
{d\Psi^{''}_s(0)}V_{s-}\ .\label{ZZS}\end{eqnarray}
\item Finally, we obtain the general case, for general finite signed measure $\mu$,  similarly to the proof of Theorem~\ref{propo310} (in the sequel) in the case of exponential of PII processes. The use of Fubini's theorem is essential.
\end{enumerate}
\end{proof}

\begin{example}
We take $X=M=W$ the classical Wiener process. We have $\Psi_s(u )=-\frac{u ^2s}{2}$ so that $\Psi^{'}_s(u )=-u s$ and $\Psi^{''}_s(u )=-s$. So $Z_s=iu V_s$. We recall that 
\begin{eqnarray*}V_s=\mathbb{E}[\exp(iu W_T)|\mathcal{F}_s]=\exp(iu W_s)\exp\left(-u ^2\frac{T-s}{2}\right)\ .\end{eqnarray*}
In particular,  $V_0=\exp(-\frac{u ^2T}{2})$ and so
\begin{eqnarray*}\exp(iu W_T)=i\int_0^Tu \exp(iu W_s)
\exp\left(-u ^2\frac{T-s}{2}\right) dW_s + 
\exp(-\frac{u ^2T}{2}).\end{eqnarray*}
In fact that expression is classical and 
it can be derived from Clark-Ocone formula.
\end{example}

\subsubsection{Explicit Föllmer-Schweizer decomposition}


We introduce a quantity which will be useful in the sequel. For
$t \in [0,T], u \in \R$ we set
\begin{equation} \label{EA3}
\eta(u,t) = \int_0^t \frac{d(\Psi_s^{'}(u )-\Psi_s{'}(0))}
{d(\Psi^{''}_s(0))}\Psi^{'}_{ds}(0)\ .
\end{equation}
\begin{remarque} \label{rhmCO}
\begin{enumerate}
\item $\eta$ is defined unambiguously since  
$d\left(\Psi^{'}_t(u )-\Psi^{'}_t(0)\right) $ is absolutely 
continuous with respect to $d\Psi^{''}_t(0) \ .$
\item $\eta$ is well-defined, because for any $u \in  \R$,
$$ \eta(u,t) = \int_0^t  \frac{d(\Psi_s^{'}(u )-\Psi_s{'}(0))}
{d(\Psi^{''}_s(0))}  
\frac{d(\Psi_s{'}(0))}
{d(\Psi^{''}_s(0))}
d\Psi^{''}_{ds}(0)  
$$
is bounded by Cauchy-Schwarz, taking into account Assumption
\ref{HypoFSPII:1} point 2.
\end{enumerate}
\end{remarque}

We are now able to evaluate the FS decomposition of
 $H=f(X_T)$ where $f$ is given by~(\ref{FORM}).

We introduce now a supplementary  hypothesis.
\begin{Hyp} \label{ASupeta} The quantity
$$ \sup_{u \in {\rm supp} \mu, t \in [0,T]} 
( {\rm Re(\eta(u,t)}) < \infty \ .$$
\end{Hyp}

\begin{thm}\label{thmCO}
Under the assumptions of Proposition~\ref{propo521}
and Assumption \ref{ASupeta},
 the FS decomposition of $H$ 
is the following
\begin{eqnarray}H_t=H_0+\int_0^t \xi_sdX_s+L_t \quad\textrm{with}\quad
  H_T=H\ 
\end{eqnarray}
and
\begin{eqnarray} \label{EA1}
H_t & = & \int_{\mathbb{R}}H(u)_t\mu(du)\ , \nonumber \\
&& \\
 \xi_t&=&\int_{\mathbb{R}}\xi(u )_t\mu(du) \  \nonumber ,
\end{eqnarray}
where
\begin{eqnarray} \label{EA2}
\xi(u )_t &=& i\frac{d(\Psi^{'}_t(u )-\Psi^{'}_t(0))}
{d\Psi^{''}_t(0)}\,H(u)_{t-}\ , \nonumber \\
&& \\
H(u )_t&=&\exp\left\{\eta(u ,T)-\eta(u ,t)+\Psi_T(u )-\Psi_t(u )
\right\}\,e^{iuX_t}\ .
 \nonumber
\end{eqnarray}
\end{thm}

\begin{proof}
Using Fubini's theorem, we reduce the problem to show that
\begin{eqnarray*}H(u)_t=H(u )_0+\int_0^t  \xi(u)_sdX_s+L(u )_t
\quad \textrm{with}\quad H(u )_T=\exp(iu X_T)\ ,\end{eqnarray*}
for fixed $u \in\mathbb{R}$
where $L(u)$ is a square integrable martingale and 
$\langle L(u),M \rangle  = 0$, where $M$ is the martingale part of the
special semimartingale $X$. 
Notice that by equation~(\ref{EA2}), 
$$
H(u)_t= H(u)_0+ e^{\int_t^T\eta(u ,ds)} V(u)_t\quad \textrm{with}\quad  V(u )_t=\exp(iu X_t+\Psi_T(u )-\Psi_t(u ))\ .
$$
 Integrating by parts, gives
%
\begin{eqnarray*}H(u)_t=H(u)_0-\int_0^t e^{\int_r^T\eta(u ,ds)} V(u)_r
\eta(u ,d r)+\int_0^t e^{\int_r^T \eta(u ,ds)} dV(u)_r\ .\end{eqnarray*}
We denote again by $Z(u)$ the expression provided by $(\ref{ZZS})$. 
We recall that 
$$
dV(u)_r=Z(u)_rdM_r+dO(u)_r=Z(u)_r(dX_r-dA_r)+dO(u)_r\ ,
$$ 
where $A$ is given by~$(\ref{97bis})$ and $O$ is a square
 integrable martingale  strongly 
orthogonal to $M$ (i.e. $\left\langle M,O\right\rangle_.=0$).
$$
H(u )_t=H(u )_0+L(u )_t+\int_0^t e^{\int_r^T\eta(u ,ds)} Z(u )_rdX_r-\int_0^t e^{\int_r^T\eta(u ,ds)}Z(u )_r(-i\Psi_{dr}^{'}(0))-\int_0^t e^{\int_r^T\eta(u ,ds)}V(u )_r\eta(u ,dr)\ ,
$$
where
\begin{eqnarray*}L(u)_t=\int_0^t e^{\int_r^T\eta(u ,ds)} dO(u )_r\ ,\end{eqnarray*}
is a martingale strongly orthogonal to $M$. To conclude, 
we need to choose $\eta$ so that
\begin{eqnarray*}\int_0^tZ(u )_r e^{\int_r^T\eta(u ,ds)}(-i\Psi_{dr}^{'}(0))=\int_0^t e^{\int_r^T\eta(u ,ds)} V(u )_r\eta(u ,dr))\ .\end{eqnarray*}
This requires
\begin{eqnarray*}\eta(u ,dr)=
\frac{d(\Psi_r^{'}(u )-\Psi_r{'}(0))}{d(\Psi^{''}_r(0))}\Psi^{'}_{dr}(0)\ .\end{eqnarray*}
So we define $\eta$ as in
 \eqref{EA3}.
\end{proof}

\subsubsection{The Lévy case}

Let X be a square integrable Lévy process, with characteristic function
$\exp(\Psi(u )t)$. In particular, $\Psi$ is of class $C^2(\mathbb{R})$.
 We have
$$
\frac{d\Psi^{'}_t(u )}{d\Psi^{''}_t(0)}=
\frac{\Psi^{'}(u )}{\Psi^{''}(0)}\quad \textrm{and}\quad \eta(u ,t)=t\frac{\Psi^{'}(u )-
\Psi^{'}(0)}{\Psi^{''}(0)}\Psi^{'}(0)\ .
$$
%
%
We remark that Assumptions \ref{HPAIND} is verified. Concerning
Assumption \ref{HypoFSPII:1}, point 1. is trivial; point 2. is verified
because 
$K(u)=\left|{\displaystyle
    \frac{\Psi^{'}(u)}{\Psi^{''}(u)}}\right |^2 (-T\Psi^{''}(u)) \ .$ 
On the other hand Assumption \ref{ASupeta} is verified if 
\begin{eqnarray}Re\left(\frac{\Psi^{'}(u)\Psi^{'}(0)}
{\Psi^{''}(0)} \right)<\infty \ .\label{ARePsi}\end{eqnarray}
Since $\Psi^{'}(0)=i\E[X_1]$ and $\Psi^{''}(0) < 0$,  \eqref{ARePsi} is fulfilled if
\begin{eqnarray} \E[X_1]Im(\Psi^{'}(u))\label{EARePsi}>-\infty \ .
\end{eqnarray}
 Assumption \ref{HypoFSPII:2} is verified if 
\begin{eqnarray}\int_{\R}\left|\Psi^{'}(u)\right|^2d\vert \mu(u)\vert
  <\infty\ .\label{EAmu}\end{eqnarray}
\begin{example}\label{GaussianFS}
We start with the toy model $X_t=\sigma W_t+m t$, $\sigma,m \in \R$. We have $\Psi(u)=-\frac{u^2}{2}\sigma^2+im u$ so 
$\Psi^{'}(u)=-u\sigma^2+im$ and $Im(\Psi^{'}(u))=m$. Condition
\eqref{EARePsi} is always verified and Condition 
\eqref{EAmu} is
 verified if 
\begin{eqnarray}\int_{\R}u^2d\mu(u)<\infty\ .\label{EAmu1}\end{eqnarray}

\eqref{EAmu1} is verified for instance in the example of the
 beginning of subsection \ref{RSCFT} since $\int_{-\infty}^cu^2e^udu<\infty$ for $c>0.$
\end{example}

\begin{remarque}\label{FSLevy10}
In the examples introduced in Remark \ref{remark47Bis}, we can show that $u\mapsto \left|\Psi^{'}(u)\right|$ is bounded and so \eqref{EARePsi} and \eqref{EAmu} are always verified for the following reasons.
\begin{enumerate}
\item \underline{Poisson case}\\ We have $\Psi^{'}(u)=i\lambda e^{iu}\ . $
\item \underline{NIG case}\\ We have $\Psi{'}(u)=i\mu+i\delta\left(\beta+iu\right)\left(\alpha^2-(\beta+iu)^2\right)^{-\frac{1}{2}} \ .$
Now 
\begin{eqnarray*}\left|\Psi{'}(u)\right|\leq2\left(|\mu|^2+2\delta \sqrt{\frac{\beta^2+u^2}{(\alpha^2-\beta^2+u^2)^2+4u^2\beta^2}}\right)\ .\end{eqnarray*}
Since $|\alpha|>|\beta|$, $u\mapsto \left|\Psi{'}(u)\right|$ is bounded.
\item \underline{Variance Gamma case}\\We have $\Psi^{'}(u)=i\mu-\frac{u-i\beta}{\alpha-iu\beta+\frac{u^2}{2}}$
Clearly $|\Psi{'}(u)|$ is again bounded.
\end{enumerate}
\end{remarque}

In conclusion, we can apply Theorem \ref{thmCO} and we obtain
\begin{eqnarray*}
V(u )_t&=& \exp(iu X_t+(T-t)\Psi(u ))\ ,\\
H(u )_t&=&\exp\left((T-t)\Psi(u )+\eta(u ,T)-\eta(u ,t)\right)\,e^{iuX_t}\ ,\\
 \xi(u )_t&=&H_t(u )i\frac{\Psi^{'}(u )-\Psi^{'}(0)}{\Psi^{''}(0)}\ .
\end{eqnarray*}

%

\subsection{Representation of some contingent claims by Fourier transforms}
\label{RSCFT}

In general, it is not possible to find a Fourier representation, of the form~(\ref{ES1}), for a given payoff function which is not necessarily bounded or integrable. Hence, it can be more convenient to use the bilateral Laplace transform that allows an extended domain of definition including non integrable functions. We refer to~\cite{Cramer},~\cite{Raible98} and more recently~\cite{Eberlein} for such characterizations of payoff functions. This will be done in the next section. 
%
However, to illustrate the results of this section restricted to payoff functions represented as classical Fourier transforms, we give here two simple examples of such representation extracted from~\cite{Eberlein}:  

\begin{enumerate}
	\item A variant of the digital option is the so-called \textit{asset-or-nothing digital}, where the option holder receives one unit of the \textit{asset}, instead of \textit{currency}, depending on wether the underlying reaches some barrier or not. Hence, the payoff of the asset-or-nothing digital put with barrier is 
$$
f(x)=e^x\mathbf{1}_{e^x<B}  \quad \textrm{and}\quad \hat{f}(u)=\int_{\mathbb{R}} e^{iux} f(x)\,dx=\frac{B^{1+iu}}{1+iu}\ .
$$
	\item The payoff of a {\it self quanto put option} with strike $K$ is 
$$
f(x)=e^x(K-e^x)_+\quad \textrm{and}\quad \hat{f}(u)=\int_{\mathbb{R}} e^{iux} f(x)\,dx=\frac{K^{2+iu}}{(1+iu)(2+iu)}\ .
$$
\end{enumerate}
In both cases the measure $\mu$ is finite.

\section{Föllmer Schweizer decomposition for exponential of PII processes}
\label{sec:expPII}

\setcounter{equation}{0}

In this section, we consisder the case of exponential of PII
corresponding to geometric models (such as the Black-Scholes model) much
more used in finance than arithmetic models (such as the Bachelier
model). The aim of this section is to generalize the results
of~\cite{Ka06} to the case of PII with possibly non stationary increments. Here again, this generalization is motivated by applications to energy derivatives where forward prices show a volatility term structure that requires the use of models with non stationary increments.

\subsection{A reference variance measure}
 We come back to the main optimization problem which was formulated in 
Section \ref{Generalities}. We assume that the process $S$ is the 
discounted price of the non-dividend paying stock which is supposed to be of the form,
\begin{center}
 $S_t=s_0\exp(X_t)\ ,\quad \textrm{for all}\ t\in[0,T]\ ,$
\end{center}
where $s_0$ is a strictly positive constant and $X$ is a
semimartingale process with independent increments (PII),
 in the sense of Definition~\ref{defPAI},  but not necessarily
 with stationary increments.

\begin{remarque}\label{R1}
Let $\gamma \in \R^\ast$,
\begin{enumerate}
\item $\mathbb{E}[\exp(\gamma (X_t-X_s))]>0$, since $X_t-X_s > - \infty$ a.s.
\item $\exp(\gamma(X_t-X_s))$ has a strictly positive variance if
 $(X_t-X_s)$ is non-deterministic.
\end{enumerate}
\end{remarque}
We introduce a new function that will be useful in the sequel. 
\begin{defi}\label{defi:rho} 
For any $t\in [0,T]$, let $\rho_t$ denote the complex valued function such that for all $z,y \in D$
\begin{equation}
\label{eq:rho:zy}
\rho_t(z,y)=\kappa_t(z+y)-\kappa_t(z)-\kappa_t(y)\ .
\end{equation}
For all $z\in D$, then $\bar z\in D$ and $\rho_t(z,\bar z)$ is well defined. To shorten notations $\rho_t$ will also denote the real valued function defined on $D$ such that, 
\begin{equation}
\label{L1}
\rho_t(z)= \rho_t(z,\bar z)=\kappa_t(2Re(z))-2Re(\kappa_t(z))\ . 
\end{equation}
Notice that the last equality results from Remark~\ref{remarkR2}. 
\end{defi}
%
An important technical lemma follows below.
\begin{lemme}\label{L2}
Let $z\in D$, with $z \neq 0$, then, 
$t\mapsto \rho_t(z)$ is strictly increasing if and only if $X$ has no deterministic increments.
\end{lemme}
\begin{proof}
It is enough to show that $X$ has no deterministic increment if and only if for any $0\leq s<t\leq T$, the following quantity is positive, 
\begin{equation}
\label{eq1}
\rho_t(z)-\rho_s(z)=\big [\kappa_t\big(2Re(z)\big)-\kappa_s\big(2Re(z)\big)\big ]-2Re\big (\kappa_t(z)-\kappa_s(z)\big ) \ .
\end{equation}
By Remark~\ref{R3}, for all $z\in D$, we have
$$
\exp[\kappa_t(z)-\kappa_s(z)]=\mathbb{E}[\exp(z\Delta_s^t )]\ ,\quad \textrm{where}\  \Delta_s^t=X_t-X_s\ .
$$
Applying this property and Remark~\ref{remarkR2} 1.,
 to the exponential of the first term on the right-hand side of \eqref{eq1} yields
\begin{eqnarray*}
\exp\left [\kappa_t\big(2Re(z)\big)-\kappa_s\big(2Re(z)\big)\right ]
&=&\mathbb{E}[\exp(2Re(z)\Delta_s^t )]
= 
\mathbb{E}[\exp((z+\bar z)\Delta_s^t )]\\ 
&=&
\mathbb{E}[\left|\exp(z\Delta_s^t )\right|^2]\ .
\end{eqnarray*}
Similarly, for the exponential of the second term on the right-hand 
side difference of
\eqref{eq1}, one gets
\begin{eqnarray*}
\exp\left [2Re\big (\kappa_t(z)-\kappa_s(z)\big )\right ]
&=&
\exp\left [\big(\kappa_t(z)-\kappa_s(z)\big)+\overline{\big(\kappa_t(z)-\kappa_s(z)\big)}\right]
=
\left|\mathbb{E}[\exp(z\Delta_s^t )]\right|^2\ .
\end{eqnarray*}
Hence taking the exponential of $\rho_t(z)-\rho_s(z)$ yields
\begin{eqnarray}
\label{E33}
\exp[\rho_t(z)-\rho_s(z)]-1
&=&\frac{\mathbb{E}[\left|\exp(z\Delta_s^t )\right|^2]}{\left|\mathbb{E}[\exp(z\Delta_s^t )]\right|^2}-1\ ,
\nonumber \\ \nonumber \\
&=&\frac{\mathbb{E}[\left|\Gamma_s^t(z)\right|^2]}{\left|\mathbb{E}[\Gamma_s^t(z)]\right|^2}-1\ ,\quad \textrm{where} \ \Gamma_s^t(z)=\exp(z\Delta_s^t )\ ,
\nonumber \\ \nonumber \\
&=&\frac{Var\left [Re\big (\Gamma_s^t(z)\big )\right ]+Var\left [Im\big (\Gamma_s^t(z)\big )\right ]}{\left|\mathbb{E}[\Gamma_s^t(z)]\right|^2}\ .
\end{eqnarray}
\begin{itemize}
\item If $X$ has a deterministic increment $\Delta_s^t =X_t-X_s$, then $\Gamma_s^t(z)$ is again deterministic and~ \eqref{E33} vanishes and hence  $t\rightarrow \rho_t(z)$ is not strictly increasing.
\item If $X$ has never deterministic increments, then the nominator is never zero, otherwise $Re\big (\Gamma_s^t(z)\big )$, $Im\big (\Gamma_s^t(z)\big )$ and therefore $\Gamma_s^t(z)$ would be deterministic.
\end{itemize}
\end{proof}

From now on, we will always suppose the following assumption.
\begin{Hyp}\label{HypD}
\begin{enumerate}
\item $(X_t)$ has no deterministic increments.
\item $2 \in D$. 
\end{enumerate}
\end{Hyp}

\begin{remarque}\label{33bis}
\begin{enumerate}
\item In particular for $\gamma \in D$, $\gamma \neq 0$, the function $t \mapsto \rho_t(\gamma)$ 
 is strictly increasing.
\item If $z=1$,~(\ref{E33}) equals ${\displaystyle 
\frac{Var\big(\exp(\Delta_s^t )\big)}{\big(\mathbb{E}[\exp(\Delta_s^t) ]
\big)^2}}$, which is a mean-variance quantity.
\end{enumerate}
\end{remarque}
%

%
We continue with a simple observation.
\begin{lemme}\label{lemme38a}
Let $I$ be a compact real interval included in $D$.
\begin{eqnarray*}\sup_{\gamma \in I}
\sup_{t\leq T}\mathbb{E}[S_t^{\gamma}]<\infty\ .\end{eqnarray*}
\end{lemme}
\begin{proof}
Let $t \in [0,T]$ and $\gamma\in I$,  we have
\begin{eqnarray*}\mathbb{E}[S_t^{\gamma}]=
s_0^{\gamma}\exp\{\kappa_t(\gamma)\}\leq 
\max(1,s_0^{\sup I})\exp(\sup_{t\leq T,\gamma \in I} \vert \kappa_t(\gamma)
\vert)\ .
\end{eqnarray*}
since $\kappa$ is continuous.
\end{proof}
%
%

We state now a result that will help us to show that $\kappa_{dt}(z)$ is absolutely continuous with respect to $\rho_{dt}(1)=\kappa_{dt}(2)-2\kappa_{dt}(1)$.
\begin{lemme}\label{propoRusso}
We consider two positive finite non-atomic Borel measures on $E \subset \mathbb{R}^n$, $\mu$ and $\nu$. We suppose the following:
\begin{enumerate}
\item $\mu \ll  \nu\ ;$
\item $\mu(I)\neq 0$ for every open ball of $E$.
\end{enumerate} 
Then ${\displaystyle \frac{d\mu}{d\nu}:=h\neq 0}$ $\nu$ a.e. In particular $\mu$ and $\nu$ are equivalent.
\end{lemme}
\begin{proof}
We consider the Borel set
\begin{eqnarray*}
B=\{x\in E|h(x)=0\}\ .
\end{eqnarray*}
We want to prove that $\nu(B)=0$. So we suppose that there exists a
constant $c>0$ such that $\nu(B)=c>0$
and another constant $\epsilon$ such that $0<\epsilon<c$. 
Since $\nu$ is a Radon measure, there are compact subsets 
 $K_{\epsilon}$ and $K_{\frac{\epsilon}{2}}$ of $E$ such that
\begin{eqnarray*}
K_{\epsilon} \subset K_{\frac{\epsilon}{2}} \subset B\quad \textrm{and}\quad \nu(B - K_{\epsilon})<\epsilon\ ,\quad \nu(B - K_{\frac{\epsilon}{2}})<\frac{\epsilon}{2}\ .
\end{eqnarray*} 
Setting $\epsilon=\frac{c}{2}$, we have
\begin{eqnarray*}
\nu(K_{\epsilon})>\frac{c}{2}
\quad \textrm{and}\quad \nu(K_{\frac{\epsilon}{2}})>\frac{3c}{4}\ .
\end{eqnarray*}
By Urysohn lemma, there is a continuous function $\varphi:E\rightarrow \mathbb{R}$ such that, $0 \leq \varphi \leq 1 $ with
\begin{eqnarray*}\varphi=1 \ \textrm{on}\  K_{\epsilon} \quad
 \textrm{and}\quad \varphi=0 
\ \textrm{on}\  K_{\frac{\epsilon}{2}}^c\ .\end{eqnarray*}
Now
\begin{eqnarray*}
\int_E \varphi(x)\nu(dx) \ge \nu(K_{\epsilon})>\frac{c}{2}>0\ .
\end{eqnarray*}
By continuity of $\varphi$ there is an open set $O \subset E$ with
$\varphi(x)>0$ for $x \in O$. Clearly $O \subset K_{\frac{\epsilon}{2}}
\subset B$; since $O$ is relatively compact, it is a countable union of
balls, and so $B$ contains a ball $I$. The fact that $h=0$ on $I$
implies $\mu(I)=0$ and this 
contradicts Hypothesis 2. of the statement.
Hence the result follows.

\end{proof}

\begin{remarque}\label{remarque313}
From now on, in this section, $d\rho_t=\rho_{dt}$ will denote the measure 
\begin{equation}
\label{eq:rho}
 d\rho_t=\rho_{dt}(1)=d(\kappa_t(2)-2\kappa_t(1))\ .
\end{equation}
According to Remark \ref{33bis} 1.,  it is a positive measure which is strictly positive on each interval.
This measure will play a fundamental role.
\end{remarque}

\begin{remarque}\label{remarque3123}
\begin{enumerate}
\item If $E=[0,T]$, then point 2. of Lemma~\ref{propoRusso} becomes $\mu(I)\neq 0$ for every open interval $I\subset [0,T]$.
\item The result holds for every normal metric locally connected space
 $E$, provided  $\nu$ are Radon measures.
\end{enumerate} 
\end{remarque}

\begin{propo}\label{remarque314Bis}
Under Assumption \ref{HypD} 
\begin{equation}
 d(\kappa_t(z))\ll  d\rho_t\ , \quad \textrm{for all} \ 
z \in D\label{AncHyp1}\ .
\end{equation}
\end{propo}
\begin{proof}
We apply Lemma~\ref{propoRusso}, with $d\mu=d\rho_t$ and $d\nu=da_t$. 
 Indeed, Corollary~\ref{corroR1} implies Condition 1. of Lemma~\ref{propoRusso} and Lemma~\ref{L2} implies Condition 2. of Lemma~\ref{propoRusso}. Therefore, $da_t$ 
is equivalent to $d\rho_t$.
\end{proof}
\begin{remarque}\label{R314bis}
Notice that this result also holds  with $d\rho_t(y)$ instead of 
$d\rho_t=d\rho_t(1)$, for any $y\in D$ such that $Re(y) \neq 0$. 
\end{remarque}

\subsection{On some semimartingale decompositions and covariations} \label{OSSDAC}
%
\begin{propo}\label{corro36}
 Let $y,z \in D$ such that $y+z,2Re(y),Re(y)+1,2Re(z)$ and $Re(z)+1 \in D$.
 Then $S^z$ is a special semimartingale whose canonical decomposition
 $S_t^z =
M(z)_t+A(z)_t$ satisfies
\begin{equation}\label{LL2}
A(z)_t=\int_0^tS_{u-}^z\kappa_{du}(z)\ ,
\quad 
\left\langle M(y),M(z)\right\rangle_t= 
\int_0^tS_{u-}^{y+z}\rho_{du}(z,y)\ , \quad
M(z)_0 = s_0^z,
\end{equation}
where  $d\rho_u(z)$ is defined by equation~(\ref{L1}). 
In particular we have the following: 
\begin{enumerate} 
\item \label{MAexpPII} 
$ \left\langle M(z),M\right\rangle_t=\int_0^tS_{u-}^{z+1}\rho_{du}(z,1)$
\item $\left\langle M(z),M(\bar z)
\right\rangle_t=\int_0^tS_{u-}^{2Re(z)}\rho_{du}(z)\ .$ \label{L22}
\end{enumerate}

\end{propo}
\begin{proof}
For simplicity, we will only treat  the case $y=1$ in \eqref{LL2},
i.e. statement \ref{MAexpPII}.
The general case will follow similarly. 
By Remark \ref{R3}, $N(z)_t:=e^{-\kappa_t(z)}S_t^z$ is a martingale.
 Integration by parts yields
$$
S_t^z=e^{\kappa_t(z)}N(z)_t=
M(z)_t+A(z)_t \quad \textrm{with} \quad
M_0(z) = s_0^z, \quad A(z)_t=\int_0^tS_{u-}^z\kappa_{du}(z) 
\quad \textrm{and}
$$
\begin{eqnarray*}
[M(z),M]_t&=&[S^z,S]_t\ ,\\ \\
&=&S^z_tS^1_t-S^z_0S^1_0-\int_0^t S^z_{s-}dS^1_s \int_0^tS^1_{s-}dS^z_{s}\ ,\\ \\
&=&S^z_tS^1_t-S^z_0S^1_0-\int_0^t S^z_{s-}dM_s-\int_0^t S^z_{s-}dA_s-
\int_0^tS^1_{s-} dM(z)_{s}-\int_0^tS^1_{s-}dA(z)_{s}\ ,\\ \\
&=&S^{z+1}_t-S^{z+1}_0-\int_0^t S^z_{s-}dM_s-\int_0^tS^1_{s-} dM(z)_{s}-
\int_0^t S^{z+1}_{s-}\kappa_{ds}(1)-\int_0^tS^{z+1}_{s-}\kappa_{ds}(z)\ , \\ \\
&=&M(z+1)_t-\int_0^t S^z_{s-}dM_s-\int_0^tS_{s-} dM(z)_{s}+
\int_0^t(\kappa_{ds}(z+1)-\kappa_{ds}(z)-\kappa_{ds}(1))S^{z+1}_{s-}\ .
\end{eqnarray*}
Note that the first three terms on the right-hand side are local martingales. Since $\left\langle M(z),M\right\rangle_t$ is the predictable part of finite variation of the special semimartingale $M(z)M$, equation~(\ref{MAexpPII}) follows.   

\end{proof}
\begin{remarque}\label{remarque312a}
Lemma~\ref{lemme38a} implies that $\mathbb{E}\left[|\left\langle M(y),M(z)\right\rangle\right|]<\infty$ and so $M(z)$ is a square integrable martingale for any $z \in D$ such that $2Re(z),Re(z)+1\in D$.
\end{remarque}


\subsection{On the Structure Condition}
If we apply Proposition~\ref{corro36} with $y=z=1$, we obtain
$S = M+A$ where $M$ is a martingale and
\begin{eqnarray}
A_t=\int_0^tS_{u-}\kappa_{du}(1)\ ,
\end{eqnarray}
and
\begin{eqnarray}\left\langle M,M\right\rangle_t=\int_0^tS_{u-}^{2}(\kappa_{du}(2)-2\kappa_{du}(1))=\int_0^tS_{u-}^{2}\rho_{du}\ .
\label{MMCrochetObl}
\end{eqnarray}

At this point, the aim is to exhibit a predictable $\mathbb{R}$-valued process $\alpha$ such that 
\begin{enumerate}
 \item $A_t=\int_0^t \alpha_s d\left\langle M\right\rangle_s\ ;$
 \item $K_T=\int_0^T \alpha^2_s d\left\langle M\right\rangle_s$ is bounded. 
\end{enumerate}
In that case, according Theorem \ref{ThmExistenceFS}, there will exist a unique FS decomposition for any $H \in \mathcal{L}^2$ and so the minimization problem~(\ref{problem2}) will have a unique solution, by Theorem \ref{ThmSolutionPb1}.

\begin{propo}\label{propoabsconti} Under Assumption \ref{HypD},
we have 
\begin{eqnarray}A_t=\int_0^t \alpha_s d\left\langle M\right\rangle_s\ ,\end{eqnarray}
where $\alpha$ is given by
\begin{eqnarray}\alpha_u:=\frac{\lambda_u}{S_{u-}} \quad \textrm{with}\quad \lambda_u:=\frac{d\kappa_{u}(1)}{d\rho_u}\ ,\quad \textrm{for all}\ u\in [0,T]\ .\label{alphau}\end{eqnarray}
Moreover the MVT process is given by 
\begin{eqnarray}K_t=\int_0^t \left(\frac{d(\kappa_{u}(1))}{d\rho_u}\right)^2d\rho_u\ .\label{KPII}\end{eqnarray}
\end{propo}

\begin{corro} Under Assumption \ref{HypD},
the structure condition (SC) is verified if and only if 
\begin{eqnarray*}K_T=\int_0^T \left(\frac{d(\kappa_{u}(1))}{d\rho_u}\right)^2d\rho_u<\infty\ .\end{eqnarray*}
In particular, $(K_t)$ is deterministic therefore bounded.
\end{corro}
\begin{proof}[\textbf{Proof of Proposition \ref{propoabsconti}}]
By Proposition~\ref{remarque314Bis},  $d\kappa_{t}(1)$ is absolutely continuous with respect to $d\rho_t$. Setting $\alpha_u$ as in~(\ref{alphau}), relation
~(\ref{KPII}) follows from Proposition~\ref{corro36}, expressing $K_t=\int_0^t\alpha_u^2d\left\langle M\right\rangle_u$. 
\end{proof}

\begin{lemme}\label{lemmeTheta}
The space $\Theta$ is constituted by all predictable processes $v$ such that
\begin{eqnarray*}\mathbb{E}[\int_0^Tv_t^2S_{t-}^{2}d\rho_t]<\infty\ . \end{eqnarray*}
\end{lemme}
\begin{proof}
According to Proposition \ref{lemmaTheta}, the fact that $K$ is bounded and $S$ satisfies (SC), then $v \in \Theta$ holds if and only if $v$ is predictable and $\mathbb{E}[\int_0^Tv^2_td\left\langle M,M\right\rangle_t]<\infty$. Since
\begin{eqnarray*}\left\langle M,M\right\rangle_t=\int_0^tS_{s-}^{2}d\rho_s\ ,\end{eqnarray*}
the assertion follows.
\end{proof}

\subsection{Explicit Föllmer-Schweizer decomposition}

%
We denote by $\shd$ the set of $z \in D$ such that
\begin{equation} \label{Eqshd}
   \int_0^T\left \vert \frac{d\kappa_u(z)}{d\rho_u}\right \vert^2
d\rho_u<\infty.
\end{equation}
From now on, we formulate another assumption which will be  in
force for the whole section.
\begin{Hyp}\label{Hyp1bis}
$1 \in \shd$.
\end{Hyp}


\begin{remarque}
\begin{enumerate}
\item 
Because of Proposition~\ref{remarque314Bis}, ${\displaystyle 
\frac{d\kappa_t(z)}{d\rho_t}}$ exists for every $z \in D$.
\item Assumption \ref{Hyp1bis} implies that $K$ is uniformly bounded.
\end{enumerate}
\end{remarque}
The proposition below will constitute an important step for determining
the FS decomposition of the contingent claim $H=f(S_T)$
for a significant class of functions $f$, see 
Section \ref{SSDSCC}.
\begin{propo}\label{lemme38} 
Let $z \in \shd$ with $ z+1  \in \shd$
and $2Re(z) \in D$.

\begin{enumerate}
\item 
 $S_T^z \in \mathcal{L}^2
(\Omega,\mathcal{F}_T)$.
\item  We suppose Assumptions \ref{HypD} and \ref{Hyp1bis}
and we define  
\begin{equation} \label{gammaZT}
 \gamma(z,t) :=  \frac{d(\rho_{t}(z,1))}{d\rho_t}, \ t \in [0,T]. 
\end{equation}
 $ \int_0^T \vert \gamma(z,t)\vert^2 \rho_{dt} < \infty$ 
and 
\begin{eqnarray} \label{etaZT}
\eta(z,t) &:=& \kappa_{t}(z)-\int_0^t\gamma(z,s)\kappa_{ds}(1) \nonumber
\\
&& \\
&=& \kappa_{t}(z)-\int_0^t\gamma(z,s) \frac{d\kappa_s(1)}{d\rho_{s}}
 \rho_{ds}   \nonumber
\end{eqnarray}
is well-defined and 
 $\eta(z,\cdot)$ is absolutely
continuous with respect to $\rho_{ds}$ and therefore bounded.
\item Under the same assumptions  
 $H(z)=S_T^z$ admits a FS decomposition 
$H(z)=H(z)_0+\int_0^T\xi(z)_tdS_t+L(z)_T$ where
\begin{eqnarray}
H(z)_t &:= &e^{\int_t^T\eta(z,ds)}S_t^z\ , \label{FS1}  \\
\xi(z)_t&:=&\gamma(z,t)e^{\int_t^T\eta(z,ds)}S_{t-}^{z-1}\ , \label{FS2} \\
L(z)_t&:=&H(z)_t-H(z)_0-\int_0^t\xi(z)_udS_u\ . \label{FS3}
\end{eqnarray}
\end{enumerate}
\end{propo}


\begin{proof}
\begin{enumerate}
\item is a consequence of Lemma \ref{lemme38a}.
\item $\gamma (z, \cdot)$ is square integrable because Assumption
\ref{Hyp1bis} and $z, z+1 \in \shd$. 
Moreover $\eta$ is well-defined  since
\begin{equation} \label{FS10}
 \left(\int_0^T \vert \gamma(z,s) \vert
\left \vert \frac{d\kappa_s(1)}{d\rho_{s}} \right \vert \rho_{ds} \right)^2 \le 
 \int_0^T \vert \gamma(z,s)\vert^2 \rho_{ds}
\int_0^T \vert \frac{d\kappa_s(1)}{d\rho_{s}} \vert^2 \rho_{ds} .
\end{equation}
\item
In order to prove that~(\ref{FS1}),(\ref{FS2}) and~(\ref{FS3})
constitute
 the FS decomposition of $H(z)$, taking into account Remark \ref{R126}
 we need to show that
\begin{enumerate}
\item $H(z)_0$ is $\mathcal{F}_0$-measurable,
\item $\left\langle L(z),M\right\rangle=0,$
\item $\xi(z)\in\Theta,$
\item $L(z)$ is a square integrable martingale.
\end{enumerate}
Point (a) is obvious. 
Partial integration and point \ref{MAexpPII} of Proposition \ref{corro36}
yield
\begin{eqnarray*}
H(z)_t&=&H(z)_0+\int_0^te^{\int_u^T\eta(z,ds)}dS_u^z+\int_0^tS_u^z
d(e^{\int_u^T\eta(z,ds)}) \\ 
&=&H(z)_0+\int_0^te^{\int_u^T\eta(z,ds)}dM(z)_{u}+\int_0^t
e^{\int_u^T\eta(z,ds)}dA(z)_{u}+\int_0^tS_u^zd(e^{\int_u^T\eta(z,ds)})\\ 
&=&H(z)_0+\int_0^te^{\int_u^T\eta(z,ds)}dM(z)_{u}+\int_0^te^{\int_u^T\eta(z,ds)}dA(z)_{u}-\int_0^t e^{\int_u^T\eta(z,ds)}S_u^z \eta(z,du)  \\ 
&=&H(z)_0+\int_0^te^{\int_u^T\eta(z,ds)}dM(z)_{u}-\int_0^t 
e^{\int_u^T\eta(z,ds)}S_u^z \eta(z,du)
+\int_0^te^{\int_u^T\eta(z,ds)}S_{u-}^z\kappa_{du}(z)  \ .
\end{eqnarray*}
On the other hand
\begin{eqnarray*}
\int_0^t\xi(z)_udS_u&=&\int_0^t\xi(z)_udM_u+\int_0^t\xi(z)_udA_u\ ,\\ \\
&=&\int_0^t\xi(z)_udM_u+\int_0^t\xi(z)_uS_{u-}\kappa_{du}(1)\ ,\\ \\
&=&\int_0^t\xi(z)_udM_u+\int_0^t\gamma(z,u)e^{\int_u^T\eta(z,ds)}S_{u-}^{z}\kappa_{du}(1)\ .
\end{eqnarray*}
Hence, 
\begin{eqnarray*}
L(z)_t&=&H(z)_t-H(z)_0-\int_0^t\xi(z)_udS_u\ ,\\ \\
&=&\int_0^te^{\int_u^T\eta(z,ds)}dM(z)_{u}-\int_0^t
  e^{\int_u^T\eta(z,ds)}S_u^z \eta(z,du)
+\int_0^te^{\int_u^T\eta(z,ds)}S_{u-}^z\kappa_{du}(z)\\ \\
&-&\int_0^t\xi(z)_udM_u-\int_0^t\gamma(z,u)
e^{\int_u^T\eta(z,ds)}S_{u-}^{z}\kappa_{du}(1)\ ,\\ \\
&=&\int_0^te^{\int_u^T\eta(z,ds)}dM(z)_{u}-\int_0^t\xi(z)_udM_u\\ \\
&+& \int_0^t
e^{\int_u^T\eta(z,ds)}S_{u-}^z  
 [\kappa_{du}(z)-\eta(z,du)-\gamma(z,u)\kappa_{du}(1)] .
\end{eqnarray*}
Then, by definition of $\eta$ in (\ref{etaZT}), 
$
\eta(z,du)=\kappa_{du}(z)-\gamma(z,u)\kappa_{du}(1)\ ,
$ 
hence, 
\begin{eqnarray}
L(z)_t=\int_0^te^{\int_u^T\eta(z,ds)}dM(z)_{u}-\int_0^t\xi(z)_udM_u\ ,
\label{43}\end{eqnarray}
which implies that $L(z)$ is a local martingale.\\
From point \ref{MAexpPII}  of Proposition \ref{corro36},
 it follows that
\begin{eqnarray*}
\left\langle L(z),M\right\rangle_t
&=&\int_0^te^{\int_u^T\eta(z,ds)}d\left\langle M(z),M\right\rangle_u-\int_0^t\xi(z)_ud\left\langle M,M\right\rangle_u\ ,\\ \\
&=&\int_0^te^{\int_u^T\eta(z,ds)}S_{u-}^{z+1} \rho_{du}(z,1)-\int_0^t\xi(z)_uS_{u-}^{2}\rho_{du}\ , \\ \\
&=&\int_0^te^{\int_u^T\eta(z,ds)}S_{u-}^{z+1}\rho_{du}(z,1)-\int_0^t\gamma(z,u)e^{\int_u^T\eta(z,ds)}S_{u-}^{z+1}\rho_{du}\ .
\end{eqnarray*}
Consequently,
\begin{eqnarray*}\left\langle L(z),M\right\rangle_t=\int_0^te^{\int_u^T\eta(z,ds)}S_{u-}^{z+1}[\rho_{du}(z,1)-\gamma(z,u)\rho_{du}]\ .\end{eqnarray*}
Then by definition of $\gamma$ in (\ref{gammaZT}), 
$\rho_{dt}(z,1)=\gamma(z,t)\rho_{dt}\ ,$
which yields,  
\begin{eqnarray}
\left\langle L(z),M\right\rangle_t=0\label{51Bis}\ .
\end{eqnarray}
Consequently, point (b) follows.
To continue the proof of this proposition we need the lemma below.
\end{enumerate}

\begin{lemme}\label{lemmepassagecomplexe}
For all $z \in \mathbb{C}$ as in Proposition  \ref{lemme38},
 $d\rho_t$ a.e.   we have
\begin{enumerate}
\item $\overline{\gamma(z,t)}=\gamma(\bar{z},t)\,;$
\item $\overline{\eta(z,t)}=\eta(\bar{z},t)\,.$
\end{enumerate}
\end{lemme}
\begin{proof}
Using Remark \ref{remarkR2} 1) 
we observe $\bar z, \bar z + 1 \in \shd$.
\begin{enumerate}
	\item  By definition of $\gamma$ in (\ref{gammaZT}),
 $\gamma(z,t)\rho_{dt}=\rho_{dt}(z,1)\ .$
 Then, taking the complex conjugate of the integral
from $0$ to $t$
and using Remark~\ref{remarkR2}.1 yields, 
\begin{eqnarray*}
\overline{\int_0^t \gamma(z,s)\rho_{ds}}&=&\int_0^t \overline{\gamma(z,s)}\rho_{ds}\ ,\\
&=&\overline{\rho_t(z,1)}=\overline{\kappa_t(z+1)-\kappa_t(z)-\kappa_t(1)}\ ,\\
&=&
=\kappa_t(\bar{z}+1)-\kappa_t(\bar{z})-\kappa_t(1)  =
 \rho_t(\bar z,1) \ ,\\ 
&=&\int_0^t \gamma(\bar{z},s)\rho_{ds}\ .
\end{eqnarray*}
\item  By definition of $\eta$ in (\ref{etaZT}),
$
\eta(z,t)=\kappa_{t}(z)-\int_0^t\gamma(z,u)\kappa_{du}(1)\ ,
$
so taking the complex conjugate,
\begin{eqnarray*}\overline{\eta(z,t)}&=&\kappa_{t}(\bar{z})-\int_0^t\overline{\gamma(z,s)}\kappa_{ds}(1)\ ,\\ 
&=&\kappa_{t}(\bar{z})-\int_0^t\gamma(\bar{z},s)\kappa_{ds}(1)\ ,\\ 
&=&\eta(\bar{z},t)\ .\end{eqnarray*}
\end{enumerate}
\end{proof}
We continue with the proof of point 3. of Proposition \ref{lemme38}.     
It remains to prove that $L(z)$ is a square-integrable martingale for all $z \in D$ and that $Re(\xi(z))$ and $Im(\xi(z))$ are in $\Theta$. 
(\ref{43}) says that
\begin{eqnarray*}
L(z)_t=\int_0^te^{\int_s^T\eta(z,du)}dM_s(z)-\int_0^t\xi(z)_sdM_s\ .
\end{eqnarray*}
By Proposition~\ref{corro36} and Lemma~\ref{lemmepassagecomplexe}, it follows
\begin{eqnarray}
\overline{L(z)_t}=L(\bar{z})_t\label{44bis}\ ,
\end{eqnarray}
hence, 
\begin{eqnarray} \label{44ter}
\left\langle L(z),\overline{L(z)}\right\rangle_t
&=&\left\langle L(z),L(\bar{z})\right\rangle_t, \nonumber \\ 
&=&\left\langle
L(z),\int_0^{.}e^{\int_s^T\eta(\bar{z},du)}dM_s(\bar{z})
\right\rangle_t\ ,\nonumber \\ 
&& \\
&=&\int_0^te^{\int_s^T\eta(z,du)}e^{\int_s^T\eta(\bar{z},du)}
d\left\langle M(z),M(\bar{z})\right\rangle_s\nonumber \\ 
&&-\int_0^t\xi(z)_se^{\int_s^T\eta(\bar{z},du)}d\left\langle 
M,M(\bar{z})\right\rangle_s\ . \nonumber
\end{eqnarray}
By Proposition~\ref{corro36} we have 
\begin{eqnarray*}
\left\langle L(z),\overline{L(z)}\right\rangle_t
&=&\int_0^te^{\int_s^T\eta(z,du)}e^{\int_s^T\eta(\bar{z},du)}S_{s-}^{2Re(z)}\rho_{ds}(z)\\ \\
%
&&-\int_0^t\xi(z)_se^{\int_s^T\eta(\bar{z},du)}S_{s-}^{1+\bar{z}}\rho_{ds}(\bar{z},1)\ .
\end{eqnarray*}
Using Lemma~\ref{lemmepassagecomplexe} and expressions~ \eqref{gammaZT},
and \eqref{FS2} of $\gamma(z,s)$ and $\xi(z)_s$, we have
\begin{eqnarray*}
\left\langle L(z),\overline{L(z)}\right\rangle_t&=&\int_0^te^{\int_s^T2Re(\eta(z,du))}S_{s-}^{2Re(z)}\rho_{ds}(z)-\int_0^t\xi(z)_se^{\int_s^T\eta(\bar{z},du)}S_{s-}^{1+\bar{z}}\,\gamma(\bar z,s)\rho_{ds}(1)\ ,\\ \\
%
&=&\int_0^te^{\int_s^T2Re(\eta(z,du))}S_{s-}^{2Re(z)}\rho_{ds}(z) -\int_0^t\gamma(z,s)e^{\int_s^T2Re(\eta(z,du))}S_{s-}^{2Re(z)} \gamma(\bar z,s)\rho_{ds}\ ,\\\\
&=&\int_0^te^{\int_s^T2Re(\eta(z,du))}S_{s-}^{2Re(z)}\rho_{ds}(z)
-\int_0^te^{\int_s^T2Re(\eta(z,du))}S_{s-}^{2Re(z)} \vert \gamma(
z,s)\vert ^2\rho_{ds}\ 
.\end{eqnarray*}
Consequently
\begin{eqnarray} \label{Eellz}
\left\langle L(z),\overline{L(z)}\right\rangle_t=
\int_0^te^{\int_s^T 2Re(\eta(z,du)})S_{s-}^{2Re(z)}
[\rho_{ds}(z) -|\gamma(z,s)|^2\rho_{ds}]\ .
\end{eqnarray}
Then, point 2. implies 
\begin{eqnarray} \label{XiXiz}
\int_0^T|\xi(z)_s|^2S_{s-}^2 \rho_{ds}
&=&\int_0^T\xi(z)_s\xi(\bar{z})_sS_{s-}^2 \rho_{ds} \nonumber \\
&=&\int_0^T\gamma(z,s)e^{\int_s^T\eta(z,du)}S_{s-}^{z-1}\gamma(\bar{z},s)
e^{\int_s^T\eta(\bar{z},du)}S_{s-}^{\bar{z}-1}S_{s-}^2\rho_{ds}
 \ , \\
&=&\int_0^T|\gamma(z,s)|^2e^{\int_t^T2Re(\eta(z,du))}S_{s-}^{2Re(z)}\rho_{ds} 
\ . \nonumber\end{eqnarray}


Taking the expectation in~(\ref{XiXiz}), using again point 2.,
 \eqref{gammaZT}, \eqref{etaZT} and
Lemma~\ref{lemme38a}, we obtain
\begin{eqnarray}\mathbb{E}\left[\left\langle L(z),\overline{L(z)}\right\rangle_T\right]<\infty\label{LsquareMart}\ .\end{eqnarray}
Therefore, $L$ is a square-integrable martingale. Similarly,~(\ref{XiXiz}) 
yields that $Re(\xi(z)) \in \Theta$ and $Im(\xi(z)) \in \Theta$. This concludes the proof of Proposition~\ref{lemme38}.
\end{proof}



\subsection{FS decomposition of special contingent claims} 
\label{SSDSCC}

Now, we will proceed to the FS decomposition of more general contingent claims. We consider now options of the type 
\begin{eqnarray}
H=f(S_T) \quad\textrm{with}\quad f(s)=\int_{\mathbb{C}} s^z\Pi(dz)\ ,
\label{FORM}\end{eqnarray}
where $\Pi$ is a (finite) complex measure in the sense of Rudin~\cite{RU87},  Section~6.1. 
An integral representation of some basic European calls can be found later.\\
We need now the new following assumption. 
\begin{Hyp}\label{Hyp1}
Let $I_0= {\rm supp} \Pi\cap \mathbb{R}$. We denote
 $I:=\left[\inf I_0\wedge 2\inf I_0,2\sup I_0 \vee \sup I_0+1\right].$ 
\begin{enumerate}
\item $\forall z \in {\rm supp} \Pi, \quad z, z+1 \in \shd.$
\item $I \subset D$ and
 $\sup_{x \in I \cup \{1\} }\left\|\frac{d(\kappa_t(x))}{d\rho_t}
\right\|_{\infty}<\infty$.
\end{enumerate}
\end{Hyp}



\begin{remarque}\label{remarque320}
\begin{enumerate}
\item Point 2. of Assumption \ref{Hyp1} implies
 $\sup_{z \in I + i \R}
\left\|\kappa_{dt}(Re(z))\right\|_{T}<\infty\,.$
\item Under Assumption \ref{Hyp1}, $H = f(S_T)$ is square integrable.
In particular it admits an FS decomposition.
\item Because of~(\ref{AncHyp1}) in Proposition~\ref{remarque314Bis}, 
the Radon-Nykodim derivative at Point 2. of Assumption 
\ref{Hyp1}, always exists.
\end{enumerate}
\end{remarque}
We need now to obtain upper bounds on $z$ for the quantity~(\ref{LsquareMart}). We will first need the following lemma. 
%







\begin{lemme}\label{lemmeconstante}
There are positive constants $c_1, c_2, c_3$ such that $d\rho_s$ a.e.
 \begin{enumerate}
\item
 $$\sup_{z\in {\rm supp \Pi}}
 \frac{dRe(\eta(z,s))}{d\rho_s} \le c_1. $$
\item For any $z \in {\rm supp} \Pi$
$$ \vert \gamma(z,s) \vert^2 \le    \frac{d\rho_s(z)}{d\rho_s} \le
c_2 - c_3
\frac{dRe(\eta(z,s))}{d\rho_s} $$ 
\item $$- \sup_{z \in {\rm supp \Pi}} \int_0^T 2 Re (\eta(z,dt)) \exp(
 \int_t^T Re (\eta(z,ds)) ) < \infty.$$
\end{enumerate}
\end{lemme}
\begin{remarque}
According to Proposition \ref{lemme38},
$t \mapsto Re (\eta(z,t))$ is absolutely continuous 
with respect to $d\rho_t$.
\end{remarque} 
\textit {Proof} ({\bf  of  Lemma} \ref{lemmeconstante}).\\
The proof is inspired by Lemma 3.9 of \cite{Ka06}.
According to Point 2. of Assumption \ref{Hyp1} 
we denote 
\begin{equation} \label{E300}
 c_{11} :=  
\sup_{x \in I}\left 
\Vert\frac{d(\kappa_t(x))}{d\rho_t}\right\Vert_{\infty}. 
\end{equation}
For $z \in  {\rm supp} \Pi, t\in [0,T]$, we have
\begin{eqnarray*}\eta(z,t)=\kappa_{t}(z)-\int_0^t\gamma(z,s)d\kappa_{s}(1)
\quad \textrm{and}\quad \eta(\bar{z},t)=
\kappa_{t}(\bar{z})-\int_0^t\gamma(\bar{z},s)d\kappa_{s}(1).\end{eqnarray*}
Then, we get
$Re(\eta(z,t))=Re(\kappa_{t}(z))-\int_0^tRe(\gamma(z,s))d\kappa_{s}(1)\ .$
We obtain
\begin{eqnarray}\label{D0}
  \int_t^TRe(\eta(z,ds)) &\leq&  Re\left(\kappa_T(z)-\kappa_t(z)\right)
+ \left \vert \int_t^T \gamma(z,s) d\kappa_s(1) \right \vert  \nonumber\\
&&\\
&=&  \int_t^T \frac{Re (d\kappa_{s}(z))}{d\rho_s} d\rho_s
+ \left \vert \int_t^T \gamma(z,s) d\kappa_s(1) \right \vert \ .\nonumber
\end{eqnarray}
 Since $\left\langle L(z),\overline{L(z)}\right\rangle_t$ is increasing, 
and taking into account \eqref{Eellz}, 
 the measure,  $\left (d\rho_s(z)-|\gamma(z,s)|^2d\rho_s\right)$, is non-negative. It follows that
\begin{eqnarray} \label{D0bis}
\frac{d\rho_s(z)}{d\rho_s}-|\gamma(z,s)|^2\geq 0\ ,\quad d\rho_s \ a.e.
\label{MajorationGamma}\end{eqnarray}

\begin{remarque}\label{lemmadensitepositive}
By Lemma \eqref{MajorationGamma},
in particular  the density ${\displaystyle \frac{d\rho_s(z)}{d\rho_s}}$
 is non-negative $d\rho_s$ a.e.
\bigskip


Consequently,  
\begin{eqnarray} \label{H02}
 2\frac{dRe(\kappa_s(z))}{d\rho_s}\leq 
\frac{d\kappa_s(2Re(z))}{d\rho_s}\ ,
\quad d\rho_s\ a.e. \end{eqnarray}
\end{remarque}




In order to prove 1. it is enough to verify that, for some $c_0 > 0$,
\begin{eqnarray}\frac{dRe(\eta(z,s))}{d\rho_s}\leq c_0+
\frac{1}{2}\frac{dRe(\kappa_s(z))}{d\rho_s} & d\rho_s \ a.e.
 &\label{EB1}\end{eqnarray}

In fact, \eqref{H02} and Assumption \ref{Hyp1} point 2. and \eqref{E300},
 imply that

\begin{eqnarray}\frac{dRe(\eta(z,s))}{d\rho_s}
\leq c_0 + \frac{1}{2} c_{11} =: c_1. 
\label{EB2}\end{eqnarray}

To prove \eqref{EB1} it is enough to show that 
\begin{equation} \label{D0ter}
Re(\eta(z,T)-\eta(z,t))
\leq c_{0}(\rho_T-\rho_t)+\frac{1}{2}Re(\kappa_T(z)-\kappa_t(z)),
 \quad \forall t \in [0,T]. \end{equation}
Again Assumption \ref{Hyp1} point 2. implies that
\begin{equation} \label{EA11}
\left |\int_t^T\gamma(z,s)d\kappa_s(1)\right| \le c_{12} 
\int_t^T \vert \gamma(z,s) \vert d\rho_s\ ,
 \end{equation}
where $c_{12} = \Vert \frac{d\kappa_s(1)}{d\rho_s} \Vert_\infty.$
%
%
Using \eqref{MajorationGamma}, and Assumption \ref{Hyp1} it follows
%
\begin{eqnarray} \label{EC1}
\left|\gamma(z,s)\right|^2&\leq&\frac{d\rho_s(z)}{d\rho_s}=\frac{d\kappa(2Re(z))}{d\rho_s}-\frac{2dRe(\kappa_s(z))}{d\rho_s}\ , \nonumber \\
&& \\
&\leq &c_{11}-\frac{2dRe(\kappa_s(z))}{d\rho_s}\ .\nonumber
  \end{eqnarray}
This implies that
$$ c_{12}^2 \left|\gamma(z,s)\right|^2 \le
 \left(c_{13}^2+\frac{1}{4}
\left (\frac{dRe(\kappa_s(z))}{d\rho_s}\right)^2 \right) \ ,
$$
where $c_{13} > 0$ is chosen such that 
 $c_{13}^2 \ge 4 c_{12}^4 +  c_{12}^2 c_{11}$.
%
Consequently
\begin{eqnarray*}\left|\int_t^T\gamma(z,s)d\kappa_s(1)\right| \le 
\int_t^T d\rho_s \left(c_{13}+\frac{1}{2}
\left|\frac{dRe(\kappa_s(z))}{d\rho_s}\right|\right)\ .
\end{eqnarray*}
Coming back to \eqref{D0}, we obtain
\begin{eqnarray*}
Re(\eta(z,T)-\eta(z,t)) &\le &
\int_t^T \left ( \frac{Re(d\kappa_s(z))}{d\rho_s}  
   + c_{13} + \halb 
\left \vert \frac{Re(d\kappa_s(z))}{d\rho_s} \right \vert \right) d\rho_s  \\
&\le &  
\int_t^T \left ( \halb \frac{Re(d\kappa_s(z))}{d\rho_s}  
 + \left(\frac{Re(d\kappa_s(z))}{d\rho_s} \right)^+ 
 + c_{13} \right)  d\rho_s
\end{eqnarray*}
\eqref{H02} and Assumption \ref{Hyp1} allow to establish
\begin{eqnarray}Re(\eta(z,T)-\eta(z,t))\leq \int_t^T d\rho_s
\left(c_{0}+\frac{1}{2}\frac{dRe(\kappa_s(z))}{d\rho_s}\right)\ ,
\label{EA33}\end{eqnarray}
where $c_0=\frac{c_{11}}{2}+c_{13}$. This concludes the proof of point 1.\\
In order to prove point 2. we first observe that \eqref{EB1} implies
\begin{eqnarray}-\frac{dRe(\kappa_s(z))}{d\rho_s}\leq
 2\left(c_0-\frac{dRe(\eta(z,s))}{d\rho_s}\right)\label{EB3}\end{eqnarray}
$d\rho_s$ a.e. \eqref{EC1} implies
\begin{eqnarray}\left|\gamma(z,s)\right|^2\leq c_{21}-4\frac{dRe(\eta(z,s))}{d\rho_s}\ ,\label{EC2}\end{eqnarray}
where $c_{21}=c_{11}+4c_0$. Point 2. is now established with $c_2=c_{21}$ and $c_3=4$.\\
We continue with the proof of point 3. We decompose
\begin{eqnarray*}Re(\eta(z,t))=A^+(z,t)-A^-(z,t)\ ,\end{eqnarray*}
where
$$
A^+(z,t)=\int_0^t\left(\frac{dRe(\eta(z,s))}{d\rho_s}\right)_+d\rho_s\ , \quad\textrm{and}\quad 
A^-(z,t)=\int_0^t\left(\frac{dRe(\eta(z,s))}{d\rho_s}\right)_-d\rho_s\ .
$$
$A^+(z,.)$ and $A^-(z,.)$ are increasing non negative functions. Moreover 
point 1. implies
\begin{eqnarray*}A^+(z,t)\leq c_1\rho_t\ .\end{eqnarray*}
At this points for $z\in I+i\mathbb{R}$
\begin{eqnarray*}
-\int_0^TRe(\eta(z,dt))e^{\int_t^T2Re(\eta(z,ds))}&=&
\int_0^T\left(A^-(z,dt)- A^+(z,dt)\right)e^{
2\int_t^TRe(\eta(z,ds))}\,\\
&\leq&\int_0^TA^-(z,dt)e^{2\left(A^+(z,T)-A^+(z,t)\right)}e^{-2\left(A^-(z,T)-A^-(z,t)\right)}\,\\
&\leq&e^{2c_1\rho_T}\int_0^Te^{-2\left(A^-(z,T)-A^-(z,t)\right)}A^-(z,dt)\,\\
&=&\frac{e^{2c_1\rho_T}}{2}\left\{1-e^{-2A^-(z,T)}\right\}\,\\
&\leq&\frac{e^{2c_1\rho_T}}{2}\ ,
\end{eqnarray*}
which concludes the proof of point 3 of Lemma \ref{lemmeconstante}.


\bigskip






\qed \\

Let $\gamma=\sup_{z \in I}\left(2Re(z)\right)$, by Lemma~\ref{lemme38a}, 
it follows
\begin{eqnarray}
\sup_{z \in I, s\leq T}\mathbb{E}\left[S^{2Re(z)}_s\right]
<\infty\ .
\label{EspStBornee}\end{eqnarray}



\begin{thm}\label{propo310} Let $\Pi$ be a finite complex-valued 
Borel  measure on $\C$.\\
Suppose Assumptions \ref{HypD},
  \ref{Hyp1bis}, \ref{Hyp1}.
Any complex-valued contingent claim $H=f(S_T)$, where
 $f$ is of the form~(\ref{FORM}), 
and $H \in \shl^2$,
 admits a unique FS decomposition $H=H_0+\int_0^T\xi_tdS_t+L_T$ with the
 following
properties.
\begin{enumerate} 
\item 
 $H \in \shl^2$ and
\begin{itemize}
\item $H_t=\int H(z)_t\Pi(dz)\,,$
\item $\xi_t=\int \xi(z)_t\Pi(dz)\,,$
\item $L_t=\int L(z)_t\Pi(dz)\,,$ 
\end{itemize}
where for $z \in {\rm supp} (\Pi)$,  $H(z),\xi(z)$ and $L(z)$ 
are the same as those introduced in 
Proposition~\ref{lemme38}  and we convene that they vanish if 
$z \notin   {\rm supp} (\Pi)$.
\item Previous decomposition is real-valued
if $f$ is real-valued.
\end{enumerate}
\end{thm}
\begin{remarque} \label{Rsimbis}
Taking $\Pi = \delta_{z_0}(dz), \ z_0 \in \C$,
Assumption \ref{Hyp1} is equivalent to the assumptions of Proposition 
\ref{lemme38}.
\end{remarque}

\begin{proof}
\begin{description}
\item{a)} 
$f(S_T) \in \shl^2$ since by Jensen,
$$E \left \vert \int_\C \Pi(dz) S_T^z \right \vert ^2
  \le \int_\C \vert \Pi\vert(dz)
 E \vert S_T^{2Rez}\vert  \vert \Pi\vert(\C)
 \le \sup_{x \in I_0} E( S_T^{2x}) \vert \Pi\vert(\C)^2,$$
where $|\Pi|$ denotes the total variation of the finite measure $\Pi$.
Previous quantity is bounded because of Lemma \ref{lemme38}. \\
We go on with the FS decomposition.
 We would like to prove first that $H$ and $L$ are well defined
 square-integrable processes and 
$E(\int_0^T \vert \xi_s\vert^2 d\langle M\rangle_s) < \infty$.\\
 We denote  $K=supp(\Pi)$. 
By Jensen's inequality, we have
\begin{eqnarray*}\mathbb{E} \left \vert \int_\mathbb{C}
    L(z)_t\Pi(dz)\right \vert^2]\leq 
\mathbb{E} \left ( \int_\mathbb{C}|\Pi|(dz) |L_t(z)|^2_t \right )
\vert \Pi (\C) \vert 
=\int_\mathbb{C}|\Pi|(dz)\mathbb{E}[|L_t(z)|^2_t] \vert \Pi \vert \,
\end{eqnarray*}
Similar calculations allow to show that
\begin{eqnarray*}\mathbb{E}[\xi_t^2]\leq 
\vert \Pi \vert(\C) \int_\mathbb{C} |\Pi|dz)\mathbb{E}[|\xi_t(z)|^2]\ 
\quad \textrm{and}\quad 
\mathbb{E}[L_t^2]\leq \vert \Pi(\C)\vert
\int_\mathbb{C} |\Pi|(dz)\mathbb{E}[|L_t(z)|^2]\ .
\end{eqnarray*}
We will show now that
\begin{itemize}
\item (A1): $\sup_{t\leq T,z \in I + i \R }\mathbb{E}[|H_t(z)|^2]<\infty\,;$
\item (A2): $ \int_\mathbb{C}|\Pi|(dz)\mathbb{E}[|L_t(z)|^2_t] < \infty;$
\item (A3): 
$$ E\left( \int_0^T d\rho_t S_t^2
\int_\C  \vert \xi_t(z)\vert^2 \vert \Pi \vert(dz) \right)<\infty\,. $$
\end{itemize}
(A1): Since $H(z)_t=e^{\int_t^T\eta(z,ds)}S_t^z$, we have
\begin{eqnarray*}|H(z)_t|^2=H(z)_t\overline{H(z)_t}=e^{\int_t^T2Re(\eta(z,ds))}S_t^{2Re(z)}\ ,\end{eqnarray*}
so
\begin{eqnarray*}\mathbb{E}[|H(z)_t|^2]=e^{\int_t^T2Re(\eta(z,ds))}\mathbb{E}[S_t^{2Re(z)}]\leq e^{\int_t^T2Re(\eta(z,ds))}\sup_{t\leq T}\mathbb{E}[S_t^{\gamma}]\ ,\end{eqnarray*}
with $\gamma=\sup_{z \in I}2Re(z)$. Inequality~(\ref{EspStBornee}) and
Lemma~\ref{lemmeconstante} imply (A1). Therefore $(H_t)$ is a
well-defined square-integrable process.
%

(A2): 
$\mathbb{E}[|L_t(z)|^2]\leq \mathbb{E}[|L_T(z)|^2]=\mathbb{E}[\left\langle L(z),\overline{L(z)}\right\rangle_T]\ ,$
where the first inequality is due to the fact that $|L_t(z)|^2$ is a 
submartingale. 
\begin{eqnarray*}\mathbb{E}\left[\left\langle L(z),
\overline{L(z)}\right\rangle_T\right]=
\mathbb{E}\left[\int_0^Te^{\int_s^T2Re(\eta(z,du)}S_{s-}^{2Re(z)}
\left[d\rho_s(z)-|\gamma(z,s)|^2d\rho_s\right]\right] \ .\end{eqnarray*}
By Fubini, Lemma \ref{lemme38a} and \eqref{Eellz}, we have
\begin{eqnarray*}
\mathbb{E}\left[\left\langle L(z),\overline{L(z)}\right\rangle_T\right]
&=& \int_0^Te^{\int_s^T2Re(\eta(z,du)} \mathbb{E}[S_{s-}^{2Re(z)}]
\left[\frac{d\rho_s(z)}{d\rho_s} -|\gamma(z,s)|^2 \right]
d\rho_s 
\\ &\le &  \int_0^Te^{\int_s^T 2Re(\eta(z,du)}
\left[\frac{d\rho_s(z)}{d\rho_s} -|\gamma(z,s)|^2 \right]
\mathbb{E}[S_{s-}^{2Re(z)}] d\rho_s\\
&\le & c_4 \int_0^Te^{\int_s^T2Re(\eta(z,du)}
\left[\frac{d\rho_s(z)}{d\rho_s} \right]
d\rho_s \ ,
\end{eqnarray*}
where $c_4 = \sup_{s \le T} \mathbb{E}[S_{s}^{2Re(z)}]\,$.

According to Lemma \ref{lemmeconstante} point 2, previous expression is 
bounded by $c_4 I(z)$, where 
\begin{eqnarray}\label{EIZ}
I(z)&:=&\int_0^Td\rho_t\exp\left(\int_t^T2Re(\eta(z,ds))\left[c_2-c_3
\frac{dRe(\eta(z,t))}{d\rho_t}\right]\right)\,\nonumber\\
&&\\
&=& c_2I_1(z)+c_3I_2(z) \ ,\nonumber
\end{eqnarray}
where
$$
I_1(z)=\int_0^Td\rho_t\exp\left(\int_t^T2Re(\eta(z,ds))\right)\,\quad \textrm{and}\quad I_2(z)=\int_0^T\exp\left(\int_t^T2Re(\eta(z,ds))\right)dRe(\eta(z,ds))\ .
$$
Using Lemma \ref{lemmeconstante}, we obtain
\begin{equation} \label{EIZbis}
\sup_{z\in I+i\mathbb{R}}\left|I_1(z)\right|\leq
\rho_T\exp\left(2c_1\rho_T\right)\,\quad \textrm{and}\quad \sup_{z\in I+i\mathbb{R}}\left|I_2(z)\right|<\infty \ ,
\end{equation}
and so
\begin{eqnarray}
\sup_{z\in I+i\mathbb{R}} \mathbb{E}\left[\left\langle L(z),
\overline{L(z)}\right\rangle_T\right]
<\infty\ .
\label{EIZ1}\end{eqnarray}
This concludes (A2).\\
We verify now the validity of (A3). This requires to control
\begin{eqnarray*}\mathbb{E}\left[\int_0^T\rho_{dt}S_t^2\left(\int_{\mathbb{C}}|\Pi|(dz)|\xi(z)_t|^2\right)\right]\leq \mathbb{E}\left[\int_0^T\rho_{dt}S_t^2\left(\int_{\mathbb{C}}|\Pi|(dz)\left|\gamma(z,t)\exp\left(\int_t^TRe(\eta(z,ds))\right)S_t^{z-1}\right|^2\right)\right]\ .\end{eqnarray*}
Using Jensen inequality, this is smaller or equal than  
\begin{eqnarray*}
\vert \Pi(\C) \vert \int_{\mathbb{C}}|\Pi|(dz)\int_0^T\rho_{dt}\mathbb{E}\left[S_t^{2Re(z)}\right]|\gamma(z,t)|^2\exp\left(2\int_t^TRe(\eta(z,ds))\right)\ .
\end{eqnarray*}
Lemma \ref{lemmeconstante} gives the upper bound
$$
\vert \Pi\vert(\C)  \sup_{t\leq T,\gamma \in I}\mathbb{E}
\left[S_t^{2Re(z)}\right]\int_{\mathbb{C}}|\Pi|(dz)I(z)\ ,
$$
where $I(z)$ was defined in \eqref{EIZbis}.
 Since $\Pi$ is finite and because of \eqref{EIZ1}, (A3) is now established.

%


%

We show now that $(L_t)$ is an $(\mathcal{F}_t)$-martingale. Let
 $0\leq s\leq t \leq T$, $\mathcal{B} \in \mathcal{F}_s$. By Proposition
 \ref{lemme38}, since $(L(z)_t)$ is a martingale, we obtain
\begin{eqnarray*}\mathbb{E}[(L_t-L_s)1_\mathcal{B}]=\mathbb{E}[\int_\mathbb{C}(L(z)_t-L(z)_s)\Pi(dz)1_\mathcal{B}]\ .\end{eqnarray*}
By Fubini's theorem we conclude that
\begin{eqnarray*}\mathbb{E}[(L_t-L_s)1_\mathcal{B}]=\int_\mathbb{C}\mathbb{E}[(L(z)_t-L(z)_s)1_\mathcal{B}]\Pi(dz)\ ,\end{eqnarray*}
and $\mathbb{E}[(L(z)_t-L(z)_s)1_\mathcal{B}]=0$. So
\begin{eqnarray*}\mathbb{E}[(L_t-L_s)1_\mathcal{B}]=0\ .\end{eqnarray*}
Hence, $L$ is a square-integrable martingale.\\
 Similarly, it can be shown that
 $\mathbb{E}[(M_tL_t-M_sL_s)1_\mathcal{B}]=0$
 and so ML is a square-integrable martingale as well. 
Hence $L$ is orthogonal to M. 
 By Fubini's theorem for stochastic integrals, cf.~\cite{Pr92},
 Theorem IV.46, we have
\begin{eqnarray*}\int\int_0^t\xi(z)_sdS_s\Pi(dz)=\int_0^t\int\xi(z)_s\Pi(dz)dS_s=\int_0^t\xi_sdS_s\ .\end{eqnarray*}
Consequently,  $(H_0,\xi,L)$ provide a (possibly complexe)
 FS decomposition of $H$.
\item{b)}
It remains to prove that the decomposition is real-valued.
 Let $(H_0,\xi,L)$ and $(\overline{H_0},\overline{\xi},\overline{L})$
 be two FS decomposition of $H$. Consequently,
since $H$ and $(S_t)$ are real-valued, we have
\begin{eqnarray*}0=H-\overline{H}=(H_0-\overline{H}_0)+\int_0^T(\xi_s-\overline{\xi}_s)dS_s+(L_T-\overline{L}_T)\ ,\end{eqnarray*}
which implies that $0=Im(H_0)+\int_0^TIm(\xi_s)dS_s+Im(L_T)$.
 By Theorem \ref{ThmExistenceFS}, the uniqueness of the real-valued Föllmer-Schweizer decomposition yields that the processes $(H_t)$,$(\xi_t)$ and $(L_t)$ are real-valued.
\end{description}
\end{proof}

\subsection{Representation of some typical contingent claims}
\label{sec:payoff:laplace}
We used some integral representations of  payoffs of the
form~(\ref{FORM}). We refer to~\cite{Cramer},~\cite{Raible98} and more
recently~\cite{Eberlein}, for some characterizations of classes of
functions which admit this kind of representation. In order to apply
the results of this paper, we need explicit formulae
 for the complex measure $\Pi$ in some example of contingent claims. 

\subsubsection{Call}
The first example is  the European Call option $H=(S_T-K)_+$.
 We have two representations of the form~(\ref{FORM}) which result
 from the following lemma. 

\begin{lemme}\label{Call}
Let $K>0$, the European Call option $H=(S_T-K)_+$ has two representations of the form~(\ref{FORM}):
\begin{enumerate}
\item For arbitrary $R>1$, $s>0$, we have
\begin{eqnarray}(s-K)_+=\frac{1}{2\pi i}\int_{R-i\infty}^{R+i\infty}s^z\frac{K^{1-z}}{z(z-1)}dz\ .\label{Call1}\end{eqnarray}
\item For arbitrary $0<R<1$, $s>0$, we have
\begin{eqnarray}(s-K)_+-s=\frac{1}{2\pi i}\int_{R-i\infty}^{R+i\infty}s^z\frac{K^{1-z}}{z(z-1)}dz\ .\label{Call2}\end{eqnarray}
\end{enumerate}
\end{lemme}

\subsubsection{Put}

\begin{lemme}\label{Put}
Let $K>0$, the European Put option $H=(K-S_T)_+$ gives for an arbitrary $R<0$, $s>0$
\begin{eqnarray}(K-s)_+=\frac{1}{2\pi i}\int_{R-i\infty}^{R+i\infty}s^z\frac{K^{1-z}}{z(z-1)}dz\ .\label{Put1}\end{eqnarray}
\end{lemme}



%
%

\section{The solution to the minimization problem}

\setcounter{equation}{0}

\subsection{Mean-Variance Hedging}
FS decomposition will help to provide the solution to the global
minimization  problem. Next theorem deals with the case where the 
underlying process is a PII. 
\begin{thm}\label{mainthmPII}
Let $X=(X_t)_{t \in [0,T]}$ be a process with independent increments
with log-characteristic function $\Psi_t$. Let $H=f(X_T)$ where $f$ is
of the form~(\ref{ES1}). We suppose that the PII, $X$, satisfies
 Assumptions \ref{HPAIND}, \ref{HypoFSPII:1}, \ref{HypoFSPII:2}
and \ref{ASupeta}.
Then, the variance-optimal capital $V_0$ and the variance-optimal hedging strategy $\varphi$, solution of the minimization problem $(\ref{problem2})$, are given by
\begin{eqnarray}V_0=H_0\ ,\end{eqnarray}
and the implicit expression
\begin{eqnarray}
\varphi_t=\xi_t+\frac{\lambda_t }{S_{t-}}(H_{t-}-V_0-\int_0^t\varphi_sdS_s)\ ,
\end{eqnarray}
where the processes $(H_t)$,$(\xi_t)$ and $(\lambda_t)$ are defined by
\begin{eqnarray}
H_t=\int_{\mathbb{R}}H(u)_t\mu(du)\ , \quad \xi_t=\int_{\mathbb{R}}
 i\frac{d(\Psi^{'}_t(u )-\Psi^{'}_t(0))}{d\Psi^{''}_t(0)}\,H(u)_t
 \mu(du)
 \quad\textrm{and}\quad 
\lambda_t=i\frac{d\Psi_t'(0)}{d\Psi''_t(0)}\ ,
\end{eqnarray}
and
\begin{eqnarray}
H(u )_t=e^{\eta(u ,T)-\eta(u ,t)+\Psi_T(u )-\Psi_t(u )}\,e^{iu X_{t-}} \quad\textrm{with}\quad 
\eta(u ,t)=i\int_0^t\frac{d\Psi^{'}_t(0)}{d\Psi^{''}_t(0)}d\left(\Psi^{'}_s(u )-\Psi^{'}_s(0)\right)\ .\end{eqnarray}
The optimal initial capital is unique. The optimal hedging strategy $\varphi_t(\omega)$ is unique up to some $(P(d\omega)\otimes dt)$-null set.
\end{thm}
\begin{proof}
Since $K$ is deterministic, the optimality follows from Theorem~\ref{thmCO}, Theorem~\ref{ThmSolutionPb1} and Corollary~\ref{ThmSolutionPb2}. Uniqueness follows from Theorem~\ref{thm128}.
\end{proof}

Next theorem deals with the case where the payoff to hedge is given as a
bilateral Laplace transform of the exponential of a PII. It is an
extension of Theorem~3.3 of~\cite{Ka06} to PII with no stationary increments. 
\begin{thm}\label{mainthm}
Let $X=(X_t)_{t \in [0,T]}$ be a process with independent increments
with cumulant generating function $\kappa$.
 Let $H=f(e^{X_T})$ where $f$ is of the form~(\ref{FORM}). 
We assume the validity of Assumptions~\ref{HypD}, \ref{Hyp1bis}, 
 \ref{Hyp1}.
The variance-optimal capital $V_0$ and the variance-optimal hedging strategy $\varphi$, solution of the minimization problem $(\ref{problem2})$, are given by
\begin{eqnarray}V_0=H_0 \end{eqnarray}
and the implicit expression
\begin{eqnarray}\varphi_t=\xi_t+\frac{\lambda_t }{S_{t-}}(H_{t-}-V_0-\int_0^t\varphi_sdS_s)\ ,\end{eqnarray}
where the processes $(H_t)$, $(\xi_t)$ and $(\lambda_t)$ are defined by
\begin{eqnarray}\gamma(z,t):=\frac{d\rho_{t}(z,1)}{d\rho_t} \quad \textrm{with}\quad \rho_t(z,y)=\kappa_t(z+y)-\kappa_t(z)-\kappa_t(y)\ ,\end{eqnarray}
\begin{eqnarray}\eta(z,dt):=\kappa_{dt}(z)-\gamma(z,t)\kappa_{dt}(1)\ ,\end{eqnarray}
\begin{eqnarray}\lambda_t :=\frac{d(\kappa_{t}(1))}{d\rho_t}\ ,\end{eqnarray}
\begin{eqnarray}H_t:=\int_\C e^{\int_t^T\eta(z,ds)}S_t^z\Pi(dz)\ ,\end{eqnarray}
\begin{eqnarray}\xi_t:=\int_\C \gamma(z,t)e^{\int_t^T\eta(z,ds)}S_{t-}^{z-1}\Pi(dz)\ .\end{eqnarray}
The optimal initial capital is unique. The optimal hedging strategy $\varphi_t(\omega)$ is unique up to some $(P(d\omega)\otimes dt)$-null set.
\end{thm}

\begin{remarque}\label{remarque328}
The mean variance tradeoff process can be expressed as follows,
 see~(\ref{KPII}):
\begin{eqnarray*}K_t=\int_0^t\frac{d\kappa_u(1)}{d\rho_u}\kappa_{du}(1)\ .\end{eqnarray*}
\end{remarque}

\begin{proof}[\textbf{Proof of Theorem~\ref{mainthm}}]
Since $K$ is deterministic, the optimality follows from Theorem~\ref{propo310}, Theorem~\ref{ThmSolutionPb1} and Corollary~\ref{ThmSolutionPb2}. Uniqueness follows from Theorem~\ref{thm128}.

\end{proof}

\subsection{The quadratic error}
Again, $\rho_{dt}$ denotes the measure
 $\kappa_{dt}(2)-2\kappa_{dt}(1)$. Let $V,\varphi$ and $H$ appearing in Theorem \ref{mainthm}. The quantity $\mathbb{E}[(V_0+G_T(\varphi)-H)^2]$ will be called \textbf{the variance of the hedging error}.
\begin{thm}\label{thm34}
Under the assumptions of Theorem \ref{mainthm},
the variance of the hedging error equals
\begin{eqnarray*}J_0:=\left(\int_\mathbb{C}\int_\mathbb{C} J_0(y,z)\Pi(dy)\Pi(dz)\right)\ ,\end{eqnarray*}
where
$$
J_0(y,z):= \left \{
\begin{array}{ccc}
s_0^{y+z}\int_0^T\beta(y,z,dt)e^{\kappa_{t}(y+z)+\alpha(y,z,t)}dt\ &:& y, z
 \in{\rm supp} \Pi \\
0 &:& otherwise.
\end{array}
\right.
$$
and
\begin{eqnarray*}\alpha(y,z,t)&:=&\eta(z,T)-\eta(z,t)-(\eta(y,T)-\eta(y,t))-\int_t^T\left(\frac{d\kappa_{s}(1)}{d\rho_s}\right)^2d\rho_s\ ,\\ 
\beta(y,z,t)&:=&\rho_{t}(y,z)-\int_0^t\gamma(z,s)\rho_{ds}(y,1)
\ .\end{eqnarray*}
\end{thm}

\begin{remarque}\label{remarque329}
We have
\begin{eqnarray*}\alpha(y,z,t)=(\eta(z,T)-\eta(z,t))-(\eta(y,T)-\eta(y,t))-(K_T-K_t)\ ,\end{eqnarray*}
where $K$ is the MVT process.
\end{remarque}

\begin{proof} 
The quadratic error can be calculated using Corollary~\ref{corro132Bis} and Corollary~\ref{ThmSolutionPb2}. It gives 
\begin{eqnarray}\mathbb{E}\left[\int_0^T\exp\left\{-(K_T-K_s)\right\}
d\left\langle L\right\rangle_s\right]\ ,\label{E1}\end{eqnarray}
where $L$ is the remainder martingale in the FS decomposition of $H$.
We proceed now to the evaluation of  
 $\left\langle L\right\rangle$.

Using $(\ref{44bis}),(\ref{44ter})$, Remark \ref{remarque17} and the
 bilinearity of the covariation give 
\begin{eqnarray*}
Re\left(\left\langle L(y),L(z)\right\rangle\right)&=&
\frac{1}{2}\left ( \left\langle L(y),{L(z)}\right
\rangle+\left\langle \overline{L(y)},\overline{L(z)}\right\rangle\ \right)  
\\
&=& \frac{1}{2}\left(\left\langle
    L(y)+L(\bar{z}),\overline{L(y)+L(\bar{z})}\right\rangle-\left\langle
    L(y),\overline{L(y)}\right\rangle-\left\langle
    L(z),\overline{L(z)}\right\rangle\right)\ ,
\end{eqnarray*}
and
\begin{eqnarray*}
\left\langle L(y)+L(\bar{z}),\overline{L(y)+L(\bar{z})}\right\rangle
&\leq&
\left\langle L(y)+L(\bar{z}),\overline{L(y)+L(\bar{z})}\right\rangle+
\left\langle L(y)-L(\bar{z}),\overline{L(y)-L(\bar{z})}\right\rangle\ ,\\ \\
&=&2\left\langle L(y),\overline{L(y)}\right\rangle+2\left\langle L(z),\overline{L(z)}\right\rangle\ .
\end{eqnarray*}
\eqref{EIZ1}  in the proof of Theorem \ref{propo310}, 
and considerations above allow to prove that
$$ \sup_{y, z\in I+i\mathbb{R}}\left| 
Re\left(\left\langle L(y),L(z)\right\rangle \right) \right|
< \infty. $$
Similarly we can bound 
 $Im(\left\langle L(y),L(z)\right\rangle_t)$, writing
$$ Im \left(\left\langle L(y),L(z)\right\rangle\right) 
=  \frac{1}{2}\left(\left\langle
    L(y)-L(\bar{z}),\overline{L(y)-L(\bar{z})}\right\rangle-\left\langle
    L(y),\overline{L(y)}\right\rangle-\left\langle
    L(z),\overline{L(z)}\right\rangle\right)\ ,
$$
so that we obtain 
$$  Im \left(\left\langle L(y),L(z)\right\rangle\right) \le  
\left\langle L(y),\overline{L(y)}\right\rangle+\left\langle L(z),
\overline{L(z)}\right\rangle\ $$
and 
$$ \sup_{y, z\in I+i\mathbb{R}}\left| 
Im \left(\left\langle L(y),L(z)\right\rangle \right) \right|
<\infty \ .
$$
Therefore
\begin{eqnarray*}\int\int\left\langle L(y),L(z)\right\rangle_t\Pi(dy)\Pi(dz)\end{eqnarray*}
is a well-defined, continuous, predictable, 
with bounded variation complex-valued process. \\
We recall that 
 $L_t=\int L(z)_t\Pi(dz)$ 
so 
\begin{eqnarray*}L_t^2=\int\int L(y)_tL(z)_t\Pi(dy)\Pi(dz)\ .\end{eqnarray*}
An application of Fubini's theorem yields that
\begin{eqnarray*}L_t^2-\int\int\left\langle L(y),L(z)\right\rangle_t\Pi(dy)\Pi(dz)\ ,\end{eqnarray*}
is a martingale. This implies
\begin{eqnarray*}\left\langle L,L\right\rangle_t=\int\int\left\langle L(y),L(z)\right\rangle_t\Pi(dy)\Pi(dz)\ ,\end{eqnarray*}
by definition of oblique bracket. It remains to evaluate
 $\left\langle L(y),L(z)\right\rangle$ for $y,z \in  {\rm supp (\Pi)}$.
\\
We know by Proposition~\ref{corro36} that for all $y,z,y+z \in D$,
\begin{eqnarray*}\left\langle M(y),M(z)\right\rangle_t=\int_0^tS_{u-}^{y+z}\rho_{du}(y,z)\ .\end{eqnarray*}
Using the same terminology of Proposition~\ref{lemme38},~(\ref{51Bis}) says $\left\langle L(z),M\right\rangle_t=0$ and~(\ref{43}) imply
\begin{eqnarray*}
\left\langle L(y),L(z)\right\rangle_t
&=&\int_0^t e^{\int_s^T(\eta(z,du)+\eta(y,du))}d\left\langle
  M(y),M(z)\right
\rangle_s
-\int_0^t\xi(z)_se^{\int_s^T\eta(y,du)}d\left\langle M,M(y)\right\rangle_s\ ,\\\\
&=&\int_0^t e^{\int_s^T(\eta(z,du)+\eta(y,du))}d\left\langle M(y),M(z)
\right\rangle_s
-\int_0^t\gamma(z,s)e^{\int_s^T(\eta(z,du)+\eta(y,du))}S_{s-}^{z-1}d\left\langle M,M(y)\right\rangle_s\ ,\\\\
&=&\int_0^t e^{\int_s^T(\eta(z,du)+\eta(y,du))}S_{s-}^{y+z}\rho_{ds}(y,z)-\int_0^t\gamma(z,t)e^{\int_s^T(\eta(z,du)+\eta(y,du))}S_{s-}^{z-1}S_{s-}^{y+1}\rho_{ds}(y,1)\ , \\ \\
&=&\int_0^t e^{\int_s^T(\eta(z,du)+\eta(y,du))}S_{s-}^{y+z}
\left[\rho_{ds}(y,z)-\gamma(z,s)\rho_{ds}(y,1)\right]\ .
\end{eqnarray*}
Hence,
\begin{eqnarray*}\left\langle L(y),L(z)\right\rangle_t=\int_0^t e^{\int_s^T(\eta(z,du)+\eta(y,du))}S_{s-}^{y+z}\beta(y,z,ds)\ .\end{eqnarray*}
We come back to~(\ref{E1}).
Recalling Remark \ref{remarque328} we have
\begin{eqnarray*}
\int_0^T e^{-(K_T-K_t)}d\left\langle L(y),L(z)\right\rangle_t
&=&\int_0^T e^{-(K_T-K_t)+\int_t^T
(\eta(z,du)+\eta(y,du))}S_{t-}^{y+z}\beta(y,z,dt)\ ,\\ \\
&=&\int_0^T e^{\alpha(y,z,t)}S_{t-}^{y+z}\beta(y,z,dt)\ .
\end{eqnarray*}
Since $\mathbb{E}[S_{t-}^{y+z}]=s_{0}^{y+z}e^{\kappa_t(y+z)}$, an application of Fubini's theorem yields
\begin{eqnarray}\label{E2}
\mathbb{E} \left (\int_0^T e^{-(K_T-K_t)}d\left\langle
    L(y),L(z)\right\rangle_t
\right )
&=&\mathbb{E} \left(\int_0^T
  e^{\alpha(y,z,t)}S_{t-}^{y+z}\beta(y,z,dt)\right)\
 , \nonumber \\ 
&&\\
&=& s_{0}^{y+z}\int_0^T e^{\alpha(y,z,t)+\kappa_t(y+z)}\beta(y,z,dt)\ .
\nonumber
\end{eqnarray}
which equals $J_0(y,z)$.

Another application of Fubini's theorem gives
\begin{eqnarray*}
\int_0^Te^{-(K_T-K_t)}d\left\langle L,L\right\rangle_t=\int_\mathbb{C}\int_\mathbb{C}\int_0^Te^{-(K_T-K_t)}d\left\langle L(y),L(z)\right\rangle_t\Pi(dy)\Pi(dz)\ ,
\end{eqnarray*}
hence
\begin{eqnarray*}
\mathbb{E}\left[\int_0^Te^{-(K_T-K_t)}d\left\langle L,L\right\rangle_t\right]
&=&
\int_\mathbb{C}\int_\mathbb{C}\mathbb{E}\left[\int_0^Te^{-(K_T-K_t)}d\left\langle L(y),L(z)\right\rangle_t\right]\Pi(dy)\Pi(dz)\ ,\\\\
&=&\int_\mathbb{C}\int_\mathbb{C} J_0(y,z)\Pi(dy)\Pi(dz)\ .
\end{eqnarray*}
Corollaries~\ref{corro132Bis} and \ref{ThmSolutionPb2} imply that the left-hand side of the previous equation provides the variance of the hedging error.
\end{proof}




\subsection{The exponential Lévy case}


\bigskip
In this section, we specify rapidly the results concerning
FS decomposition and the minimization problem when $(X_t)$ is a Lévy
process $(\Lambda_t)$.
 Using the fact that $(\Lambda_t)$ is a process with independent stationary increments it is not difficult to show that 
\begin{eqnarray}
\kappa_t(z)=t\kappa^\Lambda(z)\ ,
\label{431}\end{eqnarray}
where $\kappa^\Lambda(z)=\kappa_1(z)$,
 $\kappa^\Lambda: D \rightarrow \mathbb{C}$.
Since for every $z \in D$, $t \mapsto \kappa_t(z)$ has bounded
variation then $X=\Lambda$ is a semimartingale and Proposition
 \ref{proposition38} implies that $(t,z) \mapsto \kappa_t(z)$
is continuous.


We make the following hypothesis.

\begin{Hyp}\label{HypLevy}
\begin{enumerate}
\item $2\in D\,;$
\item $\kappa^\Lambda(2)-2\kappa^\Lambda(1)\neq 0\,.$
\end{enumerate}
\end{Hyp}

\begin{remarque}
\begin{enumerate}
\item $\rho_{dt}=\left(\kappa^\Lambda(2)-2\kappa^\Lambda(1)\right)dt\,;$
\item ${\displaystyle \frac{d\kappa_t}{d\rho_t}(z)=\frac{1}
{\kappa^\Lambda(2)-2\kappa^\Lambda(1)}\kappa^\Lambda(z)}$ for any $t \in [0,T], z\in D\,;$
so $D = \shd$.
\item Assumptions \ref{HypD},
and \ref{Hyp1bis} are verified.
\end{enumerate}
\end{remarque}

Again we denote the process $S$ as
\begin{eqnarray*}S_t=s_0\exp(X_t)=s_0\exp(\Lambda_t)\ .\end{eqnarray*}

It remains to verify Assumption \ref{Hyp1} which of course
 depends on the contingent claim.
\begin{example}
\begin{enumerate}
\item $H=(S_T-K)_+$. We choose the second representation for the call.
 So, for $0<R<1$,
\begin{eqnarray*}
I_0=supp( \Pi) \cap \R =\left\{R,1\right\}.
\end{eqnarray*}
In this case Assumption \ref{Hyp1}.1 becomes $I=[R,R+1]\subset D$. This is
 always satisfied since $D\supset [0,2]$ and it is convex. Assumption \ref{Hyp1}.2 is always verified because I is compact and $\kappa^\Lambda$ is continuous.
\item $H=(K-S_T)_+$. We recall that $R<0$ and so
\begin{eqnarray*}
I_0  = supp( \Pi) \cap \R  = \{R\}.
\end{eqnarray*}
In this case, Assumption~\ref{Hyp1}.1, gives again $I=[2R,1]\subset D$.
 Since $[0,2]$ is always included in D, we need to suppose here that
 $2R$ (which is a negative value) belongs to D. \\
This is not a restriction provided that $D$ contains some negative
values since we have the degree of freedom for choosing $R$.
\end{enumerate}
\end{example}

In this subsection, we reobtain
 results obtain in~\cite{Ka06}.
 From Proposition~\ref{propoabsconti}, we obtain the following.
\begin{corro}\label{propolevycase}
Under Assumption \ref{HypLevy}, the process $(S_t)$ can be written as
\begin{eqnarray*}
S_t=M_t+A_t\ ,
\end{eqnarray*}
where 
\begin{eqnarray*}
A_t=\kappa(1)\int_0^t S_{u-} du \quad \textrm{and}\quad \left\langle M,M\right\rangle_t=(\kappa(2)-2\kappa(1))\int_0^tS_{u-}^2du\ .
\end{eqnarray*}
The mean-variance tradeoff process equals
\begin{eqnarray}K_t=\int_0^t\alpha_u^2d\left\langle M,M \right\rangle_u=\frac{\kappa(1)^2}{\kappa(2)-2\kappa(1)}t\ .\label{Kexplevy}\end{eqnarray}
\end{corro}
From Theorem \ref{propo310}  and Theorem \ref{mainthm}, we obtain the 
following result.
\begin{thm} \label{THMREPLEVY}
We suppose the validity of Assumption \ref{HypLevy}.
 We consider  an option $H$ of the type $(\ref{FORM})$. 
The following properties hold true.
\begin{enumerate}
\item The FS decomposition is given by $H_T=H_0+\int_0^T\xi_tdS_t+L_T$ where
\begin{itemize}
\item $H_t=\int H(z)_t\Pi(dz)$ with $H(z)_t=\exp(\eta^\Lambda(z)(T-t))S_t^z$
 and $z \in I$, $t \in [0,T]\,;$
\item $\xi_t=\int \xi(z)_t\Pi(dz)$ with $\xi(z)_t=\gamma^\Lambda(z)
\exp(\eta^\Lambda(z)(T-t))S_{t-}^{z-1}$ and $z \in I$, $t \in [0,T]\,;$
\item $L_t=H_t-H_0-\int_0^t \xi_u dS_u\,.$
\end{itemize}
Moreover, for $z \in {\rm supp} \Pi$,
\begin{itemize}
\item $\gamma^\Lambda(z)={\displaystyle \frac{\kappa(z+1)-
\kappa(z)-\kappa(1)}{\kappa(2)-2\kappa(1)}}\,;$
\item $\eta^\Lambda(z)=\kappa(z)-\kappa(1)\gamma^\Lambda(z)\,.$
\end{itemize}
According to the notations of Lemma \ref{lemme38}, 
we have 
$$ \eta(z,t) = \eta^\Lambda(z) t, \quad \gamma(z,t) =  \gamma^\Lambda(z). $$
\item The solution of the minimization problem is given by
 a pair $(V_0,\varphi)$ where
\begin{eqnarray*}V_0=H_0 \quad \textrm{and}\quad 
\varphi_t=\xi_t+\frac{\lambda}{S_{t-}}(H_{t-}-V_0-G_{t-}(\varphi)) 
\quad \textrm{with}\quad \lambda=\frac{\kappa(1)}{\kappa(2)-2\kappa(1)}\,.
\end{eqnarray*}
\end{enumerate}
\end{thm}
\begin{remarque} 
Lemma \ref{lemmaTheta}  implies that 
  $\Theta$ is the linear space of predictable processes $v$ 
such that
$ \mathbb{E}\left(\int_0^T v_t^2 S_{t-}^2 dt\right) < \infty.$
\end{remarque}

\begin{remarque} 
We come back to the examples introduced in Remark \ref{remark47Bis}.
In all the three cases, Assumption \ref{HypLevy} is verified if $2 \in D$. This is happens in the following situations:
\begin{enumerate}
\item always in the Poisson case;
\item if $\Lambda=X$ is a NIG process and if $2<\alpha-\beta\,;$
\item if $\Lambda=X$ is a VG process and if $2<-\beta+\sqrt{\beta^2+2\alpha}\,.$
\end{enumerate}
\end{remarque}

\begin{remarque}
If X is a Poisson process with parameter $\lambda > 0$
 then the quadratic error is zero. In fact, the quantities
\begin{eqnarray*}
\kappa^\Lambda(z) &=& \lambda (\exp(z) -1)) \\
\rho_t(y,z) &=& \lambda t(\exp(y) - 1)(\exp(z) -1) \\
\gamma(z,t) &=& \frac{\exp(z) -1}{e-1}
\end{eqnarray*}
imply that $\beta(y,z,t) = 0$ for every $y,z \in \C, t \in [0,T]$.

Therefore $J_0(y,z,t)\equiv0$. In particular all the options of type~(\ref{FORM})  are perfectly hedgeable.
\end{remarque}

\subsection{Exponential of a Wiener integral driven by a Lévy process}
\label{ExpWiener}

Let $\Lambda$ be a Lévy process. The cumulant function  of $\Lambda_t$ 
equals $\kappa^\Lambda_t(z)=t\kappa_1^\Lambda(z)$ for $\kappa_1^\Lambda
=\kappa^\Lambda:D_\Lambda\rightarrow\mathbb{C}$. 
 We formulate the following hypothesis: 
\begin{Hyp}\label{HypWiener}
\begin{enumerate}
\item There is $r>0$ such that $r \in D_\Lambda$.
\item $\Lambda$ has no deterministic increments.
\end{enumerate}
\end{Hyp}

\begin{remarque}
\label{rem:kappa:gamma}
According to Lemma~\ref{L2} for every $\gamma>0$,  such that $\gamma \in D$, 
\begin{eqnarray}\kappa^\Lambda(2\gamma)-2\kappa^\Lambda(\gamma)>0\ .
\label{equa1}\end{eqnarray}
\end{remarque}

We consider the PII process $X_t=\int_0^t l_sd\Lambda_s$ where
$l: [0,T]\rightarrow [\varepsilon, r/2]$
 is a (deterministic continuous) function 
and $\varepsilon,r>0$ such that $2 \varepsilon \le r$.

\begin{remarque} \label{R5.14bis}
\begin{enumerate}
\item Lemma~\ref{lemme27} says that $D$ contains
$ D_{\varepsilon, r}:= \left\{x\in \mathbb{R}\,|\,\varepsilon x, 
\frac{r  x}{2}
  \in D_\Lambda\right\}+i\mathbb{R}\,,$ and $\kappa_t(z)=
\int_0^t\kappa^\Lambda(zl_s)ds\,.$ 

\item $\rho_t=\int_0^t\left(\kappa^\Lambda(2 l_s)-
2\kappa^\Lambda(l_s)\right)ds\,;$
\item $2\in D\,;$
$X$ is a PII semimartingale since 
$t \mapsto \kappa_t(2)$ has bounded variation, see Lemma \ref{P1}.
\item $1 \in D_{\varepsilon, r}$ since $0, r \in D_\Lambda$.
\end{enumerate}
\end{remarque}

\begin{propo}\label{propo411}
Assumptions \ref{HypD}  and \ref{Hyp1bis} are verified.
Moreover $D_{\varepsilon, r} \subset \shd$.
\end{propo}
\begin{proof}
\begin{enumerate}
\item Using Lemma~\ref{L2}, Assumption~\ref{HypD} is verified if we show that $t\mapsto \rho_t(1)=\kappa_t(2)-2\kappa_t(1)$ is strictly increasing. Now
\begin{eqnarray*}\kappa_t(2)-2\kappa_t(1)=\int_0^t\left(\kappa^\Lambda(2l_s)-2\kappa^\Lambda(l_s)\right)ds \ .\end{eqnarray*}
Inequality~(\ref{equa1}) and Lemma~\ref{L2} imply that $\forall s\in[0,T]$
\begin{eqnarray*}\kappa^\Lambda(2l_s)-2\kappa^\Lambda(l_s)>0\ .\end{eqnarray*}
In fact, $\Lambda$ has no deterministic increments. This shows Assumption \ref{HypD}.
\item 
For $z \in D_{\varepsilon, r}$, by Remark \ref{R5.14bis} 1.
 we have
\begin{eqnarray*}\left|\frac{d\kappa_t(z)}{d\rho_t}\right|
=\left|\frac{\kappa^\Lambda(zl_t)}{\kappa^\Lambda(2l_t)-
2\kappa^\Lambda(l_t)}\right|\leq
 \frac{\sup_{x\in[\varepsilon,r]}\vert \kappa^\Lambda(xz) \vert}
{\inf_{x\in [\varepsilon, r/2]}\left
(\kappa^\Lambda(2x)-2\kappa^\Lambda(x)\right)}\ .\end{eqnarray*}
Previuous supremum and infimum exist since $x\mapsto \kappa^\Lambda(zx)$ is continuous and it attains a maximum and a minimum on a compact interval.
So, $D_{\varepsilon, r} \subset \shd$ and
 Assumption \ref{Hyp1bis} is verified because of Remark \ref{propo411} 4.
\end{enumerate}
\end{proof}

\begin{remarque}\label{remarque412}
\begin{enumerate}
\item Point 1. of Assumption  \ref{Hyp1} is  also  verified if 
we show that $I  \subset D_{\varepsilon, r}$; in fact 
 $ D_{\varepsilon, r} \subset \shd$ and  $I_0 \cup (I_0 + 1) \subset I$.
\item From previous proof it follows that
\begin{eqnarray*}\frac{d\kappa_t(z)}
{d\rho_t}=\frac{\kappa^\Lambda(zl_t)}
{\kappa^\Lambda(2l_t)-2\kappa^\Lambda(l_t)}\ .\end{eqnarray*}
\item Since $I$ is compact 
and $t \mapsto \frac{d\kappa_t(z)}{d\rho_t}$ is continuous,
point 2. of Assumption \ref{Hyp1}
would be verified again for all cases provided that $I  \subset 
 D_{\varepsilon, r}$.
\end{enumerate}
\end{remarque}
It remains to verify Assumption \ref{Hyp1} for the same class 
of options as in previous subsections. 
The only  point to establish will be to show 
\begin{equation} \label{E4000}
I \subset \{x \vert \varepsilon x, \frac{r x}{2} \in D_\Lambda \}.
\end{equation}
%
\begin{example} \label{E517}

\begin{enumerate}
\item $H=(S_T-K)_+$. Similarly to the case where $X$ is a Lévy process, 
we take the second representation of the European Call.
In this case $I = [R,R+1]$ and \eqref{E4000} is verified.
\item $H=(K-S_T)_+$. Again, here $R<0$,
$I = [2R, R+1]$. \\
We only have to require that $D_\Lambda$
contains some negative values, which is the case for the three
examples introduced at Section \ref{LevyPC}.
Selecting $R$  in a proper way, \eqref{E4000}is fulfilled.
\end{enumerate}
\end{example}

We provide now the solution to the minimization problem
under Assumption \ref{HypWiener}.
. By Theorem~\ref{mainthm}, we have
\begin{eqnarray*}\lambda(s)=\frac{\kappa^\Lambda(l_s)}{\kappa^\Lambda(2l_s)-2\kappa^\Lambda(ls)}\ ,\end{eqnarray*}
\begin{eqnarray*}\gamma(z,s)=\frac{\kappa^\Lambda((z+1)l_s)-\kappa^\Lambda(zl_s)-\kappa^\Lambda(l_s)}{\kappa^\Lambda(2l_s)-2\kappa^\Lambda(l_s)}\ ,\end{eqnarray*}
\begin{eqnarray*}\eta(z,s)=\kappa^\Lambda(zl_s)-\frac{\kappa^\Lambda(l_s)}{\kappa^\Lambda(2l_s)-2\kappa^\Lambda(l_s)} \left(\kappa^\Lambda((z+1)l_s)-\kappa^\Lambda(zl_s)-\kappa^\Lambda(l_s)\right)\ ,\end{eqnarray*}
hence
\begin{eqnarray*}\eta(z,s)=\kappa^\Lambda(zl_s)-\lambda(s)\left(\kappa^\Lambda((z+1)l_s)-\kappa^\Lambda(zl_s)-\kappa^\Lambda(l_s)\right)\ ,\end{eqnarray*}
We obtain finally the optimal hedging
\begin{eqnarray*}\varphi_t=\xi_t+\frac{\lambda_t }{S_{t-}}(H_{t-}-V_0-\int_0^t\varphi_sdS_s)\ ,\end{eqnarray*}
where the processes $(H_t)$,$(\xi_t)$ are defined by
\begin{eqnarray*}H_t=\int_\mathbb{C} e^{\int_t^T\eta(z,ds)}S_t^z\Pi(dz)\ ,\end{eqnarray*}
\begin{eqnarray*}\xi_t=\int_\mathbb{C} \gamma(z,t)e^{\int_t^T\eta(z,ds)}S_{t-}^{z-1}\Pi(dz)\ .\end{eqnarray*}

\subsection{A toy example}
Let $(W_t)$ be a standard Brownian motion, we consider $X_t=W_{\psi(t)}$, where $\psi:\mathbb{R}_+\rightarrow\mathbb{R}_+$ is a strictly increasing function, including the pathological case where $\psi^{'}_t=0$ a.e. We have
\begin{eqnarray*}\mathbb{E}[e^{z X_t}]=\mathbb{E}[e^{z W_{\psi(z)}}]=e^{\kappa_t(z)}=e^{\frac{z^2}{2}\psi(t)}\ ,\end{eqnarray*}
so that 
\begin{eqnarray*}\kappa_t(z)=\frac{z^2}{2}\psi(t)\ ,\quad
  \kappa_{t}(2)-2\kappa_{t}(1)=\psi(t) \quad \textrm{and}\quad \kappa_{t}(z+1)-\kappa_{t}(z)-\kappa_{t}(1)=z\psi(t)\ .\end{eqnarray*}
So
\begin{eqnarray*}\left\langle M,M\right\rangle_t=\int_0^t
  S_{s-}^2\psi(ds) \quad\textrm{and}\quad A_t=\int_0^t \frac{1}{2S_{s-}}d\left\langle M,M\right\rangle_s=\int_0^t \frac{1}{2}S_{s-}\psi(ds)\ ,\end{eqnarray*}
and the MVT process verifies 
\begin{eqnarray*}K_t=\int_0^t\frac{1}{4S^2_{s-}}d\left\langle M,M\right\rangle_s=\int_0^t\frac{1}{4}\psi(ds)=\frac{1}{4}\psi(t)\ .\end{eqnarray*}
All the conditions to apply Theorem~\ref{mainthm} are satisfied so the function $\gamma(z,t)$ is equal to the Radon-Nykodim derivative of $\kappa_t(z+1)-\kappa_t(z)-\kappa_t(1)$ with respect to $\kappa_t(2)-2\kappa_t(1)$, so
\begin{eqnarray*}\gamma(z,t)=z\ ,\quad
  \eta(z,t)=\frac{\psi(t)}{2}(z^2-z) 
\quad \textrm{and}\quad \lambda(t) =\lambda=\frac{1}{2}\ .\end{eqnarray*}
Hence we can compute the variance-optimal hedging strategy $\varphi$ and the variance-optimal initial capital $V_0$ in this case
\begin{eqnarray*}\varphi_t=\xi_t+\frac{1}{2S_{t-}}(H_{t-}-V_0-\int_0^t\varphi_sdS_s) \end{eqnarray*}
and
\begin{eqnarray*}H_t=\int_\mathbb{C} e^{\int_t^T\eta(z,ds)}S_t^z\Pi(dz)=\int_\mathbb{C}\exp\left\{\frac{z^2-z}{2}(\Psi(T)-\Psi(t))\right\}S_t^z\Pi(dz)\end{eqnarray*}

\begin{eqnarray*}\xi_t=\int_\mathbb{C} \gamma(z,t)e^{\int_t^T\eta(z,ds)}S_{t-}^{z-1}\Pi(dz)=\int_\mathbb{C}z\exp\left\{\frac{z^2-z}{2}(\Psi(T)-\Psi(t))\right\}S_{t-}^{z-1}\Pi(dz)\end{eqnarray*}

\begin{remarque}
Calculating $\beta(y,z,t)$ of the quadratic error section, we find $\beta\equiv 0$. Therefore here also the quadratic error is zero. This confirms the fact that the market is \textbf{complete}, at least for the considered class of options.
\end{remarque}

\section{Application to Electricity}
\label{sec:elec}

%
%

\subsection{Hedging electricity derivatives with forward contacts}
%
Electricity markets are composed by the Spot market setting prices for each delivery hour of the next day and the forward or futures market setting prices for more distant delivery periods. For simplicity, we will assume that interest rates are deterministic and zero so that futures prices are equivalent to forward prices. 
Forward prices given by the market correspond to a fixed price of one MWh of electricity for delivery in a given future period, typically a month, a quarter or a year. Hence, the corresponding term contracts are in fact  swaps (i.e. forward contracts with delivery over a period) but are improperly named forward.  However, the strong assumption that there are tradable forward contracts for all future time points $T_d\geq 0$ is usual and will be assumed here. \\
Because of non-storability of electricity, no dynamic hedging strategy can be performed on the spot market. Hedging instruments for electricity derivatives are then futures or forward contracts. The value of a forward contract offering the fixed price $F_0^{T_d}$ at time $0$ for delivery of 1MWh at time $T_d$ is by definition of the forward price,  $S^{0,T_d}_0=0$. Indeed, there is no cost to enter at time $0$ the forward contract with the current market forward price $F_0^{T_d}$.  
Then, the value of the same forward contract $S^{0,T_d}$ at time $t\in
[0,T_d]$ is deduced by an argument of Absence of (static) Arbitrage as 
$S^{0,T_d}_t = e^{-r(T-t)}(F_t^{T_d} - F_0^{T_d})$. Hence, the dynamic of the hedging instrument $(S^{0,T_d}_t)_{0\leq t\leq T_d}$ is directly related (for deterministic interest rates) to the dynamic of forward prices $(F_t^{T_d})_{0\leq t\leq T_d}$.   
Consequently, in the sequel we will focus on the dynamic of forward prices.

\subsection{Electricity price models for pricing and hedging application}
Observing  market data, one can notice two main stylised features of electricity spot and forward prices: 
\begin{itemize}
	\item Volatility term structure of forward prices: the volatility increases when the time to maturity decreases;
	\item Non-Gaussianity of log-returns: log-returns can be considered as Gaussian for long-term contracts but they are clearly leptokurtic for short-term contratcs with huge spikes on the Spot market. 
\end{itemize}
Hence, a challenge is to be able to describe with a single model, both the spikes on the short term and the volatility term structure of the forward curve. 
One reasonable attempt to do so is to consider the exponential Lévy factor model, proposed by Benth and Benth~\cite{BS03}, or~\cite{oudjaneCollet}. The forward price given at time $t$ for delivery at time $T_d\geq t$, denoted $F_t^{T_d}$ is then modeled by a $p$-factors model, such that 
\begin{equation}
\label{eq:elec}
F_t^{T_d}=F_0^{T_d}\exp(m_t^{T_d}+\sum_{k=1}^pX_t^{k,{T_d}})\ ,\quad \textrm{for all}\ t\in [0,{T_d}]\ , \textrm{where}
\end{equation}
\begin{itemize}
	\item $(m_t^{T_d})_{0\leq t\leq {T_d}}$ is a real deterministic trend; 
	\item For any $k=1,\cdots p$, $(X_t^{k,{T_d}})_{0\leq t\leq {T_d}}$ is such
      that $X_t^{k,{T_d}}=\int_0^t \sigma_k e^{-\lambda_k ({T_d}-s)}
      d\Lambda^k_s$,
 where $\Lambda=(\Lambda^1,\cdots,\Lambda ^p)$ is a Lévy process 
on $\mathbb{R}^d$, with
 $\mathbb{E}[\Lambda^k_1]=0$ and $Var[\Lambda^k_1]=1$; 
	\item $\sigma_k>0\,,\,\lambda_k \geq0\ ,$ are called respectively the \textit{volatilities} and the \textit{mean-reverting rates}. 
\end{itemize}
Hence,  forward prices are given as exponentials of PII with \textit{non-stationary increments}. Then, the spot model is derived by setting $S_{T_d}=F_{T_d}^{T_d}$ and reduces to the exponential of a sum of possibly non-Gaussian Ornstein-Uhlenbeck processes. 
In practice, we consider the case of a one or a two factors model ($p=1$
 or $2$), where the first factor $X^1$ is a non-Gaussian PII and the second factor $X^2$ is a Brownian motion with $\sigma_1\gg \sigma_2$. 
Notice that this kind of model was originally developed and studied in
details for interest rates in~\cite{Raible98}, as an extension of
 the Heath-Jarrow-Morton model where the Brownian motion has
 been replaced by a general Lévy process. 
Recent contributions  in the subject are \cite{tappe, ruediger}.
\\
Of course, this modeling procedure~(\ref{eq:elec}), implies
incompleteness of the market. Hence, if we aim at pricing and hedging a
European call on a forward with maturity $T\leq T_d$, it won't be possible,
in general, to hedge perfectly the payoff $(F^{T_d}_T-K)_+$ with a hedging
portfolio of forward contracts.
Then, a natural approach could consist in looking for the variance
optimal price and hedging portfolio. In this framework, the results of
Section~\ref{sec:expPII} generalizing the results of Hubalek \& al in~\cite{Ka06} to the case of non stationary PII can be useful. Similarly, some arithmetic models proposed in~\cite{Benth-Kallsen} for  electricity prices, consists of replacing the right-hand side of~(\ref{eq:elec}) by its logarithm. Hence, with this kind of models the results of Section~\ref{sec:FSPII} can also be useful.

\subsection{The non Gaussian two factors model}

To simplify let us forget the upperscript $T_d$ denoting the delivery period (since we will consider a fixed delivery period).
 We suppose that the forward price $F$ follows the two factors model
\begin{equation}
\label{eq:elec:2F}
F_t=F_0\exp(m_t+X^1_t+X^2_t)\ ,\quad \textrm{for all}\ t\in [0,T_d]\ ,
 \textrm{where}
\end{equation}
\begin{itemize}
\item $m$ is a real deterministic trend starting at $0$. 
It is supposed to be absolutely continuous with respect to Lebesgue;
	\item  $X^1_t=\int_0^t \sigma_s e^{-\lambda (T_d-u)} d\Lambda_u$,
      where $\Lambda$ is a 
 Lévy process on
 $\mathbb{R}$ with $\Lambda$ following a Normal Inverse Gaussian (NIG) 
distribution
or a Variance Gamma  (VG) distribution. Moreover, we will assume that
  $\mathbb{E}[\Lambda_1]=0$ and $Var[\Lambda_1]=1$; 
	\item $X^2=\sigma_l W$ where $W$ is  a standard Brownian motion 
on $\mathbb{R}$;
\item $\Lambda$ and $W$ are independent.
	\item $\sigma_s$ and $\sigma_l$ standing respectively for the short-term volatilty and long-term volatility. 
\end{itemize}
%

\subsection{Verification of the assumptions}

The result below helps to extend Theorem \ref{mainthm} to the case where
$X$ is a finite sum of independent PII semimartingales,
each one verifying Assumptions \ref{HypD},  \ref{Hyp1bis} and \ref{Hyp1} 
for a given payoff $H = f(s_0 e^{X_T})$.
\begin{lemme}\label{LSum}
Let $X^1, X^2$ be two independent PII semimartingales with cumulant
generating functions $\kappa^i$ and related domains  
$D^i, \shd^i, i = 1,2$ characterized in Remark \ref{CumGenPAI}
and \eqref{Eqshd}. Let $f: \C \rightarrow \C$ of the form~(\ref{FORM}). \\
For $X = X^1 + X^2$ with related domains $D, \shd$
and cumulant generating function $\kappa$, we have  the following.
\begin{enumerate}
\item $D = D^1 \cap D^2$.
\item $\shd^1 \cap \shd^2 \subset \shd$.
\item  If $X^1, X^2$ verify
Assumptions \ref{HypD},  \ref{Hyp1bis} and \ref{Hyp1},
then $X$ has the same property.
\end{enumerate}
\end{lemme}
\begin{proof} 
Since $X^1, X^2$ are independent and taking into account
Remark \ref{CumGenPAI} we obtain 1. and 
$$  \kappa_t(z) = \kappa_t^1(z) + \kappa^2(z), \ \forall z \in D.$$
 We denote by $\rho^i, i = 1,2$, the reference variance
 measures defined
in Remark \ref{remarque313}. Clearly 
$\rho = \rho^1 + \rho^2$ and  $d\rho^i \ll  d\rho$
with $\Vert \frac{d\rho^i}{d\rho} \Vert_{\infty} \le 1$.\\
If $z \in \shd^1 \cap \shd^2$, we can write 
\begin{eqnarray*}
\int_0^T\left|\frac{d\kappa_t(z)}{d\rho_t}\right|^2d\rho_t&\leq&2\int_0^T\left|\frac{d\kappa^1_t(z)}{d\rho^1_t}\frac{d\rho^1_t}{d\rho_t}\right|^2d\rho_t 
+2\int_0^T\left|\frac{d\kappa^2_t(z)}{d\rho^2_t}\frac{d\rho^2_t}{d\rho_t}\right|^2d\rho_t \\
&=&2\int_0^T\left|\frac{d\kappa^1_t(z)}{d\rho^1_t}\right|^2\frac{d\rho^1_t}{d\rho_t}d\rho^1_t 
+2\int_0^T\left|\frac{d\kappa^2_t(z)}{d\rho^2_t}\right|^2\frac{d\rho^2_t}{d\rho_t}d\rho^2_t \\
&\leq& 2\left(\int_0^T\left|\frac{d\kappa^1_t(z)}{d\rho^1_t}\right|^2d\rho^1_t+\int_0^T\left|\frac{d\kappa^2_t(z)}{d\rho^2_t}\right|^2d\rho^2_t\right)\ .
\end{eqnarray*}
This concludes the proof of $\shd^1 \cap \shd^2 \subset \shd$ and
therefore of the of
Point 2. \\
Finally Point 3. follows then by inspection.
\end{proof}

With the two factors model, the forward price $F$  is then given
 as the exponential of a PII, $X$, such that for all $t\in [0,{T_d}]$, 
\begin{equation}
\label{eq:X}
X_t=m_t+X^1_t+X^2_t=m_t+\sigma_s 
\int_{0}^t e^{-\lambda ({T_d}-u)}d\Lambda_u+\sigma_{l}W_t\ .
\end{equation}
For this model, we formulate the following assumption.
\begin{Hyp}  \label{AssuElec}
\begin{enumerate}
\item $ 2 \sigma_s \in D_\Lambda$.
\item If $\sigma_l = 0$, we require $\Lambda$ not to have deterministic
increments.
\item We define $\varepsilon = \sigma_s e^{-\lambda T_d},
 \quad r = 2 \sigma_s$.\\
$f: \C \rightarrow \C$ is of the type~(\ref{FORM})
 fulfilling \eqref{E4000}.
\end{enumerate}
\end{Hyp}
\begin{propo} \label{PElec61}
\begin{enumerate} 
\item The cumulant generating function of $X$ defined by~(\ref{eq:X}), 
 $\kappa:[0,{T_d}]\times D\rightarrow \mathbb{C}$ is such that for all $z \in D_{\varepsilon, r} : =  \left\{x\in \mathbb{R}\,|\,
 x\sigma_s\in D_\Lambda \right \} + i \R $, then  for all $t\in [0,T_d]$, 
\begin{equation} \label{99}
\kappa_t(z)=z m_t+\frac{z^2\sigma^2_{l}t}{2}+\int_0^t \kappa^\Lambda(z\sigma_s e^{-\lambda ({T_d}-u)})du \ .
\end{equation}
In particular for fixed $z \in D_{\varepsilon, r}$,
 $t \mapsto \kappa_t(z)$ is absolutely continuous with respect to
 Lebesgue  measure.
\item  Assumptions \ref{HypD}, \ref{Hyp1bis} and  
\ref{Hyp1} are verified.
\end{enumerate}
\end{propo}
\begin{proof}
We set $\tilde X^2 = m + X^2$. We observe that 
$$ D^2 = \shd^2 = \C, \quad \kappa^2_t(z) = \exp(z m_t + z^2 \sigma_l^2 
\frac{t}{2}). $$
We recall that $\Lambda$ and $W$ are independent so that
$\tilde X^2$ and $X^1$ are independent.\\
$X^1$ is a process of the type studied at Section \ref{ExpWiener}.
According to Proposition \ref{propo411},
Remark \ref{remarque412} and  \eqref{E4000}
it follows that Assumptions \ref{HypD}, \ref{Hyp1bis} and  
\ref{Hyp1}
are verified for $X^1$. \\
Both statements 1. and 2. are now a consequence of 
Lemma \ref{LSum}. 
\end{proof}
\begin{remarque} \label{RExample}
For examples of $f$ fulfilling \eqref{E4000}, we refer to Example 
\ref{E517}.
\end{remarque}

The solution to the mean-variance problem is provided by Theorem~\ref{mainthm}.
\begin{thm}
The variance-optimal capital $V_0$ and the variance-optimal hedging strategy $\varphi$, solution of the minimization problem $(\ref{problem2})$, are given by
\begin{eqnarray}V_0=H_0\end{eqnarray}
and the implicit expression
\begin{eqnarray}\varphi_t=\xi_t+\frac{\lambda_t
  }{S_{t-}}(H_{t-}-V_0-\int_0^t\varphi_sdS_s),\end{eqnarray}
where 
the processes $(H_t)$,$(\xi_t)$ and $(\lambda_t)$ are defined 
as follows:
\begin{eqnarray*}
\widetilde{z}_t :&=&\sigma_s e^{-\lambda ({T_d}-t)},\\
\gamma(z,t) :&=&\frac{z\sigma^2_l+\kappa^\Lambda
((z+1)\widetilde{z})-\kappa^\Lambda(z\widetilde{z})-
\kappa^\Lambda(\widetilde{z})}
{\sigma^2_l+\kappa^\Lambda(2\widetilde{z})-
2\kappa^\Lambda(\widetilde{z})}, \\
\eta(z,t):&=&\left [ zm_t+\frac{z^2\sigma^2_l}{2}+
\kappa^\Lambda(z\widetilde{z})
-\gamma(z,t) \big
(m_t+\frac{\sigma^2_l}{2}+\kappa^\Lambda(\widetilde{z})\big )
\right]\,dt\ ,
\end{eqnarray*}
\begin{eqnarray*}\lambda_t &=&\frac{m_t+\frac{\sigma_l^2}{2}+
\kappa^\Lambda(\widetilde{z})}{\sigma^2_l+\kappa^\Lambda(2\widetilde{z})-2\kappa^\Lambda(\widetilde{z})}, \\
H_t &=&\int_\mathbb{C} e^{\int_t^T\eta(z,ds)}S_t^z\Pi(dz), \\
\xi_t &=& \int_\mathbb{C} \gamma(z,t)e^{\int_t^T\eta(z,ds)}
S_{t-}^{z-1}\Pi(dz)\ .\end{eqnarray*}
The optimal initial capital is unique. The optimal hedging strategy $\varphi_t(\omega)$ is unique up to some $(P(d\omega)\otimes dt)$-null set.
\end{thm}

\begin{remarque} \label{R65Levy}
Previous formulae are practically exploitable numerically.
 The last condition to be checked is 
\begin{eqnarray} \label{E731}
2 \sigma_s \in D_\Lambda.
\end{eqnarray}
%
In our classical examples, this is always verified.
\begin{enumerate}
\item $\Lambda_1$ is a Normal Inverse Gaussian random variable. If $ \sigma_s\leq \frac{\alpha-\beta}{2}$ then $(\ref{E731})$ is verified.
\item $\Lambda_1$ is a Variance Gamma random variable then $(\ref{E731})$ is verified. if for instance $\sigma_s\leq\frac{-\beta+\sqrt{\beta^2+2\alpha}}{2}\ .$
\end{enumerate}
\end{remarque}

\section{Simulations}

\subsection{Exponential Lévy}

We consider the problem of pricing a European call, with payoff $(S_T-K)_+$, where the underlying process $S$ is given as the exponential of a NIG Lévy process i.e. for all $t\in [0,T]$, 
$$
S_t=e^{X_t}\ ,\quad \textrm{where $X$ is a Lévy process with } \  X_1\,\sim\, NIG(\alpha,\beta,\delta,\mu)\ .
$$
The time unit is the year and the interest rate is zero in all our simulations. The initial value of the underlying is $S_0=100$ Euros.  The maturity of the option is $T=0.25$ i.e. three months from now. 
Five different sets of parameters for the NIG distribution have been considered, going from the case of \textit{almost Gaussian} returns corresponding to standard equities, 
 to the case of \textit {highly non Gaussian} returns. 
The standard set of parameters 
is  estimated on the \textit{Month-ahead} \textit{base} forward prices of the French Power market in 2007:  
\begin{equation}
\label{eq:para;levy}
\alpha=38.46 \,,\  \beta= -3.85\,,\ \delta = 6.40\,,\  \mu= 0.64 \ .
\end{equation}
Those parameters imply a zero mean, a standard deviation of $41\%$, a
skewness (measuring the asymmetry) of $-0.02$ and an excess kurtosis
(measuring the \textit{fatness} of the tails) of $0.01$. The other sets
of parameters are obtained by multiplying parameter  $\alpha$ by a
coefficient $C$, ($\beta,\delta,\mu$) being such that the first three
moments are unchanged. Note that when $C$ grows to infinity the tails of
the NIG distribution get closer to the tails of the Gaussian
distribution. For instance, Table~\ref{tab:kurtosis} shows how the
excess kurtosis (which is zero for a Gaussian distribution) is modified
with the five values of $C$ chosen in our simulations. 
%
\begin{figure}[htbp]
\begin{center}
\begin{tabular}{|c||c|c|c|c|c|c|}
\hline
 \textrm{Coefficient} & $C=0.08$ & $C=0.14$ & $C=0.2$ & $C=1$ &$C=2$ \\
\hline
\hline
 $\alpha$&  $3.08$&  $5.38$&  $7.69$&  $38.46$&  $76.92$ \\
\hline
 Excess kurtosis&  $1.87$&  $0.61$&  $0.30$&  $0.01$&  $4. \,10^{-3}$ \\
\hline
\end{tabular}
\end{center}
\caption{{\small Excess kurtosis of $X_1$ for different values of $\alpha$,  $(\beta, \delta, \mu)$ insuring the same three first moments.}}
\label{tab:kurtosis}
\end{figure}

We have compared on simulations the Variance Optimal strategy (VO) using the real NIG incomplete market model with the real values of parameters to the Black-Scholes strategy (BS) assuming Gaussian returns with the real values of mean and variance.  Of course, the VO strategy is by definition theoritically optimal in continuous time, w.r.t. the quadratic norm. However, both strategies are implemented in discrete time, hence the performances observed in our simulations are spoiled w.r.t. the theoritical continuous rebalancing framework. 

\subsubsection{Strike impact on the pricing value and the hedging ratio}
Figure~\ref{fig:strike} shows the Initial Capital (on the left graph) and the initial hedge ratio (on the right graph) produced by the VO and the BS strategies as functions of the strike, for three different sets of parameters $C=0.08\,,\ C=1\,,\ C=2$. We consider $N=12$  trading dates, which corresponds to operational practices on electricity markets, for an option expirying in three months. One can observe that BS results are very similar to VO results for $C\geq 1$ which corresponds to \textit{almost Gaussian returns}.  
However, for small values of $C$, for $C=0.08$,  corresponding to highly non Gaussian returns,  BS approach under-estimates \textit{out-of-the-money} options and over-estimates \textit{at-the-money} options. For instance, on Figure~\ref{tab:IC}, one can observe that for $K=99$ Euros the Black-Scholes  Initial Capital ($\textrm{IC}_{BS}$) represents $122\%$ of the variance optimal Initial Capital ($\textrm{IC}_{VO}$), while for $K=150$ it represents only  $57\%$ of the  variance optimal price. Moreover, the hedging strategy differs sensibly for $C=0.08$, while it is quite similar to BS's ratio for $C\geq 1$.  
\begin{figure}[htbp]
\begin{center}
   \epsfig{figure=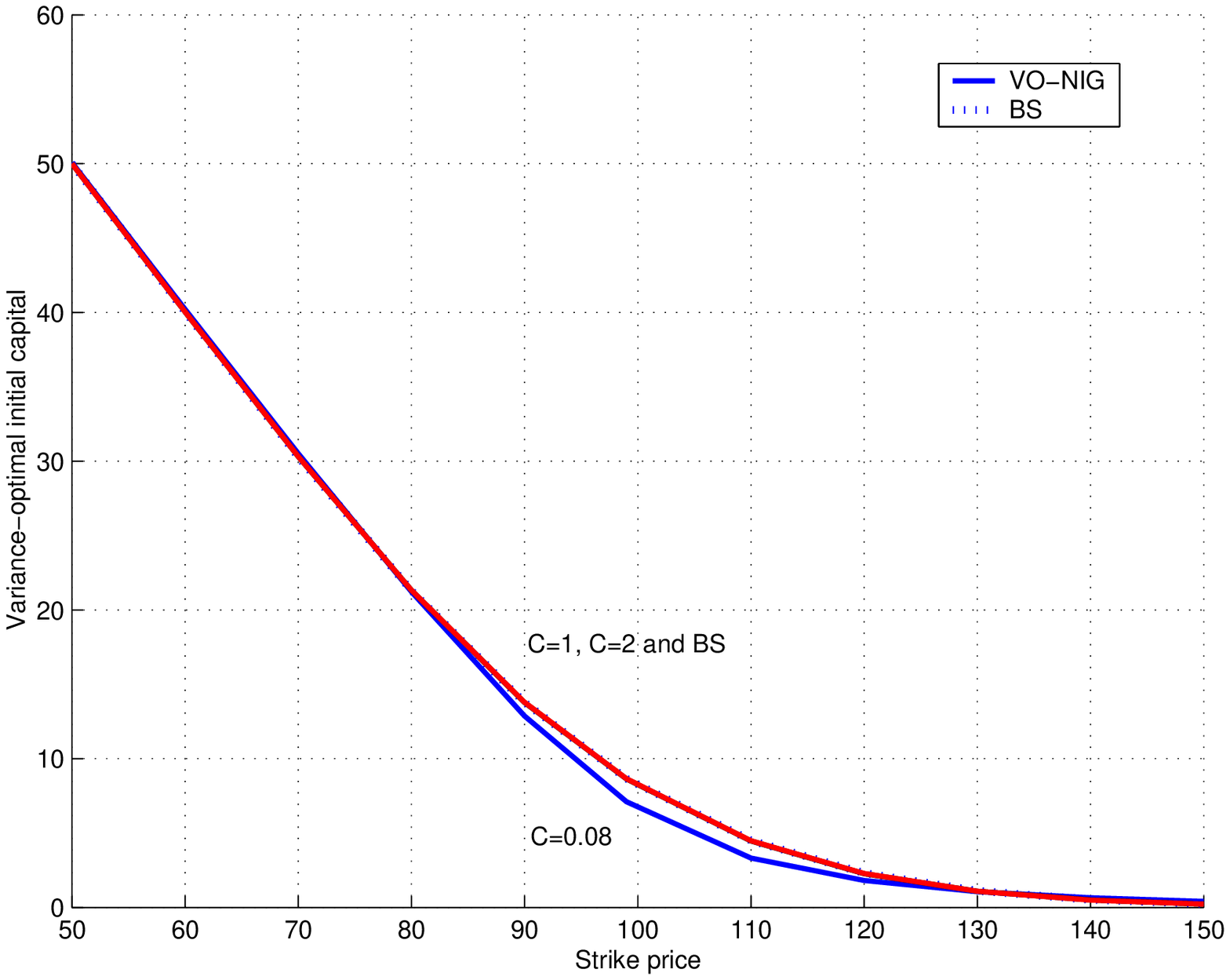,height=5cm, width=7.5cm}
   $\quad$
   \epsfig{figure=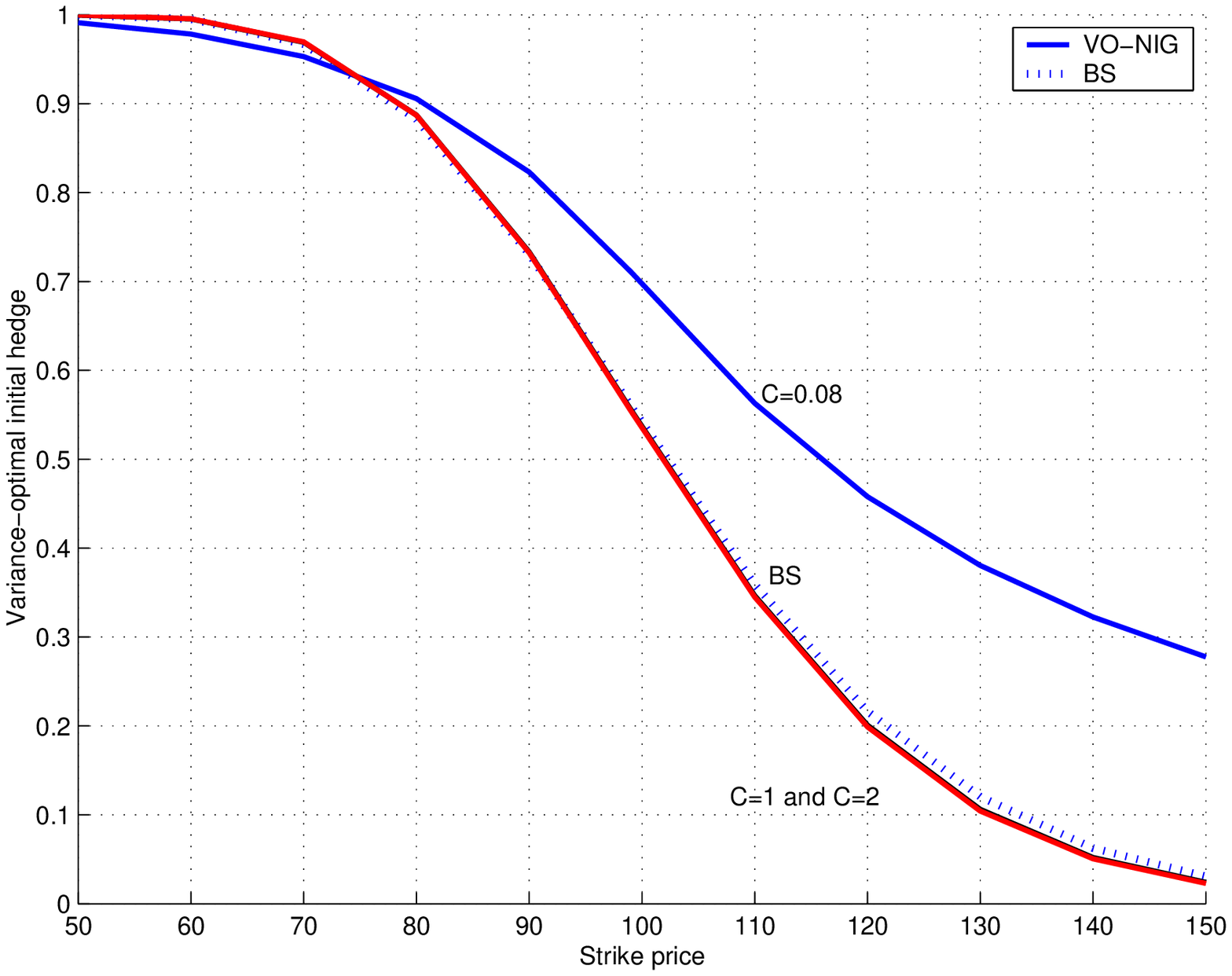,height=5cm, width=7.5cm}
\end{center}
\caption{{\small Initial Capital (on the left) and hedge ratio (on the right) w.r.t. the strike, for $C=0.08\,,\ C=1\,,\ C=2$.  }}
\label{fig:strike}
\end{figure}
\begin{figure}[htbp]
\begin{center}
\begin{tabular}{|c||c|c|c|c|}
\hline
 \textrm{Strikes} & $K=60$ & $K=99$ & $K=150$\\
\hline
\hline
$\textrm{IC}_{VO}$  &  $50.08$ &  $7.11$ &  $0.40$  \\
\hline
  $\textrm{IC}_{BS}$ (vs $\textrm{IC}_{VO}$) &  $50.00$ ($99.56\%$)&  $8.65$ ($121.73\%$)&  $0.23$ ($57.30\%$) \\
\hline
\end{tabular}
\end{center}
\caption{{\small Initial Capital of VO pricing ($\textrm{IC}_{VO}$) vs Initial Capital of BS pricing ($\textrm{IC}_{BS}$) for $C=0.08$. }}
\label{tab:IC}
\end{figure}

\subsubsection{Hedging error and number of trading dates}
Figure~\ref{fig:trading:dates} considers the hedging error (the
difference between the terminal value of the hedging portfolio and the
payoff) as a function of the number of trading dates, for a strike
$K=99$ Euros (at the money) and for five different sets of parameters $C$
described on Figure~\ref{tab:kurtosis}. 
%
The bias (on the left graph) and standard deviation (on the right graph) of the hedging error have been estimated by Monte Carlo method on $5000$ runs. Note that we could have used the formula stated in Theorem~\ref{thm34} to compute the variance of the error, but this would have give us the limiting error which does not take into account the additional error due to the finite number of trading dates. 

In terms of standard deviation, the VO strategy seems to outperform sensibly the BS strategy, for small values of $C$. For instance, one can observe on Figure~\ref{tab:hedging:error}, for $C=0.08$ that the VO strategy allows to reduce $10\%$ of the standard deviation of the error.  
As expected, one can observe that the VO error converges to the BS error when $C$ increases. This is due to the convergence of  NIG log-returns to Gaussian log-returns when $C$ increases (recall that the simulated log-returns are almost symmetric). 
One can distinguish two sources of incompleteness, the \textit{rebalancing error} due to the dicrete rebalancing strategy and the \textit{intrinsic error} due to the model incompleteness. 
On Figure~\ref{fig:trading:dates}, the hedging error (both for BS and VO) decreases with the number of trading dates and seems to converge to a limiting error corresponding to the intrinsic error. For $C=1$ and for a small number of trading dates $N\leq 5$, the rebalancing error represents the most part of the hedging error, then it  seems to vanish over $N=30$ trading dates, where the intrinsic error is predominant. For small values of $C\leq 0.2$, even for small numbers of trading dates,  the intrinsic error seems to be predominant. For $C\leq 0.2$ and $N\geq 12$ trading dates,  it seems useless to increase the number of trading dates. 
%
Moreover, one can observe that for a small number of trading dates $N\leq 12$ and for large values of $C\geq 1$,  BS seems to outperform the VO strategy, in terms of standard deviation.
 This can be interpreted as a consequence of the central limit theorem. Indeed, when the time between two trading dates increases the corresponding increments of the Lévy process converge to a Gaussian variable. Hence, the model error comitted by the BS approach decreases when the number of trading dates decreases. 

In term of bias, the over-estimation of at-the-money options (observed
for $C=0.08$, on Figures~\ref{fig:strike},~\ref{tab:IC}) seems to induce
a positive bias for the BS error (see Figure~\ref{fig:trading:dates}),
whereas the Bias of the VO error is negligeable (as expected from the
theory). However, one can observe on Figure~\ref{tab:hedging:error},
that the difference  between VO and BS bias error is smaller than the
difference between the Initial Capitals, therefore one can conclude
that, in our simulations, the BS hedging strategy induces more losses in average than the VO strategy. 

However, to be more relevant in our analysis, we have compared on Figure~\ref{fig:trading:datesINI}, the performances of the BS hedging portfolio with the VO hedging portfolio starting with the same Initial Capital as the BS hedging portfolio. One can observe on Figure~\ref{tab:hedging:error} that this approach allows to reduce the standard deviation of the VO hedging error (increasing the bias and of course the global quadratic error w.r.t. the VO strategy with optimal Initial capital). 

It is interesting to notice that, in terms of skewness and kurtosis, the
VO strategy seems to outperform sensibly the BS strategy for small
values of $C$. Figure~\ref{tab:hedging:error:moments} shows that for
$C=0.08$, the skewness of the BS hedging error is strongly negative (3
times greater than the VO error using the same Initial Capital) and the
kurtosis is high (14 times greater than the VO error). Hence, in our
simulations, BS strategy seems to imply more extreme losses than the VO strategy.   

In conclusion, the VO approach provides initial capital and hedging
strategies which are not significantly different from the BS approach
except for log-returns with high excess kurtosis (with small values of
parameter $\alpha$ in the NIG case). Similarly, we can observe (though
the figures are not reported here) the same behaviour w.r.t. to the
asymmetry of the distribution: the VO approach allows to outperform
significantly the BS approach for strongly asymmetric log-returns (with
high (absolute) values of parameter $\beta$ in the NIG case). On the
other hand, in more standard cases, the VO strategy seems to be comparable with the BS strategy in terms of quadratic error and to have the
significant and unexpected advantage to limit extreme losses (skewness and kurtosis) compared to the BS strategy.  
\begin{figure}[htbp]
\begin{center}
   \epsfig{figure=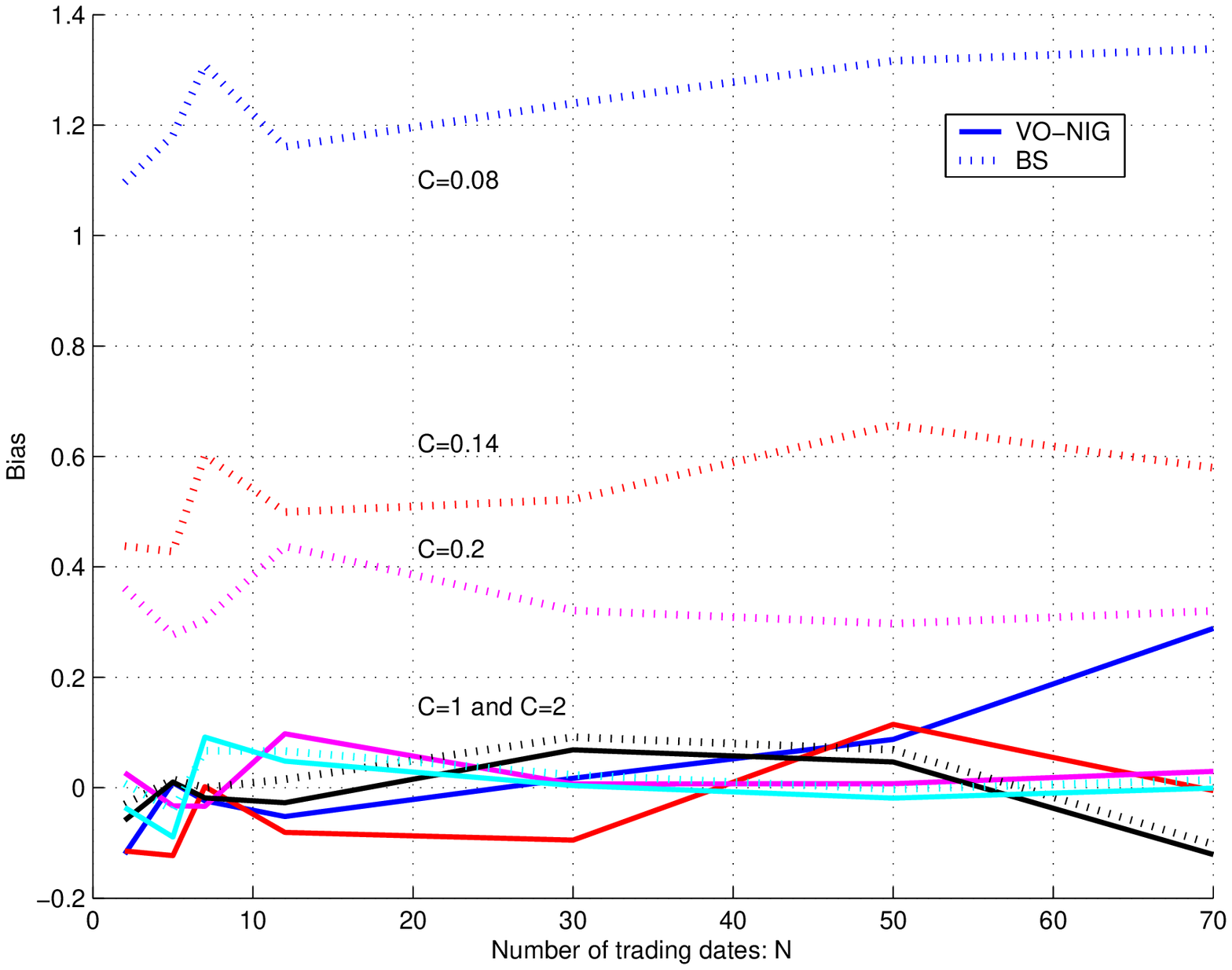,height=5cm, width=7.5cm}
   $\quad$
   \epsfig{figure=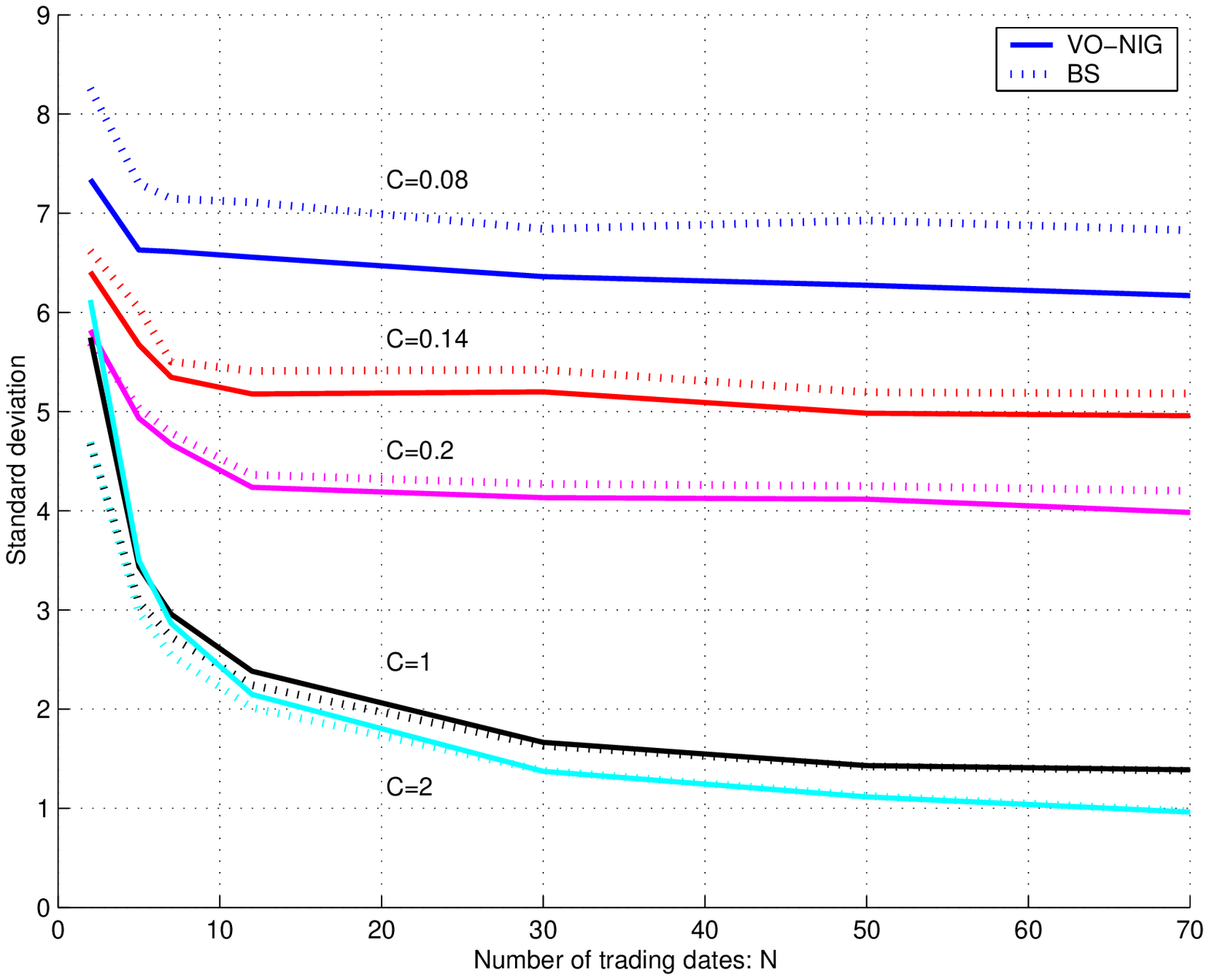,height=5cm, width=7.5cm}
\end{center}
\caption{{\small Hedging Error w.r.t. the number of trading dates for different values of $C$ and for $K=99$ Euros (Bias, on the left 	and standard deviation, on the right).}}
\label{fig:trading:dates}
\end{figure}
\begin{figure}[htbp]
\begin{center}
\begin{tabular}{|c||c|c|c|c|c|c|}
\hline
 \textrm{Coefficient} & $C=0.08$ & $C=0.14$ & $C=0.2$  & $C=1$ & $C=2$\\
\hline
\hline
 $\textrm{Std}_{VO}/\textrm{Std}_{BS}$ &  $91.19\%$&  $95.88\%$&  $97.63\%$ &  $107.52\%$ &  $109.39\%$ \\
 \hline
 $\textrm{Bias}_{BS}-\textrm{Bias}_{VO}$ &  $1.20$&  $0.57$&  $0.32$ &  $0.022$ &  $0.019$ \\
\hline
 $\textrm{IC}_{BS}-\textrm{IC}_{VO}$ &  $1.55$&  $0.7$&  $0.39$ &  $0.01$ &  $0$ \\
\hline
\end{tabular}
\end{center}
\caption{{\small Variance optimal hedging error vs Black-Scholes hedging error for different values of $C$ and for $K=99$ Euros (averaged values for different numbers of trading dates). }}
\label{tab:hedging:error}
\end{figure}
%
%
\begin{figure}[htbp]
\begin{center}
\begin{tabular}{|c||c|c|c|c|c|c|c|}
\hline
 \textrm{Moments} & Mean & Standard deviation & Skewness & Kurtosis  \\
\hline
\hline
 VO& $-0.049$&  $6.59$&  $-3.50$&  $31.51$ \\
\hline
 BS& $1.27$& $7.25$&  $-7.65$&  $152.09$ \\

\hline
 VO with $\textrm{IC}_{VO}=\textrm{IC}_{BS}$ & $1.39$&  $6.47$&  $-2.37$&  $10.70$ \\            
\hline
\end{tabular}
\end{center}
\caption{{\small Empirical moments of the hedging error for $C=0.08$, $N=12$ and $K=99$ Euros (averaged values for different number of trading dates).  }}
\label{tab:hedging:error:moments}
\end{figure}
%
%
%
\begin{figure}[htbp]
\begin{center}
   \epsfig{figure=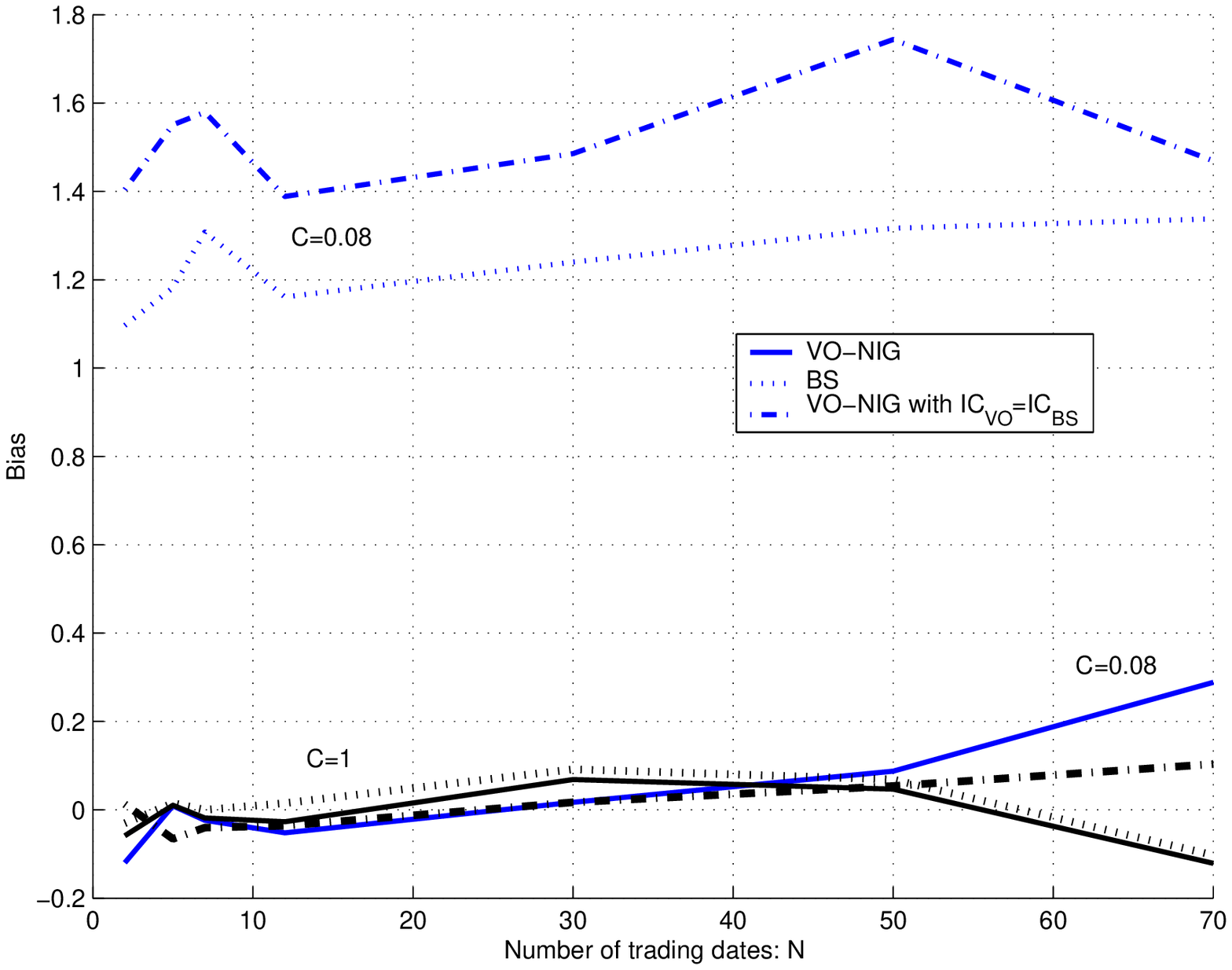,height=5cm, width=7.5cm}
   $\quad$
   \epsfig{figure=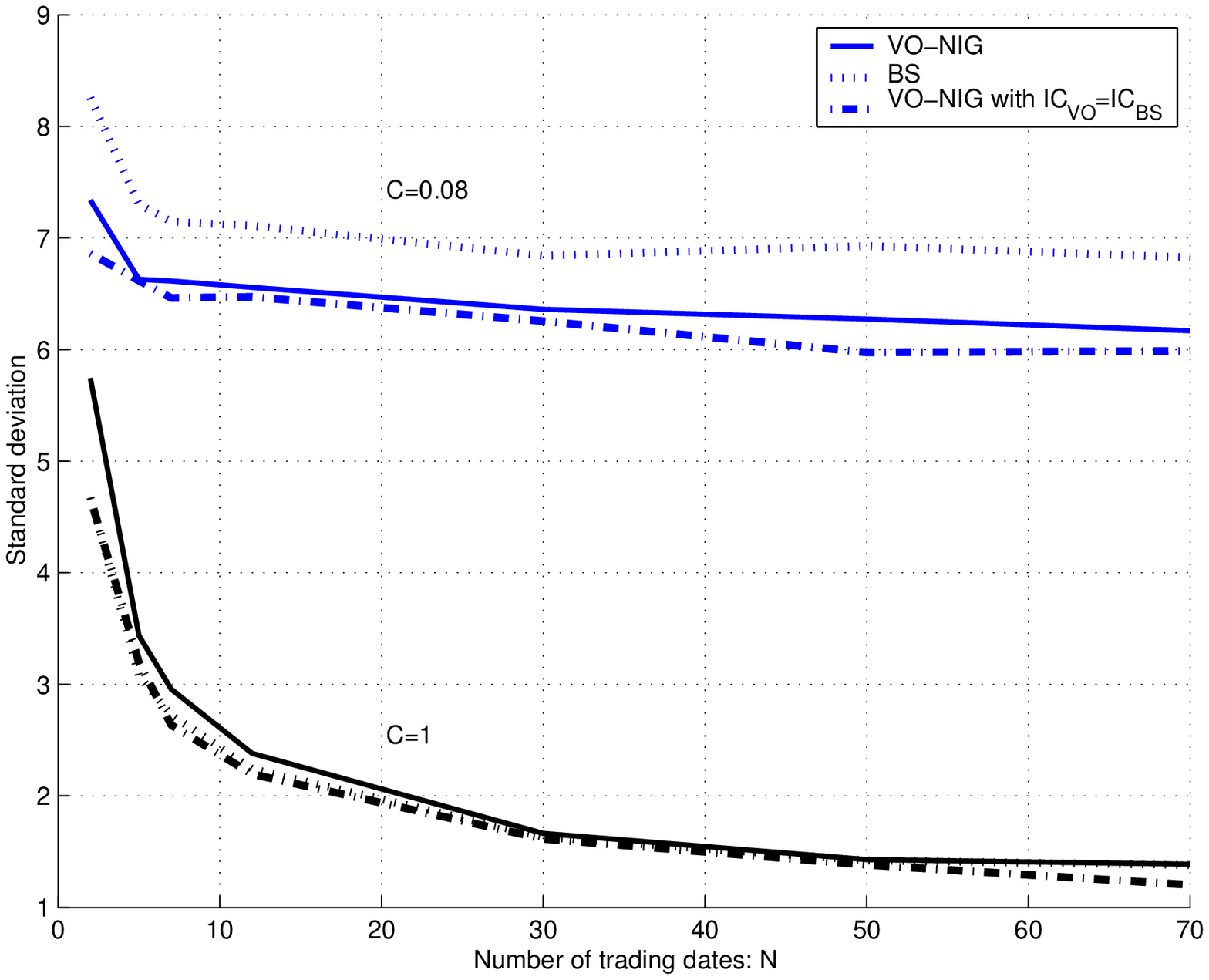,height=5cm, width=7.5cm}
\end{center}
\caption{{\small Hedging Error of BS strategy v.s. the VO strategy with the same initial capital as BS w.r.t. the number of trading dates for different values of $C$ and for $K=99$ Euros (Bias, on the left and standard deviation, on the right).}}
\label{fig:trading:datesINI}
\end{figure}

\subsection{Exponential PII}
We consider the problem of hedging and pricing a European call on an electricity forward, with a maturity $T=0.25$ of three month. The  maturity is equal to the delivery date of the forward contract $T=T_d$. As stated in Section~\ref{sec:elec},  the natural hedging instrument is the corresponding forward contract with value $S^0_t=e^{-r(T-t)}(F_t^{T}-F_0^T)$ for all $t\in [0, T]$, where $F^T=F$ is supposed to follow the  NIG one factor model: 
$$
F_t=e^{X_t}\ ,\quad \textrm{where $X_t=\int_0^t \sigma_s e^{-\lambda(T-u)}d\Lambda_u$} \quad  \textrm{where $\Lambda$ is a NIG process with  }\quad \Lambda_1\,\sim\, NIG(\alpha,\beta,\delta,\mu)\ .
$$
The standard set of parameters $(C=1)$ for the distribution of $\Lambda_1$ is  estimated on the same data as in the previous section (\textit{Month-ahead} \textit{base} forward prices of the French Power market in 2007):  
$$\alpha=15.81 \,,\  \beta= -1.581\,,\ \delta = 15.57\,,\  \mu= 1.56 \ .$$
Those parameters correspond to  a standard and centered NIG distribution with a skewness of $-0.019$.
The estimated annual  short-term volatility and mean-reverting rate are $\sigma_s=57.47\%$ and $\lambda=3$. 
The other sets of parameters considered in simulations are obtained by multiplying parameter  $\alpha$ by a coefficient $C$, ($\beta,\delta,\mu$ being such that the first three moments are unchanged). Table~\ref{tab:kurtosis} shows how the excess kurtosis is modified with $C$. 
\begin{figure}[htbp]
\begin{center}
\begin{tabular}{|c||c|c|}
\hline
 \textrm{Coefficient} & $C=0.08$ & $C=1$\\
\hline
\hline
 $\alpha$&  $1.26$&  $15.81$\\
\hline
 Excess kurtosis&  $1.87$&  $0.013$\\
\hline
\end{tabular}
\end{center}
\caption{{\small Excess kurtosis of $\Lambda_1$ for different values of $\alpha$  $(\beta, \delta, \mu)$ insuring the same three first moments}}
\label{tab:kurtosisb}
\end{figure}
%
%
%

%

Figure~\ref{fig:trading:datesPAI} shows the Bias and Standard deviation of the hedging error as a function of the number of trading dates estimated by Monte Calo method on $5000$ runs. The results are comparable to those obtained in the case of the Lévy process, on Figure~\ref{fig:trading:datesPAI}. However, one can notice that the BS strategy does no more outperform the VO strategy for small numbers of trading dates as observed in the Lévy case. This is due to the fact that $X_t$ is no more a sum of i.i.d. variables.  

%
%
\begin{figure}[htbp]
\begin{center}
   \epsfig{figure=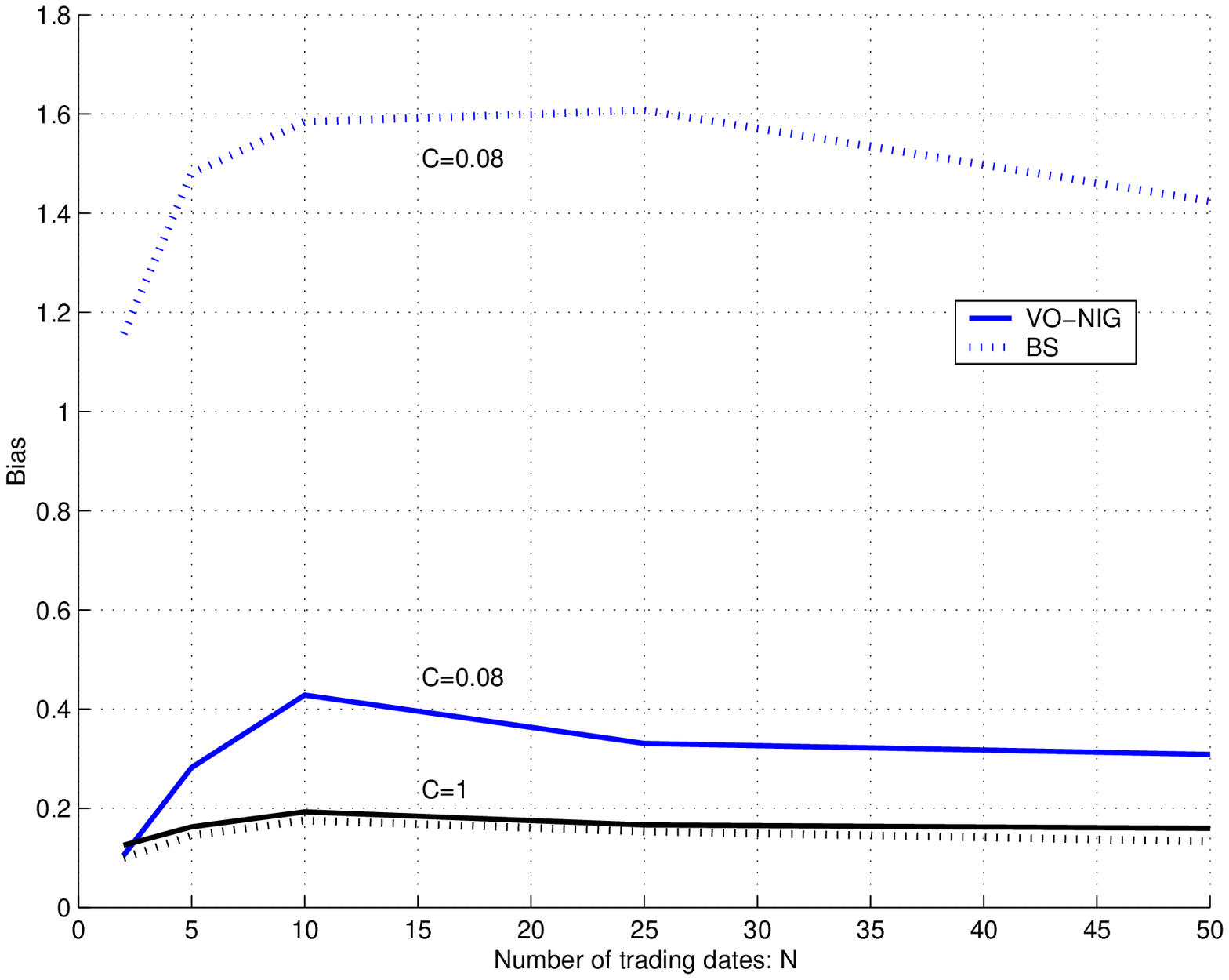,height=5cm, width=7.5cm}
   $\quad$
   \epsfig{figure=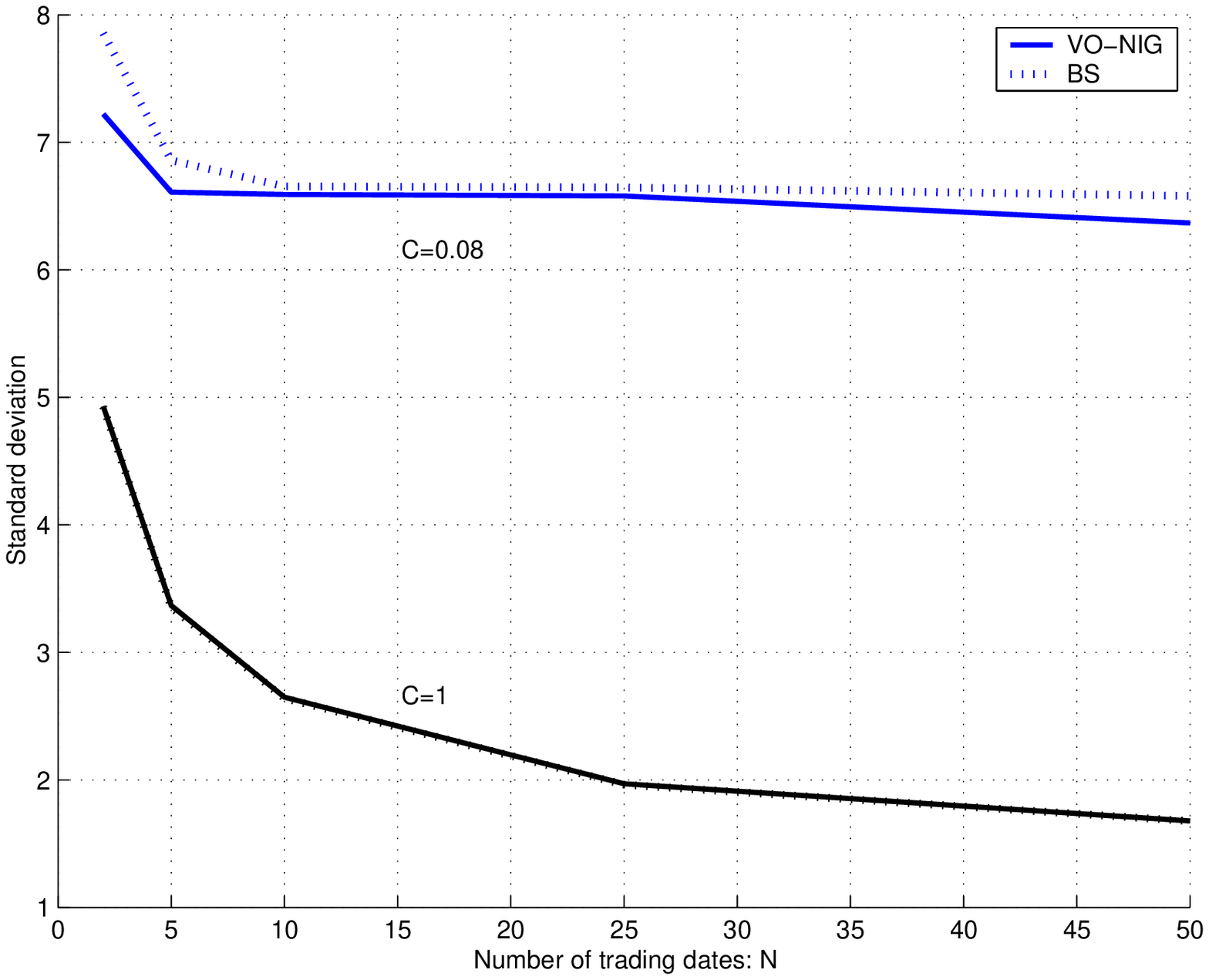,height=5cm, width=7.5cm}
\end{center}
\caption{{\small Hedging Error w.r.t. the number of trading dates for $C=0.08$ and $C=1$, for $K=99$ Euros (Bias, on the left 	and standard deviation, on the right).}}
\label{fig:trading:datesPAI}
\end{figure}
\begin{figure}[htbp]
\begin{center}
\begin{tabular}{|c||c|c|c|c|c|c|c|}
\hline
 \textrm{Moments} & Mean & Standard deviation & Skewness & Kurtosis  \\
\hline
\hline
 VO& $0.43$&  $6.59$&  $-2.89$&  $16.24$ \\
\hline
 BS& $1.58$& $6.65$&  $-3.79 $&  $25.53$ \\
\hline
\end{tabular}
\end{center}
\caption{{\small Empirical moments of the hedging error for $C=0.08$, $N=10$ and $K=99$ Euros. }}
\label{tab:hedging:error:PAI}
\end{figure}

\newpage

\bigskip
{\bf ACKNOWLEDGEMENTS:} The first named author was partially founded
by Banca Intesa San Paolo. 
The second named author was supported by FiME, Laboratoire de Finance des Marchés de l'Energie (Dauphine, CREST, EDF R\&D) www.fime-lab.org. \\
All the authors are grateful to F. Hubalek for useful advices to improve the numerical computations of Laplace transforms performed in simulations.   

\newpage


\end{document}